\documentclass[11pt]{amsart}
\usepackage{amsmath,mathtools}
\usepackage{amsmath, amsthm, amssymb, amsfonts, enumerate}
\usepackage[colorlinks=true,linkcolor=blue,urlcolor=blue]{hyperref}
\usepackage{dsfont}
\usepackage{color}
\usepackage{geometry}
\usepackage{todonotes}
\usepackage{epstopdf}
\usepackage{bbm}
\usepackage{geometry}
\usepackage[utf8]{inputenc}

\usepackage[format=plain, indention=1cm]{caption}

\usepackage{amsfonts}
\usepackage{amsfonts}
\usepackage{textcomp}
\usepackage{amssymb}
\usepackage{float}
\usepackage{tikz}
\usepackage{epsfig}
\usepackage{amsmath}
\usepackage[english]{babel}
\usepackage{a4}
\usepackage{enumerate}
\usepackage{amsmath}
\newcommand{\stkout}[1]{\ifmmode\text{\sout{\ensuremath{#1}}}\else\sout{#1}\fi}

\newtheorem{theorem}{Theorem}[section]
\newtheorem{Remark}[theorem]{Remark}

\newtheorem{assumption}[theorem]{Assumption}
\newtheorem{lemma}[theorem]{Lemma}
\newtheorem{proposition}[theorem]{Proposition}
\newtheorem{corollary}[theorem]{Corollary}

\newcommand{\RomanNumeralCaps}[1]
    {\MakeUppercase{\romannumeral #1}}
\newcommand{\one}{\text{$\mathbbm{1}$}}

\definecolor{red}{rgb}{1.0,0.0,0.0}

\definecolor{blu}{rgb}{0.0,0.0,1.0}

\definecolor{gre}{rgb}{0.03,0.50,0.03}

\renewcommand{\hat}{\widehat}
\newcommand{\ttilde}{\widetilde}

\title[Optimal Execution under Partial Observation]{Optimal Execution with Multiplicative Price Impact and Incomplete Information on the Return}
\author[Dammann]{Felix Dammann}
\author[Ferrari]{Giorgio Ferrari}

\address{F.~Dammann:  Center for Mathematical Economics (IMW), Bielefeld University, Universit\"atsstrasse 25, 33615, Bielefeld, Germany}
\email{\href{mailto:dammann@uni-bielefeld.de}{dammann@uni-bielefeld.de}}
\address{G.~Ferrari: Center for Mathematical Economics (IMW), Bielefeld University, Universit\"atsstrasse 25, 33615, Bielefeld, Germany}
\email{\href{mailto:giorgio.ferrari@uni-bielefeld.de}{giorgio.ferrari@uni-bielefeld.de}}

\date{\today}

\numberwithin{equation}{section}

\begin{document}

\begin{abstract}
We study an optimal liquidation problem with multiplicative price impact in which the trend of the asset's price is an unobservable Bernoulli random variable. The investor aims at selling over an infinite time-horizon a fixed amount of assets in order to maximize a net expected profit functional, and lump-sum as well as singularly continuous actions are allowed. Our mathematical modelling leads to a singular stochastic control problem featuring a finite-fuel constraint and partial observation. We provide the complete analysis of an equivalent three-dimensional degenerate problem under full information, whose state process is composed of the asset's price dynamics, the amount of available assets in the portfolio, and the investor's belief about the true value of the asset's trend. The optimal execution rule and the problem's value function are expressed in terms of the solution to a truly two-dimensional optimal stopping problem, whose associated belief-dependent free boundary $b$ triggers the investor's optimal selling rule. The curve $b$ is uniquely determined through a nonlinear integral equation, for which we derive a numerical solution through an application of the Monte-Carlo method. This allows us to understand the sensitivity of the problem's solution with respect to the relevant model's parameters as well as the value of information in our model. 
\end{abstract}
\maketitle

\smallskip


\smallskip

{\textbf{MSC2020 subject classification}}: 93E20, 93C41, 49L20, 91G80

\smallskip


\section{Introduction}
\label{sec:intro}

In this paper, we consider an investor who possesses a fixed amount of assets and aims at selling them on the market. We assume that the investor faces the issue of causing an adverse price reaction, so that fast selling depresses the stock price, while splitting the order over time may take too long. This problem -- also known as the \textit{optimal execution problem} in algorithmic trading -- thus deals with the question of how to trade optimally in order to maximize a given profit, and therefore of how to determine the time as well as the size of the order.  

Dating back to the early works of Bertsimas and Lo \cite{BeLo}, Almgren and Chriss \cite{AlCh} and Almgren \cite{Al}, the study of optimal execution strategies has received much attention and resulted in a series of important contributions in various settings, which, amongst other modeling features, can be distinguished with respect to the considered type of price impact: Additive or multiplicative. A comprehensive discussion on the latter class of models can be found in Guo and Zervos \cite{GuZe}, who also point out that models with multiplicative price impact seem to be more natural since they ensure prices to remain positive. Amongst those works dealing with multiplicative price impact, let us mention Bertsimas et al.\ \cite{BeLoHu} for a discrete-time framework, Forsyth et al.~\cite{For} for a continuous-time model  \`a la Black-Scholes, Guo and Zervos \cite{GuZe} and Becherer et al.~\cite{BeBiFr} for settings involving singular stochastic controls.

A common feature in the literature is the assumption that the investor has full information on the trend of the asset. This, however, can be a strong requirement. As pointed out by Ekstr\"{o}m and Lu \cite{EkLu}, a statistical estimation of the drift is not an efficient procedure, and obtaining a reasonable precision would need data of decades or even centuries under the same market conditions -- which is simply not feasible in reality (see also the discussion in Rogers \cite{Ro}, Section 4.2). In some cases, such as initial public offerings, this price history does not even exist. 

To account for this fact, we propose a model of optimal execution with multiplicative price impact in which the drift of the stock price dynamics is a random variable, which is not directly observable by the investor. Through monitoring the evolution of the price on the market, the investor is able to update her belief regarding the drift value. However, such observation is noisy as the investor cannot perfectly distinguish whether price variations are caused by the drift or the stochastic driver of the underlying dynamics. From a mathematical point of view, our model leads to a finite-fuel singular stochastic control problem  under partial observation, and we investigate how the presence of incomplete information influences the selling strategy of the investor. In particular, we show that the flow of incoming information -- through the observation of the asset's market price -- has a direct effect on the optimal execution rule. Indeed, differently to the case of full information treated in Guo and Zervos \cite{GuZe}, the decision to sell is no longer triggered by a constant critical price, but the execution threshold changes dynamically depending on the investor's current belief on the future trend of the asset. Our results show that the optimal execution strategy is in fact determined by a boundary that is increasing in the belief towards the larger drift value, underlying the intuition that the decision maker chooses to delay selling assets if future prices are expected to increase. 

In this regard, our work relates to the bunch of economic and financial literature where questions of optimal decision-making under partial observation have been considered; amongst a large number of contributions, we refer to the seminal papers on portfolio selection by Detemple \cite{Detemple} and Gennotte \cite{Gennotte}; to Veronesi \cite{Veronesi} for an equilibrium model with uncertain dividend drift; to Sass and Haussmann \cite{SassHaussmann} for a terminal-wealth portfolio optimization problem, and to the more recent Colaneri et al.~\cite{Colaneri-etal} for an optimal liquidation problem with rate strategies and partial observation. Notably, the recent Drissi \cite{Drissi} and Bismuth et al.~\cite{Bismuth} incorporate Bayesian learning in a model of multi-asset optimal execution, although restricting the agent to absolutely continuous (regular) controls. 

Furthermore, we contribute to those models dealing with problems of optimal stopping and singular stochastic control. To name just a few recent works, Callegaro et al.~\cite{CaCe} for public debt control, De Angelis \cite{DeA} and D\'ecamps and Villeneuve \cite{DeVi} for dividend payments, D\'ecamps et al.~\cite{Deca} for investment timing, Ekstr\"{o}m and Lu \cite{EkLu} as well as Ekstr\"{o}m and Vaicenavicius \cite{EkVa} for asset liquidation, Federico et al.~\cite{FeFe} for inventory management, Johnson and Peskir \cite{JoPe} for quickest detection, Gapeev \cite{Gapeev} for the pricing problem of perpetual commodity equities, and Gapeev and Rodosthenous \cite{GaRo} for a zero-sum optimal stopping game associated with perpetual convertible bonds. \\[0.3cm]
\textbf{Our model, approach and overview of the mathematical analysis.} We now discuss the mathematical modeling and analysis. We consider a risk-neutral investor holding a fixed amount $y$ of assets in her portfolio. In absence of the investor's actions, the stock price evolves according to a geometric Brownian motion $dS_t = \beta S_t dt + \sigma S_t d W_t$, where $W$ is a standard Brownian motion and $\sigma > 0$ a constant volatility parameter. Furthermore, the price process exhibits a random future trend $\beta$, which is however unknown to the decision maker, and is assumed to be a random variable, independent of the Brownian noise, taking two values $\beta_0 < \beta_1$, for some $\beta_0,\beta_1 \in \mathbb{R}$ and $\beta_0 <0$. 

The decision maker is able to sell the assets on the market over an infinite time horizon, and we denote by $\xi_t$ the cumulative amount of assets liquidated up to time $t$. Consequently, the remaining assets in the portfolio follow the dynamics $Y_t^\xi = y - \xi_t$. Clearly, it has to be $\xi_t \leq  y$ at any time $t\geq0$ (finite-fuel constraint), since no more than the initial amount of assets can be sold. As anticipated, we assume that the investor causes an adverse price reaction upon selling, which, following Guo and Zervos \cite{GuZe}, we assume to be of multiplicative type. Hence, the controlled asset's price evolves as 
\begin{align*}
dS_t^\xi = \beta S_t^\xi dt + \sigma S_t^\xi d W_t - \alpha S_t^\xi \circ d \xi_t, \qquad S_{0-}^\xi = s > 0,
\end{align*}
where $\alpha > 0$ denotes the parameter of price impact, and the operator $\circ$ is defined as in \eqref{St circ xi} below so to take care of the continuous and jump components of any admissible selling strategy $\xi$. Notice that the multiplicative price impact structure allows to express the asset's price process as $S^\xi = \exp (X^\xi)$. Here, $X^\xi$ is then a linearly controlled drifted Brownian motion with volatility $\sigma>0$ and drift value $\mu = \beta - \frac12 \sigma^2$. 

The investor aims at maximizing the total expected discounted reward upon selling, net of transaction costs; that is,
$$\sup_{\xi} \mathbb{E} \Big[ \int_0^\infty e^{-rt} \big(e^{X^\xi_t} - \kappa \big) \circ d \xi_t \Big],$$
where the optimization is taken over a suitable admissible class of selling strategies and the investor discounts her future revenues with a strictly positive factor $r>0$, that can be interpreted as her subjective impatience. 
The latter is a finite-fuel singular stochastic control problem under partial observation.

By relying on classical filtering techniques (cf.\ Shiryaev \cite{Shi}, Section 4.2), we begin by determining an equivalent Markovian problem -- the so-called \textit{separated problem} -- under full information (see Fleming and Pardoux \cite{FlPa} as a classical reference on the separated problem). To this end, we introduce the process $\Pi$, according to which the investor can update her belief regarding the true value of the drift. This is done by observing the evolution of the process $X^0$ (denoting the uncontrolled version of the process $X^\xi$), whose natural filtration $\mathcal{F}_t^{X^0}$ models the overall information available up to time $t$. More precisely, after forming a prior $\pi := \mathbb{P}[\mu = \mu_1] \in (0,1)$, the investor dynamically updates her belief upon the arrival of new information through observing the process $X^0$, so that the belief process is given by $\Pi_t = \mathbb{P}[\mu = \mu_1 \mid \mathcal{F}_t^{X^0}]$. Notice that a value of $\Pi$ close to $1$ indicates a strong belief towards the larger value of the drift, while $\Pi$ close to $0$ displays a strong belief in the lower value. Hence, we expect the investor to change the liquidation strategy dynamically and not solely base it on the current price on the market, but also on the present belief at that time. 

The separated problem turns out to be a \emph{three-dimensional degenerate finite-fuel singular stochastic control problem}, so that obtaining explicit solutions through a traditional ``guess-and-verify approach'' is in general not feasible.\footnote{A ``guess-and-verify approach'' is applicable if we take $\beta_0 = - \beta_1$, which indeed allows for a dimension reduction; see, e.g., D\'ecamps and Villeneuve \cite{DeVi}. In this paper, however, we do not consider any relation amongst $\beta_0 $ and $\beta_1$ other than $\beta_0 < \beta_1$.}

In order to tame the multidimensional nature of the resulting optimal execution problem under full information, we then follow a direct approach which hinges on the study of a suitable optimal stopping problem with value $v$, that we expect to be associated to the singular stochastic control problem. This method was studied and refined by many authors such as Bene\v{s} et al.~\cite{Ben}, El Karoui and Karatzas \cite{ElKa}, and Karatzas and Shreve \cite{KaSh}, or De Angelis \cite{DeA}, De Angelis et al.~\cite{GF} and \cite{DeAFeMo}, and Guo and Tomecek \cite{GuoTo} for more recent contributions. 
The optimal stopping problem, which involves the underlying two-dimensional diffusion $(X^0,\Pi)$ taking values in $\mathbb{R} \times (0,1)$, can be interpreted as an optimal selling problem and exhibits a structure similar to that of the problem treated by D\'ecamps et al.\ \cite{Deca} (see also Ekstr\"{o}m and Lu \cite{EkLu} for a parabolic version).
 We then solve the optimal stopping problem by relying on techniques from free-boundary theory (as illustrated in the monography by Peskir and Shiryaev \cite{PeSh}) and first show that the optimal stopping rule is characterized through a belief-dependent free boundary $a(\pi)$ for $\pi \in (0,1)$. 

However, the coupled dynamics of the underlying processes $X^0$ and $\Pi$, as well as the fact that they are driven by the same Brownian motion, makes a further study of the free boundary and the value function $v$ not feasible. It is for that reason we proceed by deriving two equivalent representations of the optimal stopping problem, which allow for a thorough analysis. First, via a change of measure, the state process $(X^0,\Pi)$ is transformed into $(X^0,\Phi)$ taking values in $\mathbb{R}\times (0,\infty)$ and with decoupled dynamics. Here, the process $\Phi$ is the so-called ``likelihood ratio''. Again, we can express the optimal stopping strategy in terms of a free boundary $\varphi \mapsto b(\varphi)$, which results from a simple transformation of the boundary $\pi \mapsto a(\pi)$. Second, we pass yet to another formulation by deriving the intrinsic parabolic formulation of the stopping problem in coordinates $(X^0,Z)$, in which the process $Z$ now follows purely deterministic dynamics and takes values in $\mathbb{R}$. Even though the monotonicity result of the associated free boundary $z \mapsto c(z)$ is certainly not trivial to derive and calls for a rigorous technical analysis, it is in this formulation that we are able to provide further regularity results of $c$ and of the transformed optimal stopping value function $\widehat{v}$. In fact, borrowing arguments from De Angelis \cite{DeA}, suitably adapted to the present setting, we achieve a global regularity of $\widehat{v}$, namely $\widehat{v} \in C^1 (\mathbb{R}^2)$. The latter result also allows proving $\widehat{v}_{xx} \in L_{\text{loc}}^\infty (\mathbb{R}^2)$, and finally obtaining a nonlinear integral equation uniquely solved by the optimal stopping boundary $c$. It is worth mentioning that such a characterization can be traced back to both optimal stopping boundaries $b$ and $a$ and is thus tantamount to a complete specification of the optimal stopping rule in the original $(x,\pi)$-coordinates. 

The thorough analysis developed for the optimal stopping problem is then exploited in order to identify an optimal execution strategy. In fact, the derived regularity results for $\hat{v}$ permits us to prove a verification theorem, that identifies an optimal execution rule and shows that the optimal stopping value function $v$ indeed coincides with a directional derivative of the separated problem's value function $V$. Namely, we show that
$$V(x,y,\pi):=\frac{1}{\alpha}\int_{x-\alpha y}^x v(x',\pi) dx', \qquad (x,y,\pi) \in \mathbb{R} \times (0,\infty) \times (0,1).$$
Notice, that if $\alpha \downarrow 0$, one finds $V(x,y,\pi)=y v(x,\pi)$, which is the value of the problem in which the investor has no market impact.

The optimal execution rule can be thought of as a ``myopic one''. Indeed, it prescribes to sell assets as if the size of the investor's portfolio were infinite, and to stop selling once the asset's inventory is depleted (see also Karatzas \cite{Ka85} and El Karoui and Karatzas \cite{ElKa}). The optimal selling rule involves lump-sum executions (whenever the asset's price is sufficiently large), that could eventually result into an immediate depletion of the portfolio (if the initial portfolio size is sufficiently small). However, for relatively large portfolios, an initial lump-sum selling is followed by a policy of oblique reflection type. This is triggered by the belief-dependent boundary $\varphi \mapsto b(\varphi)$ (equivalently, $\pi \mapsto a(\pi)$). Notably, given that all the transformations developed for the resolution of the optimal stopping problem are one-to-one and onto, the integral equation for the boundary $z \mapsto c(z)$ yields an integral equation for $\varphi \mapsto b(\varphi)$, and therefore a complete characterization of the optimal execution rule. In order to provide insights about the sensitivity of the optimal decision mechanism of the investor with respect to the model's parameters, we develop a recursive numerical scheme, which relies on an application of the Monte-Carlo method.\\[0.3cm]
\textbf{Our contributions.} Overall, we believe that the contributions of this paper are the following. Even though the literature on optimal execution problems is extensive (see, to name just a few, Almgren and Chriss \cite{AlCh}, Almgren \cite{Al}, Becherer et al.~\cite{BeBiFr}, Bertsimas and Lo \cite{BeLo}, Bertsimas et al.\ \cite{BeLoHu}, Colaneri et al.\ \cite{Colaneri-etal}, Gatheral and Schied \cite{GaSc}, Guo and Zervos \cite{GuZe}, Moreau et al.~\cite{MoMuMe}, Schied and Sch\"oneborn \cite{Schied}), the combination of incomplete information on the future price trend while allowing for lump-sum as well as singularly continuous executions constitutes a novelty. Furthermore, the present study on the optimal execution strategy complements as well as extends the literature on problems with a similar structure under full information. As a matter of fact, the derived optimal execution rule exhibits a broader structure and prescribes to take actions depending on the current belief on the future trend of the asset. 

From a mathematical point of view, to the best of our knowledge, ours is the first work providing a complete characterization of the value function and of the optimal control rule in a finite-fuel singular stochastic control problem under partial observation (which, in the present setting, is equivalent to a three-dimensional degenerate singular stochastic control problem). Furthermore, we believe that the optimal stopping (selling) problem, studied as a device to characterize the optimal solution of the optimal execution problem, is of interest of its own. By performing a thorough analysis on the regularity of (a transformed version of) its value function and free boundary, we are able to provide a complete characterization of the optimal selling rule through a nonlinear integral equation, thus extending the results of the related model studied by D\'ecamps et al.~\cite{Deca}. Notice, that an integral equation for the free boundary has been obtained also in Ekstr\"{o}m and Lu \cite{EkLu} and Ekstr\"{o}m and Vaicenavicius \cite{EkVa}, though in settings where the parabolic nature of the problem is arising because of an explicit time-dependency. Finally, the probabilistic numerical approach developed for the resolution of the free boundary's integral equation allows to understand the dependency of the investor's optimal execution strategy on relevant model's parameters such as volatility and trend. Moreover, based on the numerical evaluation of the boundary, we can compare the value of the control problem with partial information with that of an associated \textit{average drift} problem under full information. This allows us to numerically  evaluate the question on whether the introduction of uncertainty over the drift actually harms or benefits the investor. \\[0.3cm]
\textbf{Organization of the paper.} The rest of the paper is organized as follows. In Section \ref{Section: Problem Formulation} we present our setting and first preliminary results. In Section \ref{Section: Benchmark Problem} we investigate the benchmark problem under full information, before we consider a corresponding optimal stopping problem and its optimal boundary in Section \ref{Section: First related OSP}. In Section \ref{Section: Decoupling Change of Measure} and \ref{Section: Parabolic Formulation} we derive two equivalent formulations of this problem, which allow for a more thorough study. Eventually, in Section \ref{Section: Solution to Execution Problem}, we return to the optimal control problem and characterize the optimal selling rule of the investor. A numerical study based on the derived integral equation of the execution boundary in then carried out in Section \ref{Section: Comparative Statics Analysis}.  


\section{Setting and Problem Formulation}\label{Section: Problem Formulation} 

Let $(\Omega, \mathcal{F}, \mathbb{P})$ be a complete probability space, rich enough to accommodate a standard one-dimensional Brownian motion $(W_t)_{t \geq 0}$ and  an independent random variable $\beta$ taking two values $\beta_0$ and $\beta_1$. We denote by $\mathbb{F}^W := (\mathcal{F}_t^W )_{t \geq 0}$ the filtration generated by $(W_t)_{t \geq 0}$ augmented by $\mathbb{P}$-null sets of $\mathcal{F}_0^W$. We assume that, in absence of any actions of the investor, the asset's price on the stock market evolves stochastically according to a geometric Brownian motion 
\begin{align}\label{Dyn S^0}
d S_t^0 = \beta S_t^0 dt + \sigma S_t^0 d W_t, \quad S_0^0 =s > 0,
\end{align}
where $\sigma >0$ is a constant volatility. The investor holds a finite amount $y \geq 0$ of assets, which she is able to sell. We identify the cumulative amount of assets sold up to time $t \geq 0$, which we denote by $\xi_t$, as the investor's control variable. We denote by $\mathbb{F}^Z := (\mathcal{F}_t^Z)_{t \geq 0}$ the natural filtration of any process $Z$, augmented by $\mathbb{P}$-null sets of $\mathcal{F}_0^Z$, and hence, the set of admissible execution strategies in this context is given by 
\begin{align*}
\mathcal{A}(y) := \left\{ \xi : \Omega \times [0,\infty) \to \mathbb{R}_+: ~~ (\xi_t)_{t \geq 0}~ \mathbb{F}^{S^0}\text{-adapted, increasing, c\`{a}dl\`{a}g, and }\xi_{0 -} =0,~\xi_t \leq y~\text{a.s.} \right\},
\end{align*}
where the last condition naturally arises from the fact that the investor cannot sell more than the initial amount of assets. Moreover, the remaining assets in the portfolio evolve according to the dynamics
\begin{align*}
 Y_t^{ \xi} = y - \xi_t , \quad Y_{0 -}^{\xi} = y \geq 0,
\end{align*}
where we stress the dependency on the selling strategy $\xi$. Following Guo and Zervos \cite{GuZe}, in our model we assume that the investor's transactions on the market have a proportional impact on the asset's price. More precisely, when selling a small amount $\epsilon > 0$ of assets at time $t$, the price exhibits a jump of size 
\begin{align*}
\Delta S_t = S_{t} - S_{t-} = - \alpha \epsilon S_t,
\end{align*}  
for $\alpha > 0$ denoting the parameter of permanent price impact (see Almgren and Chriss \cite{AlCh},  Almgren \cite{Al} for early works and Becherer et al.~\cite{BeBi}, Ferrari and Koch \cite{FeKo}, Guo and Zervos \cite{GuZe} for more recent contributions). Hence, a small transaction is such that $S_{t} = (1 - \alpha \epsilon ) S_{t-} \simeq e^{- \alpha \epsilon } S_{t-}$ and, by interpreting a lump-sum sale of $\Delta \xi_t$ shares as a sequence of $N$ individual sales of size $\epsilon = \Delta \xi_t / N$, we have 
\begin{align*}
S_{t} = e^{ - \alpha N \epsilon} S_{t-} = e^{- \alpha \Delta \xi_t} S_{t-},
\end{align*}
for $N$ large enough. It follows that, for any $\xi \in \mathcal{A}(y)$, we can model the controlled asset's price process by 
\begin{align}\label{dyn S}
d S_t^{\xi} = \beta S_t^\xi dt + \sigma S_t^\xi dW_t - \alpha S_t^\xi \circ d \xi_t , \quad S_{0 -}^{\xi} = s, 
\end{align}
where
\begin{align}\label{St circ xi}
\int_0^\cdot S_t^\xi \circ d \xi_t := \int_0^\cdot S_t^\xi d \xi_t^c  + \sum_{t \leq \cdot \, : \Delta \xi_t \neq 0} \frac{1}{\alpha} S_{t-}^\xi ( 1 - e^{- \alpha \Delta \xi_t}) = \int_0^\cdot S_t^\xi d \xi_t^c + \sum_{t \leq \cdot \, : \Delta \xi_t \neq 0} S_{t-}^\xi \int_0^{\Delta \xi_t} e^{- \alpha u } d u,
\end{align}
$\xi^c$ denotes the continuous part of the process $\xi$, and $\Delta \xi_t := \xi_t - \xi_{t-}$. The solution to \eqref{dyn S} can be explicitely determined via It\^{o}'s formula and it is given by 
\begin{align}\label{S_t xi explicit}
S_t^\xi = s \exp \Big( (\beta - \frac12 \sigma^2 ) t + \sigma W_t - \alpha \xi_t \Big) = S_t^0 \exp ( - \alpha \xi_t ),
\end{align} 
where $S^0$ is the solution to \eqref{Dyn S^0} and we observe that the price impact of selling is additive to the logarithm of the asset's price.\\
We assume that the investor aims at maximizing the total expected (discounted) profits, net of the total cost of selling, and thus seeks to solve
\begin{align}\label{Objective with S_t}
\sup_{\xi \in \mathcal{A}(y)} \mathbb{E} \Big[ \int_0^\infty e^{-rt} \big( S_t^\xi &- \kappa \big) \circ d \xi_t \Big] \nonumber \\
&= \sup_{\xi \in \mathcal{A}(y)} \mathbb{E}\Big[ \int_0^\infty e^{-rt} \big( S_t^\xi - \kappa \big) d \xi_t^c +  \sum_{t: \Delta \xi_t \neq 0} e^{-rt} \int_0^{\Delta \xi_t} (S_{t-}^\xi e^{ - \alpha u } - \kappa ) du \Big].
\end{align}
Here, $\kappa > 0 $ is a proportional transaction cost, which, thinking of $S_t^\xi$ as the mid-price of the stock at time $t$, can also be interpreted as a constant bid spread. Notice that the structure of the expected net-profit functional in \eqref{Objective with S_t} can also be justified through stability results in the Skorokhod $M_1$-topology in probability (see Becherer et al.\,\cite{BeBiFr}). Moreover, problem \eqref{Objective with S_t} has finite value due to $\xi_t \leq y$ a.s. Thanks to \eqref{S_t xi explicit} we have $S_t^\xi = \exp (X_t^\xi )$, where
\begin{align}\label{Dyn X^xi with mu}
dX_t^\xi = \mu dt + \sigma d W_t - \alpha d \xi_t, \quad X_{0-}^\xi = x,
\end{align}
with $x := \ln (s)$ and $\mu := \beta - \frac12 \sigma^2$. In particular, the drift can take two values $\mu_i = \beta_i - \frac12 \sigma^2,~i=0,1$. In the following, when needed, we let $X^0$ denote the solution to \eqref{Dyn X^xi with mu} with $\xi \equiv 0$, which is then an arithmetic Brownian motion. Furthermore, we state the following assumption.
\begin{assumption}
We have $\beta_1 > \beta_0$ and $\beta_0 < 0$, which implies $\mu_0 < 0$.
\end{assumption}
The maximization problem \eqref{Objective with S_t} thus can be rewritten in terms of \eqref{Dyn X^xi with mu} as
\begin{align}\label{non markov problem}
\sup_{\xi \in \mathcal{A}(y)} \mathbb{E} \Big[ \int_0^\infty e^{-rt} \left( e^{X_t^\xi} - \kappa \right) \circ d \xi_t \Big]. 
\end{align}
Notice that for a constant non-random drift coefficient, a close variant of this problem was considered and solved by Guo and Zervos \cite{GuZe}, who also incorporate the option of buying shares of assets and the constraint that the whole inventory has to be depleted at terminal time. However - due to the presence of incomplete information on the drift of the asset - Problem \eqref{non markov problem} is not of Markovian nature and thus requires a thoroughly different analysis. In order to obtain an equivalent Markovian formulation of \eqref{non markov problem}, we rely on classical results from filtering theory, dating back to the contribution of Shiryaev in the context of quickest detection models (see Shiryaev \cite{Sh} for a survey). To this end, we introduce the \textit{belief} process 
\begin{align*}
\Pi_t := \mathbb{P} \left[\mu = \mu_1 \mid \mathcal{F}_t^{X^0} \right], \quad t \geq 0,
\end{align*}
which reflects the probability at time $t$ that $\mu = \mu_1$, conditional on the observations of the price process up to that time (indeed, $\mathbb{F}^{S^0} = \mathbb{F}^{X^0} = \mathbb{F}^{X^\xi} $). According to this process, the investor is able to update the belief regarding the true value of the drift, based on the arrival of new information by observing the asset's price evolution on the market. Notice that a large value of $\Pi$ close to $1$ implies a strong belief towards the larger drift value $\mu_1$, while a low value of $\Pi$ implies the contrary. It follows (see, e.g., Shiryaev \cite{Shi}, Section 4.2) that the dynamics of $X^{\xi}, \Pi$ and $Y^{\xi}$ can be written as 
\begin{align}\label{dynamics control problem}
\begin{cases}
d X_t^{\xi} = (\mu_1 \Pi_t + \mu_0 (1 - \Pi_t)) dt + \sigma d \overline{W}_t - \alpha d \xi_t , & X_{0-}^{\xi} = x \in \mathbb{R}, \\
d \Pi_t = \gamma \Pi_t (1- \Pi_t ) d \overline{W}_t, & \Pi_0 = \pi \in (0,1), \\
Y_t^{\xi} = y - \xi_t , & Y_{0 -}^{\xi} = y \geq 0,
\end{cases}
\end{align}
where $\gamma = (\mu_1 - \mu_0)/\sigma$ is the \emph{signal-to-noise ratio} and 
\begin{align*}
d \overline{W}_t = \frac{d X_t^0}{\sigma} - \Big( \frac{\mu_0}{\sigma} + \gamma \Pi_t \Big) dt 
\end{align*}
denotes the \textit{innovation process}, which is an $\mathbb{F}^{X^0}$-Brownian motion on $(\Omega, \mathcal{F}, \mathbb{P})$. Moreover, $\pi := \mathbb{P}[\mu = \mu_1]$ reflects the initial subjective belief of the investor regarding the true value of the drift. We do not question the origin of this initial belief, this can either be an instinctive decision or even the result of a constructive approach, for instance by observing the trends of similar assets over the past years. In the new formulation, the process $(X^\xi , Y^\xi ,\Pi)$ is an $\mathbb{F}^{X^0}$-adapted and time-homogeneous Markov process, as it is the unique and strong solution to the system of stochastic differential equations in \eqref{dynamics control problem}. Furthermore, we observe that the drift $\mu$ is replaced by its conditional estimate and the process $\Pi$ is a bounded martingale on $[0,1]$ with $\Pi_\infty \in \{0,1 \}$, as all information will eventually get revealed. Denoting $\mathbb{E}_{(x,y,\pi)} [ \cdot ] = \mathbb{E} [ \cdot \vert X_{0-}^\xi = x, Y_{0-}^\xi = y , \Pi_0 = \pi]$, we can thus reformulate the problem of incomplete information as a so-called \emph{separated problem} (cf. Bensoussan \cite{Bens}, Chapter 7.1 and Fleming and Pardoux \cite{FlPa})
\begin{align}\label{Value function control problem }
V(x,y, \pi) := \sup_{\xi \in \mathcal{A}(y)} J(x,y,\pi, \xi),
\end{align}
with
\begin{align}\label{objective control problem}
J(x,y,\pi ,\xi) :=  \mathbb{E}_{(x,y,\pi)} \Big[ \int_0^\infty e^{-rt} \Big( e^{X_t^\xi} - \kappa \Big) d \xi_t^c +  \sum_{t: \Delta \xi_t \neq 0} e^{-rt} \int_0^{\Delta \xi_t} (e^{X^\xi_{t-} - \alpha u } - \kappa ) du \Big],
\end{align}
for any $(x,y,\pi) \in \mathbb{R} \times (0, \infty) \times (0,1)$. Notice indeed that $\Pi_t \in (0,1)$ for all $t \geq 0$ a.s.~if $\pi \in (0,1)$, while $\Pi_t \equiv \pi_0$ for all $t \geq 0$ a.s.~if $\pi_0 \in \{ 0,1\}$. Problem \eqref{Value function control problem } is equivalent to \eqref{Objective with S_t}: They share the same value and, because of the uniqueness of the strong solution to \eqref{dynamics control problem}, a control is optimal for \eqref{Objective with S_t} if and only if it is optimal for \eqref{Value function control problem }. \\[0.3cm]
\textbf{The Hamilton-Jacobi-Bellman equation.} Problem \eqref{Value function control problem } takes the form of a \emph{three-dimensional singular stochastic control problem with finite-fuel constraint} (cf.~Baldursson \cite{Ba}, Bene\v{s} et al.~\cite{Ben}, El Karoui and Karatzas \cite{ElKa}, Karatzas \cite{Kara} and Karatzas et al.\ \cite{Ka-etal} for early contributions). We start our analysis by providing a heuristic derivation of the dynamic programming equation, that we expect the value function $V$ to satisfy. To this end, we notice that the investor is faced with two possible actions at initial time. On the one hand, the investor could choose to wait for a short period of time $\Delta t$, not sell any fraction of the assets and then continue with an optimal execution strategy (supposing that one exists). Since this strategy is not necessarily optimal, we obtain 
\begin{align*}
V(x,y, \pi) \geq \mathbb{E}_{(x,y,\pi)} \left[ e^{-r \Delta t} V(X_{\Delta t} ,y , \Pi_{\Delta t} ) \right], \qquad (x,y,\pi) \in \mathbb{R} \times (0, \infty) \times (0,1). 
\end{align*}
If we assume that the value function $V$ has enough regularity, we can apply It\^{o}'s formula, divide by  $\Delta t$ and invoke the mean value theorem in order to let $t \to 0$, so to obtain 
\begin{align*}
(\mathcal{L}_{X,\Pi} - r ) V \leq 0.
\end{align*}
Here, $\mathcal{L}_{X,\Pi}$ denotes the second-order differential operator, acting on twice-continuously differentiable functions, 
\begin{align}\label{Operator LXPi}
\mathcal{L}_{X,\Pi} := \frac12 \gamma^2 \pi^2 (1- \pi)^2 \partial_{\pi \pi} + \frac12 \sigma^2 \partial_{x x} + (\pi \mu_1 + (1- \pi) \mu_0 ) \partial_x + \sigma \gamma \pi (1 -\pi) \partial_{x \pi}.
\end{align}
On the other hand, the investor can instantaneously sell an amount $\epsilon > 0$ of the assets and then proceed by following an optimal execution strategy. Again, this strategy is a priori suboptimal and, since this action is associated with the inequality 
\begin{align*}
V(x,y,\pi) \geq V(x - \alpha \epsilon, y - \epsilon, \pi) + \frac{1}{\alpha} e^x \left( 1 - e^{- \alpha \epsilon} \right) - \kappa \epsilon ,
\end{align*}
adding and subtracting $V(x - \alpha \epsilon , y , \pi)$, and dividing by $\epsilon$, yields
\begin{align*}
\frac{V(x,y, \pi) - V(x - \alpha \epsilon , y , \pi)}{\epsilon} \geq \frac{V(x- \alpha \epsilon, y - \epsilon ,\pi) - V(x - \alpha \epsilon , y , \pi)}{	\epsilon} + \frac{1}{\alpha} e^x \frac{\left(1 - e^{- \alpha\epsilon} \right)}{\epsilon} - \kappa.
\end{align*}
Hence, by letting $\epsilon \downarrow 0$, we obtain 
\begin{align*}
\alpha V_x (x ,y , \pi)\geq - V_y (x,y,\pi) + e^x - \kappa.
\end{align*}
Since only one of these actions should be optimal, and given the Markovian setting of problem \eqref{Value function control problem }, we thus expect that the value function $V$ should identify with an appropriate solution to the Hamilton-Jacobi-Bellman equation
\begin{align}\label{HJB}
\max \left\{ \left( \mathcal{L}_{X,\Pi} -r\right) u , ~ - \alpha u_x - u_y + e^x - \kappa \right\} = 0, \quad (x,y,\pi) \in \mathbb{R} \times (0,\infty)\times (0,1),
\end{align} 
with boundary condition $u(x,0,\pi)=0$, since $y=0$ implies $\mathcal{A}(0) = \{\xi \equiv 0\}$ and $J(x,0, \pi , 0) =0$. 
It is worth noticing that the variable $y$ plays the role of a parameter in \eqref{HJB}, which is then a two-dimensional elliptic partial differential equation with a state-dependent directional derivative constraint, parametrized by $y > 0$. With reference to \eqref{HJB} and the reasoning above, we can introduce the \textit{waiting region}
\begin{align}\label{Waiting region}
\mathbb{W}_1 := \{ (x,y,\pi) \in \mathbb{R} \times (0,\infty)\times (0,1): ~ (\mathcal{L}_{X,\Pi} - r) V = 0, \, - \alpha V_x - V_y + e^x - \kappa < 0 \} ,
\end{align} 
in which it is expected to be suboptimal to sell any assets, and the \textit{selling/execution region}, where it should be profitable for the investor to sell a fraction of the assets:
\begin{align}\label{Selling region}
\mathbb{S}_1 := \{ (x,y,\pi) \in \mathbb{R} \times (0,\infty)\times (0,1): ~ (\mathcal{L}_{X,\Pi} - r) V \leq 0, \, - \alpha V_x - V_y + e^x - \kappa = 0 \}.
\end{align}
Due to the multi-dimensional structure of the problem, a traditional \textit{guess-and-verify approach}, as seen for instance in Guo and Zervos \cite{GuZe} and Ferrari and Koch \cite{FeKo}, is not effective. In fact, this would require the construction of an explicit solution to the second-order PDE with state dependent gradient constraint seen in \eqref{HJB} above, which is not feasible in general. Instead, we use a different approach and construct an optimal stopping problem connected to the stochastic control problem \eqref{Value function control problem }, which is then of a simpler structure. Before we do so, and in order to get insights from a benchmark problem, we briefly discuss the problem under \textit{full information}, i.e.~where the drift coefficient is constant and equal to either $\mu_0$ or $\mu_1$. 
\section{Benchmark Problem under Full Information}\label{Section: Benchmark Problem}
Suppose that the initial subjective belief $\pi = \mathbb{P}[\mu = \mu_1]$ is such that $\pi \in \{ 0,1\}$. Observe that there exists no uncertainty in the model other than the Brownian one and the belief process $\Pi$ will remain constant, as the investor is already certain at initial time regarding the true value of the drift. Hence  - in this formulation - we are in the case of \textit{full information}. The problem we address in this section has a similar structure to the ones studied by Guo and Zervos \cite{GuZe} as well as Koch \cite{Ko}, Chapter 2, and we therefore do not provide full details. Let us assume $\pi = 0$, we thus obtain $\Pi_t = 0$ for all $t \geq 0$ and the dynamics of $X^\xi$ and $Y^\xi$ then write as
\begin{align}\label{Dyn underline X}
\underline{X}_t^\xi = x + \mu_0 t + \sigma W_t - \alpha  \xi_t, \qquad Y_t^\xi = y - \xi_t.
\end{align}
We denote the corresponding value function as 
\begin{align}\label{Value function V mu0}
V_0 (x,y) := \sup_{\xi \in \mathcal{A}(y) } \mathbb{E}_{(x,y)} \Big[ \int_0^\infty e^{-rt} (e^{\underline{X}_t^\xi} - \kappa ) \circ d \xi_t \Big], \qquad (x,y) \in \mathbb{R} \times (0,\infty),
\end{align}
where $\mathbb{E}_{(x,y)} [\cdot] = \mathbb{E}[\cdot \vert \underline{X}_{0-}^\xi = x , Y_{0-}^\xi =y ]$. By employing similar arguments as in the case of incomplete information, we can expect that $V_0$ should identify with an appropriate solution to the HJB equation
\begin{align}\label{HJB full info}
\max \{ (\mathcal{L}_{\underline{X}} - r) w , ~ - \alpha w_x - w_y + e^x - \kappa \} = 0, \qquad \text{with}~ \mathcal{L}_{\underline{X}} = \frac12 \sigma^2 \partial_{xx} + \mu_0 \partial_x,
\end{align}
and $w(x,0) = 0$. Defining the associated waiting and selling regions as
\begin{align}\label{Waiting region mu1}
\mathbb{W}^{\mu_0} &:= \{ (x,y) \in \mathbb{R} \times [0, \infty):~ (\mathcal{L}_{\underline{X}} - r) w(x,y) = 0 , - \alpha w_x - w_y + e^x - \kappa < 0 \}, \\ \label{Selling region mu1}
\mathbb{S}^{\mu_0} &:= \{ (x,y) \in \mathbb{R} \times [0, \infty):~ (\mathcal{L}_{\underline{X}} - r) w(x,y) \leq 0 , - \alpha w_x - w_y + e^x - \kappa = 0 \},
\end{align}
we suppose that the investor is only willing to sell a share of assets when its price is sufficiently large. Hence, we guess that for every $y \geq 0$ there exists a critical price $G(y)$ such that \eqref{Waiting region mu1}-\eqref{Selling region mu1} rewrite as
\begin{align*}
\mathbb{W}^{\mu_0} &= \{ (x,y) \in \mathbb{R} \times [0, \infty): ~y >0~\text{and}~ x < G(y) \} \cup (\mathbb{R} \times \{0\}), \\
 \text{and} \qquad \mathbb{S}^{\mu_0} &= \{ (x,y) \in \mathbb{R} \times [0,\infty) : ~ y> 0 ~\text{and}~x \geq G(y)\}. 
\end{align*}
Notice that the candidate value function should then satisfy $(\mathcal{L}_{\underline{X}} -r) w(x,y) = 0$ for all $(x,y)\in \mathbb{W}^{\mu_0}$. It is well-known that the latter equation admits two fundamental strictly positive solutions; the only solution that remains bounded as $x \downarrow - \infty$ is then given by
\begin{align*}
w(x,y) = A(y) e^{nx}, 
\end{align*}
for some functions $A: [0,\infty) \to \mathbb{R}$ and where $n$ is the positive solution to $(\sigma^2/2) n^2 + \mu_0 n -r = 0$. On the other hand, for $(x,y) \in \mathbb{S}^{\mu_0}$, we expect that the value function $V_0$ should instead satisfy 
\begin{align*}
- \alpha w_x - w_y + e^x - \kappa = 0
\qquad \text{and thus} \qquad
- \alpha w_{xx} - w_{yx} + e^x = 0.
\end{align*}
In order to derive the solutions for $A(y)$ and $G(y)$, we evaluate the two previous formulas at $x = G(y)$, require that $A(0)=0$ and obtain
\begin{align}\label{x0*}
G(y) = \ln \Big( \frac{\kappa n}{n-1} \Big) =: x_0^*
\quad \text{and} \quad
A(y) = \frac{\kappa}{\alpha n (n-1)} \Big( \frac{\kappa n}{n-1} \Big)^{-n} \Big( 1 - e^{- \alpha n y} \Big).
\end{align}
Notice that the optimal execution threshold - determining the price at which the investor should sell - is independent of the current amount of assets in the portfolio. Moreover, the selling region is partitioned into 
\begin{align*}
\mathbb{S}_1^{\mu_0} &:= \Big\{ (x,y) \in \mathbb{R} \times (0, \infty): ~ x \geq x^*_0,~y \leq \frac{x - x^*_0}{\alpha} \Big\} \\
\text{and} \qquad \mathbb{S}_2^{\mu_0} &:= \Big\{ (x,y) \in \mathbb{R} \times (0,\infty):~ x\geq x^*_0,~y > \frac{x -x^*_0}{\alpha} \Big\},
\end{align*}
and we suppose that for $(x,y) \in \mathbb{S}_1^{\mu_0}$ it should be optimal to sell the complete amount of assets instantaneously, while for $(x,y) \in \mathbb{S}_2^{\mu_0} $ the investor is expected to make a lump-sum execution and then follow the strategy that keeps the process $(X,Y)$ inside $\overline{\mathbb{W}}^{\mu_0}$ until all assets are sold. The candidate value function, according to our previous considerations, then takes the shape 
\begin{align}\label{Function w for full info}
w(x,y) = 
\begin{cases}
A(y) e^{nx} & \text{for}~ (x,y) \in \mathbb{W}^{\mu_0},\\
A(y - \frac{x-x_0^*}{\alpha})e^{n x_0^*}  + \frac{1}{\alpha} (e^x- e^{x_0^*}) - \frac{\kappa}{\alpha} (x - x^*_0) & \text{for}~(x,y) \in \mathbb{S}_2^{\mu_0}, \\
\frac{1}{\alpha} e^x \big( 1 - e^{- \alpha y} \big) -  \kappa y & \text{for}~ (x,y) \in \mathbb{S}_1^{\mu_0}, \\
\end{cases}
\end{align}
and via a verification theorem (cf.\,Guo and Zervos \cite{GuZe}, Prop.~5.1, Koch \cite{Ko}, Prop. 2.4.1), one can indeed show that $w$ is a $C^{2,1}$ solution to the HJB equation \eqref{HJB full info} and coincides with the value function $V_0$ of \eqref{Value function V mu0}. Moreover, the process
\begin{align}\label{Optimal control xi mu0}
\xi_t^{\mu_0} := y \wedge \sup_{ 0 \leq s \leq t} \frac{1}{\alpha} \Big[ x - x_0^* + \mu_0 s + \sigma W_s \Big]^+ , \qquad t \geq 0, \quad \xi_{0-}^{\mu_0} = 0 , 
\end{align}
belongs to $\mathcal{A}(y)$ and provides an optimal execution strategy for problem \eqref{Value function V mu0} (cf.~Guo and Zervos \cite{GuZe}, Prop.~5.1; recall that here we are not assuming $\lim_{T \uparrow \infty} Y_T^\xi = 0$ as admissibility condition, see also Remark \ref{Remark: limY}).
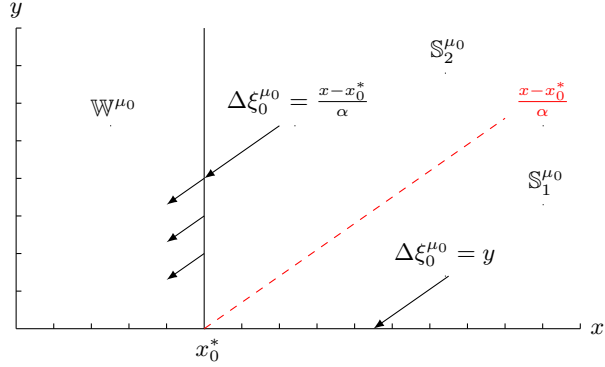
\begin{figure}
\begin{center}
\begin{tikzpicture}
\draw (0,0) -- (7.5,0) node[anchor=west] {\begin{footnotesize}$x$ \end{footnotesize}};
\draw (0,0) -- (0,4) node[anchor=south] {\begin{footnotesize}$y$ \end{footnotesize}};
\draw (2.5,0) -- (2.5,4);
\foreach \x in {0.5,1,1.5,2,3,3.5,4,4.5,5,5.5,6,6.5,7,7.5}
   \draw (\x cm,2pt) -- (\x cm,0pt);
\foreach \x in {0.5,1,1.5,2,2.5,3,3.5,4}
   \draw (2pt,\x cm) -- (0pt,\x cm);
\draw (2.5 cm,0pt) -- (2.5 cm,0pt) node[anchor=north] {\begin{footnotesize}
 $x_0^*$
\end{footnotesize}};
\draw[-latex] (2.5,1) -- (2,0.65);
\draw[-latex] (2.5,1.5) -- (2,1.15);
\draw[-latex] (2.5,2) -- (2,1.65);
\draw[-latex] (3.5,2.7) -- (2.5,2);
\draw[-latex] (5.75,0.7) -- (4.75,0);

\draw (3.7,2.7) -- (3.7,2.7) node[anchor=south] {\begin{footnotesize}
$\Delta \xi_0^{\mu_0} = \frac{x-x_0^*}{\alpha}$
\end{footnotesize}};
\draw (7,2.7) -- (7,2.7) node[anchor=south]
 {\begin{footnotesize}$\textcolor{red}{\frac{x-x_0^*}{\alpha}}$
\end{footnotesize}};
\draw (5.7,0.7) -- (5.7,0.7) node[anchor=south] {\begin{footnotesize}$\Delta \xi_0^{\mu_0} = y$
\end{footnotesize}};
\draw[red,dashed] (2.5,0) -- (6.5,2.8);
\draw (7,1.65) -- (7,1.65) node[anchor=south] {\begin{footnotesize}
$\mathbb{S}_1^{\mu_0}$
\end{footnotesize}};
\draw (5.7,3.4) -- (5.7,3.4) node[anchor=south] {\begin{footnotesize}
$\mathbb{S}_2^{\mu_0}$
\end{footnotesize}};
\draw (1.25,2.7) -- (1.25,2.7) node[anchor=south] {\begin{footnotesize}
$\mathbb{W}^{\mu_0}$
\end{footnotesize}};
\end{tikzpicture}
\caption{Illustrative drawing of the optimal execution strategy \eqref{Optimal control xi mu0} under full information.}
\label{Fig: ControlFull}
\end{center}
\end{figure}
Figure \ref{Fig: ControlFull} sketches the optimal execution strategy \eqref{Optimal control xi mu0} for problem \eqref{Value function V mu0} under full information. We observe that, for an initial price $x$ strictly larger than $x_0^*$, the investor immediately does a lump-sum execution. The latter can already deplete the whole portfolio whenever $y \leq \frac{1}{\alpha} (x- x_0^*)$, or bring it to the level $(x_0^*, y - \frac{1}{\alpha}(x - x_0^*)$ otherwise. Afterwards, the optimal strategy prescribes to keep the state process $(X,Y)$ inside the waiting region $\overline{\mathbb{W}}^{\mu_0}$ with minimal effort, by reflecting it in the direction $(-\alpha,-1)$ according to a \textit{Skorokhod reflection-type} policy (realized through the running supremum in \eqref{Optimal control xi mu0}).\\
In light of our subsequent analysis, it is interesting to notice that the directional derivative $v_0 := \alpha \partial_x V_0 + \partial_y V_0$ can be checked from \eqref{Function w for full info} to identify with the value function of an optimal stopping problem. More precisely, for any $x \in \mathbb{R}$ one has
\begin{align}\label{Value function v0}
\alpha \partial_x V_0 + \partial_y V_0 =: v_0 (x) = \sup_{\tau \geq 0} \mathbb{E}_x  [e^{-r \tau} (e^{\underline{X}_\tau^0 } - \kappa ) ],
\end{align}
where $\underline{X}^0$ denotes the solution to \eqref{Dyn underline X} with $\xi_t \equiv 0$, the optimization is performed over all stopping times of the Brownian filtration and $\mathbb{E}_x$ is the expectation under $\mathbb{P}_x [\cdot] = \mathbb{P}[\cdot \mid \underline{X}_0^0 = x]$. Moreover, the stopping time 
\begin{align}\label{OST tau 0 *}
\tau_0^*(x) := \inf \{ t \geq 0: ~ \underline{X}_t^0 \geq x_0^* \}, \qquad \mathbb{P}_x \text{-a.s.},~  x \in \mathbb{R},
\end{align}
is optimal for \eqref{Value function v0}. We can interpret \eqref{OST tau 0 *} as the optimal time at which the investor should sell another unit of shares, and notice that it in fact characterizes the time at which the marginal expected profit $\alpha \partial_x V_0 + \partial_y V_0$ coincides with the marginal instantaneous net profit $e^x - \kappa$ from selling. 

\begin{Remark} \label{Remark: Full Info mu1}
It is easily checked that the results we obtained for the case $\mu \equiv \mu_0$ can be replicated for the case $\mu \equiv \mu_1$. More precisely, considering the dynamics
\begin{align}\label{Dyn Xoverline}
\overline{X}_t^\xi = x + \mu_1 t + \sigma W_t - \alpha \xi_t, \qquad t \geq 0,
\end{align}
and the value function 
\begin{align}\label{Value function V1}
V_1 (x,y) := \sup_{\xi \in \mathcal{A}(y)} \Big[ \int_0^\infty e^{-rt } ( e^{\overline{X}_t^\xi} - \kappa) \circ d \xi_t \Big] , \qquad (x,y) \in \mathbb{R} \times (0, \infty),
\end{align}
we can verify the existence of an optimal execution threshold $x_1^*$, which triggers the selling strategy of the investor through the optimal control $\xi^{\mu_1}$, which is of similar structure as \eqref{Optimal control xi mu0}, with $\mu_0$ replaced by $\mu_1$. Furthermore, we have
\begin{align} \label{OSP with mu 1}
\alpha \partial_x V_1 + \partial_y V_1 =: v_1 (x) = \sup_{t \geq 0} \mathbb{E}_x [ e^{-r\tau} (e^{\overline{X}_\tau^0} - \kappa) ],
\end{align}
where $\overline{X}^0$ denotes the solution to \eqref{Dyn Xoverline} with $\xi_t \equiv 0$, and the stopping time $\tau_1^* (x) := \inf \{ t \geq 0 : ~ \overline{X}_t^0 \geq x_1^* \}$, $\mathbb{P}_x$-a.s., is optimal for problem \eqref{OSP with mu 1}.
\end{Remark}

 \section{A Related Optimal Stopping Problem}\label{Section: First related OSP} 
Motivated by the observed connection to an optimal stopping problem in the benchmark problem of Section \ref{Section: Benchmark Problem} (see \eqref{OSP with mu 1}), we pursue the following approach in the subsequent analysis: (i) we introduce and study an optimal stopping problem with value $v$, that we expect to be associated to the  singular stochastic control problem \eqref{Value function control problem }; (ii) we provide a complete analysis of the optimal stopping problem, which is achieved by studying two equivalent formulations of it (cf.\ Sections \ref{Section: Decoupling Change of Measure} and \ref{Section: Parabolic Formulation}). More precisely, we derive regularity results of the value function (cf.\,Proposition \ref{Proposition: Smooth fit}), as well as an integral equation for the free boundary (cf.\,Proposition \ref{Proposition: Integral Equation for c}); (iii) we verify the expected connection to the original problem of \eqref{Value function control problem } by showing that (cf.\ Theorem \ref{Theorem: Verification Theorem}) 
\begin{align*}
V(x,y,\pi) = \frac{1}{\alpha} \int_{x- \alpha y}^x v(x',\pi) dx', \hspace{0.5cm} (x,y,\pi) \in \mathbb{R} \times (0,\infty) \times (0,1),
\end{align*}
and that the optimal execution strategy is triggered by the optimal stopping boundary studied in the previous step. In fact, as in the benchmark case, we can interpret the optimal stopping problem as the \textit{marginal problem}, in the sense that its value coincides with the derivative of the value $V$ of \eqref{Value function control problem } in the direction of actions/execution and its optimal stopping strategy characterizes the time at which it is optimal to sell a unit of assets. 

We recall that $(X^0_t , \Pi_t)_{t \geq 0}$ is the two-dimensional strong Markov process solving 
 \begin{align}\label{Dynamics Opt Stop Problem 1}
\begin{cases}
d X_t^0 = (\mu_1 \Pi_t + \mu_0 (1 - \Pi_t)dt + \sigma d \overline{W}_t, &X_0^0 =x, \\
d \Pi_t = \gamma \Pi_t (1- \Pi_t) d \overline{W}_t, &\Pi_0 = \pi, 
\end{cases}
\end{align}
and in the following - in order to simplify notation - we write $X$ instead of $X^0$. For a stopping time $\tau$ of the filtration $\mathbb{F}^X$, we then define
 \begin{align}\label{Objective OSP}
\Psi (x,\pi,\tau) := \mathbb{E}_{(x,y)} [ e^{-r \tau } (e^{X_\tau} - \kappa )], \quad (x,\pi) \in \mathbb{R}\times (0,1), 
 \end{align}
and consider the optimal stopping problem 
\begin{align}\label{Value fct of opt stopp}
 v(x,\pi):= \sup_{\tau} \Psi (x,\pi,\tau).
\end{align} 
Above, and in the following, $\mathbb{E}_{(x,\pi)} [\cdot] = \mathbb{E}[ \cdot \vert X_0 = x, \Pi_0 = \pi]$. Also, denoting $(X_t^{x,\pi} , \Pi_t^\pi)_{t \geq 0}$ the unique strong solution to \eqref{Dynamics Opt Stop Problem 1} we will often employ the following equivalent notation $\mathbb{E}[f(X_t^{x,\pi} , \Pi_t^\pi )] = \mathbb{E}_{x, \pi} [f (X_t , \Pi_t)]$, for any integrable measurable function $f: \mathbb{R} \times [0,1] \to \mathbb{R}$. 

We make the next \textbf{standing assumption}.
\begin{assumption}\label{Assumption: Well-posedness}
We assume $ r > \big(\mu_1 + \frac{1}{2} \sigma^2 \big) \vee \big( \mu_1 + \frac12 \sigma^2 + \frac{(2 \mu_1 + \sigma^2)(\mu_1 - \mu_0)}{\sigma^2} \big) \vee \big( \frac{\gamma}{2 \sigma} \vert \mu_0 + \mu_1 \vert \big) $.
\end{assumption}
\begin{Remark}
(i) The different conditions we impose on the (subjective) discount factor $r$ serve distinct purposes. Notice that the the first condition is equivalent to imposing $r > \beta_1$ and guarantees well-posedness of problem \eqref{Value fct of opt stopp}. \\
(ii) Moreover, the forthcoming analysis (in particular Section \ref{Section: Parabolic Formulation}) reveals that the other two terms are sufficient to ensure monotonicity of (a transformation of) the optimal stopping boundary of the latter problem (cf. Propositions \ref{Proposition Mono 1} and \ref{Proposition Mono 2}). This result is crucial when deriving the smooth-fit property and thus, by relying on arguments developed in De Angelis and Peskir \cite{DeAPe}, the global $C^1$-regularity of (a transformation of) the value function $v$ of \eqref{Value fct of opt stopp}. When $r$ does not satisfy Assumption \ref{Assumption: Well-posedness}, the monotonicity of the (transformed version of the) boundary is not clear, and thus one needs an alternative route to achieve the needed regularity of $\hat{v}$. A possible approach could be to prove directly the (locally) Lipschitz-regularity of the free-boundary (cf. De Angelis and Stabile \cite{DeASt}), and then infer the $C^1$-property of $\hat{v}$ from the continuity of the optimal stopping time. Since this is not straightforward to obtain in our formulation, we leave it for future research.
\end{Remark}
In the following, we derive some preliminary results of the optimal stopping problem \eqref{Value fct of opt stopp} and its associated free boundary. Noticing that $(x,\pi) \mapsto X_t^{x,\pi}$ as well as $\pi \mapsto \Pi_t^\pi$ are continuous and nondecreasing, due to classical comparison theorems for strong solutions to stochastic differential equations, the proof of the following lemma follows from standard arguments and it is therefore skipped. 
\begin{lemma}\label{Lemma: Properties of v}
The value function v of (\ref{Value fct of opt stopp}) is such that 
\begin{itemize}
\item[i)] $x \mapsto v(x, \pi)$ is nondecreasing;
\item[ii)] $\pi \mapsto v(x, \pi)$ is nondecreasing;
\end{itemize}
\end{lemma}
Furthermore, using that $(x, \pi) \mapsto (X_t^{x,\pi}, \Pi_t^\pi)$ is continuous $\mathbb{P}$-a.s., by Assumption \ref{Assumption: Well-posedness} and standard estimates using the fact that $\Pi$ is bounded on $[0,1]$ we can invoke dominated convergence and obtain that
\begin{align*}
(x, \pi) \mapsto \mathbb{E} [e^{-r \tau} (e^{X_\tau^{x, \pi}} - \kappa) ],
\end{align*}
is continuous and hence, $(x,\pi) \mapsto v(x, \pi)$ is lower-semicontinuous. 
As it is customary in optimal stopping theory, we introduce the continuation and stopping regions associated to $v$ as  
\begin{align}\label{Continuation region C1}
\mathcal{C}_1 &:= \{ (x, \pi) \in \mathbb{R} \times (0,1): \quad v(x, \pi) > (e^x- \kappa) \}, \\
\label{Stopping region S1}
\mathcal{S}_1 &:= \{ (x, \pi) \in \mathbb{R} \times (0,1): \quad v(x, \pi) = (e^x- \kappa) \}.
\end{align}
Then, the continuation region $\mathcal{C}_1$ is an open set, while the stopping region $\mathcal{S}_1$ in \eqref{Stopping region S1} is closed, and by Peskir and Shiryaev \cite{PeSh}, Chapter 1, Section 2, Corollary 2.9, the stopping time 
\begin{align}\label{Optimal stopping time tau*}
\tau^* = \tau^* (x, \pi) := \inf \{ t \geq 0:~ (X_t^{x,\pi} , \Pi_t^\pi ) \in \mathcal{S}_1 \},
\end{align}
is optimal whenever it is $\mathbb{P}$-a.s.\,finite, otherwise it is an optimal Markov time. We set 
\begin{align}\label{boundary a}
a (\pi) := \inf \{ x \in \mathbb{R}: \quad v(x, \pi) \leq  (e^x- \kappa) \} ,
\end{align}
with the convention $\inf \emptyset = + \infty$, and state the following lemma. 
\begin{lemma}\label{Lemma: C1,S1 in terms of a}
It holds
\begin{align*}
\mathcal{C}_1 = \{ (x, \pi) \in \mathbb{R} \times (0,1): \quad x < a (\pi) \} \quad \text{and} \quad \mathcal{S}_1 = \{ (x, \pi) \in \mathbb{R} \times (0,1): \quad x \geq a (\pi) \}.
\end{align*}
\end{lemma}
\begin{proof}
Recalling \eqref{Operator LXPi}, an application of Dynkin's formula reveals 
\begin{align*}
u(x,\pi) := v(x,\pi) - (e^x -\kappa) = \sup_\tau \mathbb{E}_{(x,\pi)} \Big[ \int_0^\tau e^{-rt} \Big( e^{X_t} \big( \mu_1 \Pi_t + (1 - \Pi_t) \mu_0  + \frac12 \sigma^2 - r \big) +r \kappa \Big) dt \Big].
\end{align*}
For $x_1 < x_2$ and $\tau^*$ optimal for $v(x_2 , \pi)$ we have 
\begin{align*}
u(x_1,\pi) - u(x_2,\pi) &\geq \mathbb{E} \Big[ \int_0^{\tau^*} e^{-rt} \Big( e^{X^{x_1, \pi}_t} - e^{X^{x_2, \pi}_t} \Big) (\mu_1 \Pi_t^\pi + (1 - \Pi_t^\pi) \mu_0  + \frac12 \sigma^2 -r )   dt \Big] \\
& \geq \mathbb{E} \Big[ \int_0^{\tau^*} e^{-rt} \Big( e^{X^{x_2, \pi}_t} - e^{X^{x_1, \pi}_t} \Big) (r - \mu_1   - \frac12 \sigma^2)   dt \Big] \geq 0,
\end{align*}
due to Assumption \ref{Assumption: Well-posedness}. For $(x_1, \pi) \in \mathcal{S}_1$ and $x_2 > x_1$, we thus obtain $0 \leq u(x_2,\pi) \leq u(x_1, \pi) = 0$, so that $(x_2, \pi) \in \mathcal{S}_1$.  \qed
\end{proof}
The free boundary $a(\pi)$ thus splits $\mathbb{R} \times (0,1)$ into the continuation and stopping region. In the following lemma we derive some preliminary properties.
\begin{lemma}\label{Lemma: Properties of boundary a}
One has:
\begin{itemize}
\item[i)] $\pi \mapsto a(\pi)$ is nondecreasing on $(0,1)$;
\item[ii)] $\pi \mapsto a(\pi)$ is left-continuous on $(0,1)$;
\item[iii)] There exist constants such that $x_0^* \leq a(\pi) \leq x_1^*$ for all $\pi \in (0,1)$.
\end{itemize}
\end{lemma}
\begin{proof} We prove the claims separately. \\[0.1cm]
i) Let $\pi_2 > \pi_1$ and $(x,\pi_2)\in \mathcal{S}_1$. We thus have $x \geq a(\pi_2)$ and $v(x, \pi_2)=e^x-\kappa$. Since $\pi \mapsto v(x,\pi) $ is nondecreasing,  $v(x,\pi_1) \leq v(x,\pi_2) = e^x-\kappa$, which, together with $v(x, \pi_1) \geq (e^x-\kappa)$, gives $(x, \pi_1) \in \mathcal{S}_1$. Therefore, $a(\pi_2) \geq a(\pi_1)$.  \\[0.1cm]
ii) Let $(\pi_n)_n$ be a sequence such that $\pi_n \uparrow \pi$. Due to i), the sequence $a(\pi_n)$ is increasing as $n \to \infty$ and $a(\pi_n) \leq a (\pi)$. Consequently, there exists $\lim_n a(\pi_n) =: a(\pi -)$ and $a(\pi -) \leq a(\pi)$. Because $v(a(\pi_n), \pi_n) = e^{a(\pi_n)} - \kappa$ for all $n \in \mathbb{N}$, by lower-semicontinuity of $(x, \pi) \mapsto v(x,\pi)$ we find $v(a(\pi -), \pi) = e^{a(\pi -)} -\kappa$. Hence, $a(\pi) \leq a(\pi^-)$ and thus $\lim_n a(\pi_n) = a(\pi)$. \\[0.1cm]
iii) Recall $v_0$ and $v_1$ of \eqref{Value function v0} and \eqref{OSP with mu 1}, the value functions in the optimal stopping problems with full information when either $\mu \equiv \mu_0$ or $\mu \equiv \mu_1$. The associated continuation regions are given by 
\begin{align*}
\{ x \in \mathbb{R}:~x \geq x_1^* \} = \{x \in \mathbb{R}:~v_1(x) \leq e^x -\kappa \} \quad \text{and} \quad  \{ x \in \mathbb{R}:~x\geq x_0^* \} = \{ x \in \mathbb{R}:~v_0 (x) \leq e^x -\kappa \},
\end{align*}
where $x_0^*$ and $x_1^* $ are the optimal execution thresholds (cf. \eqref{x0*} and Remark \ref{Remark: Full Info mu1}). Recalling $\mu_0 < \mu_1$ and $\Pi_t \in (0,1)$ for $\pi \in (0,1)$, we have $\underline{X}^0_t \leq X_t \leq \overline{X}_t^0$ $\mathbb{P}$-a.s.~for any $t \geq 0$, due to classical comparison arguments and where $\underline{X}^0$ and $\overline{X}^0$ denote the solutions to \eqref{Dyn underline X} and \eqref{Dyn Xoverline} with $\xi_t \equiv 0$. Thus, $v_0 (x) \leq v(\pi ,x) \leq v_1 (x)$, which implies
\begin{align*}
 \{ x \in \mathbb{R}:~v_1 (x) \leq e^x -\kappa \} \subset \{ (x, \pi)\in \mathbb{R}\times (0,1):~v(x,\pi) \leq e^x - \kappa \} \subset \{x \in \mathbb{R}:~ v_0(x) \leq e^x -\kappa \}, 
\end{align*}
and the latter, combined with \eqref{boundary a}, allows to conclude that $x_0^* \leq a(\pi) \leq x_1^*$. \qed
\end{proof}
\section{Decoupling Change of Measure and a new optimal selling problem}\label{Section: Decoupling Change of Measure} 
We notice that the underlying dynamics in \eqref{Dynamics Opt Stop Problem 1} are coupled. In order to derive further results about the properties of the optimal stopping problem \eqref{Value fct of opt stopp} and its associated free boundary, it is useful to adress the problem under a different probability measure. With reference to related contributions (cf. De Angelis \cite{DeA}, Ekstr\"{o}m and Lu \cite{EkLu}, Johnson and Peskir \cite{JoPe} and Shiryaev \cite{Sh} and references therein), we introduce the so-called \textit{likelihood ratio process} via
\begin{align*}
\Phi_t := \frac{\Pi_t}{1 - \Pi_t}, \qquad t \geq 0.
\end{align*} 
Through an application of It\^{o}'s formula we can derive its associated dynamics, given by 
\begin{align}\label{Dynamics Phi}
d \Phi_t = \gamma \Phi_t (\gamma \Pi_t dt + d \overline{W}_t ), \qquad \Phi_0 = \varphi := \frac{\pi}{1 - \pi},
\end{align}
and we aim to remove its dependency on the process $\Pi$ through a change of measure. For a fixed $T>0$, we define the measure $\mathbb{Q}_T \sim \mathbb{P}$ on $(\Omega, \mathcal{F}_T)$ via the Radon-Nikodym derivative 
\begin{align}\label{Process eta change of measure}
\eta_T := \frac{d Q_T}{d \mathbb{P}} := \exp \Big( - \int_0^T \gamma \Pi_s d \overline{W}_s - \frac12 \int_0^T \gamma^2 \Pi_s^2 ds \Big),
\end{align}
and notice that the process 
\begin{align}\label{Brownian Motion Girsanov}
dB_t = d \overline{W}_t + \gamma \Pi_t dt,
\end{align}
is a Brownian motion under $\mathbb{Q}_T$ on $[0,T]$. 
Rewriting the state process $(X,\Phi)$ under $\mathbb{Q}_T$ then yields
\begin{align}\label{Dynamics after Girsanov}
\begin{cases}
d X_t = \mu_0 dt + \sigma d B_t, & t \in (0,T], \quad X_0 = x, \\
d \Phi_t = \gamma \Phi_t d B_t, &t \in (0,T], \quad \Phi_0 = \varphi,
\end{cases}
\end{align}
and we notice that the processes decouple under this formulation. In the following, when needed, we will write $\mathbb{E}_{(x, \varphi)}^{\mathbb{Q}_T} $ to denote the expectation under $\mathbb{Q}_T$,  conditioned on $X_0=x, \Phi_0 = \varphi$. In order to rewrite problem (\ref{Value fct of opt stopp}) in terms of the new variables $(X, \Phi)$, we introduce 
\begin{align*}
\Theta_t := \frac{1+\Phi_t}{1+\varphi} , \qquad t \in [0,T],
\end{align*} 
and by an application of It\^{o}'s formula, it can be verified that $\Theta$ admits the representation
\begin{align}\label{Process Theta Change of measure}
\Theta_t = \exp \Big( \int_0^t \gamma \Pi_s d \overline{W}_s + \frac12 \int_0^t \gamma^2 \Pi_s^2 ds \Big) = \frac{1}{\eta_t}, \quad \qquad t \in [0,T].
\end{align}
Upon using \eqref{Process eta change of measure} and \eqref{Process Theta Change of measure}, we find 
\begin{align}\label{Change of measure Equality}
\mathbb{E}_{(x, \pi)} [e^{-r(\tau \wedge T)} (e^{X_{\tau \wedge T}} - \kappa) ] 
&= \mathbb{E}_{(x,\pi)} [ e^{-r(\tau \wedge T)} (e^{X_{\tau \wedge T}} -\kappa) \eta_{\tau \wedge T} \Theta_{\tau \wedge T} ] \nonumber \\
&= \mathbb{E}^{\mathbb{Q}_T}_{(x, \varphi)} \Big[ e^{-r(\tau \wedge T)} (e^{X_{\tau \wedge T}} -\kappa) \frac{1 + \Phi_{\tau \wedge T}}{1+\varphi} \Big] \nonumber \\
&=(1 + \varphi)^{-1} \mathbb{E}^{\mathbb{Q}_T}_{(x, \varphi)} [ e^{-r(\tau \wedge T)} (e^{X_{\tau \wedge T}} -\kappa) (1 + \Phi_{\tau \wedge T}) ],
\end{align}
for any stopping time $\tau$ and $(x, \varphi) \in \mathbb{R} \times (0,\infty)$. With regard to \eqref{Change of measure Equality} we introduce the stopping problems
\begin{align*}
v(x,\pi; T) &:= \sup_\tau \mathbb{E}_{(x,\pi)} [e^{-r(\tau \wedge T)}(e^{X_{\tau \wedge T}} -\kappa)], \\
\text{and} \qquad \quad 
v^{\mathbb{Q}_T} (x,\varphi; T) &:= \sup_\tau \mathbb{E}^{\mathbb{Q}_T}_{(x,\varphi)} [e^{-r(\tau \wedge T)}(e^{X_{\tau \wedge T}} -\kappa)(1 + \Phi_{\tau \wedge T})],
\end{align*}
and notice that \eqref{Change of measure Equality} implies $v^{\mathbb{Q}_T} (x, \varphi ; T)= (1 + \varphi ) v(x , \varphi/(1 + \varphi) ; T)$ for fixed $T>0$.
 However, since the measure $\mathbb{Q}_T$ changes with $T$, passing to the limit $T \to \infty$ in the latter expression \eqref{Change of measure Equality} requires a bit of care. To this end, we define a probability space $(\ttilde{\Omega}, \ttilde{\mathbb{F}}, \ttilde{\mathbb{Q}})$ with a Brownian motion $\ttilde{B}$ and a filtration $\ttilde{\mathbb{F}} = (\ttilde{\mathcal{F}}_t)_{t \geq 0}$. Moreover, we let $(\ttilde{X}, \ttilde{\Phi})$ be the strong solution to the stochastic differential equation \eqref{Dynamics after Girsanov} driven by the Brownian motion $\ttilde{B}$ instead of $B$. Let $\ttilde{\mathbb{E}}_{(x, \varphi)} [\cdot]$ denote the expectation under $\ttilde{\mathbb{Q}}$ and define the stopping problems
 \begin{align*}
\overline{v} (x,\varphi; T) = \sup_\tau \ttilde{\mathbb{E}}_{(x,\varphi)} [e^{-r(\tau \wedge T)}(e^{\ttilde{X}_{\tau \wedge T}} -\kappa)(1 + \ttilde{\Phi}_{\tau \wedge T})], \quad 
\overline{v} (x,\varphi) = \sup_\tau \ttilde{\mathbb{E}}_{(x, \varphi)} [e^{-r \tau }(e^{\ttilde{X}_\tau} -\kappa)(1 + \ttilde{\Phi}_{\tau})].
 \end{align*} 
 Due to the equivalence in laws of the process $(\ttilde{X}_t, \ttilde{\Phi}_t, \ttilde{B}_t)_{t \geq 0}$ under $\ttilde{\mathbb{Q}}$ and the process $(X_t, \Phi_t , B_t)_{t \geq 0}$ under $\mathbb{Q}_T$ on $[0,T]$, we have $v^{\mathbb{Q}_T} (x,\varphi; T) = \overline{v}(x, \varphi ;T)$. Moreover, upon using Fatou's lemma and simple comparison arguments, one can show that 
 \begin{align*}
 \lim_{T \to \infty} v(x, \pi; T) = v(x,\pi) \qquad \text{as well as} \qquad \lim_{T \to \infty} \bar{v} (x, \varphi; T)  = \overline{v} (x, \varphi).
 \end{align*}
Hence, we finally obtain 
\begin{align}\label{v and vbar equivalence}
\overline{v} (x ,\varphi) &= \lim_{T \to \infty} \overline{v} (x, \varphi; T) = \lim_{T \to \infty} v^{\mathbb{Q}_T} (x, \varphi; T) \nonumber \\
&= (1 + \varphi) \lim_{T \to \infty} v(x , \varphi/(1 + \varphi) ;T) = (1 + \varphi) v(x, \varphi/(1 + \varphi)).
\end{align}
For the sake of clarity - and with a slight abuse of notation - from now on we simply write $(\Omega, \mathbb{F},(\mathcal{F}_t)_{t \geq 0},\mathbb{Q}, \allowbreak \mathbb{E}^\mathbb{Q},X ,\Phi, B)$ instead of $(\ttilde{\Omega}, \ttilde{\mathbb{F}}, (\ttilde{\mathcal{F}}_t)_{t \geq 0}, \ttilde{\mathbb{Q}}, \ttilde{\mathbb{E}}, \ttilde{X}, \ttilde{\Phi}, \ttilde{B})$. Henceforth, we thus study the optimal stopping problem 
\begin{align}\label{value fct opt after girsanov}
\overline{v} (x, \varphi) = \sup_\tau \mathbb{E}^\mathbb{Q}_{(x,\varphi)} [e^{-r \tau} (e^{X_\tau} -\kappa) (1 + \Phi_\tau) ].
\end{align}
In the sequel, we will often write $\mathbb{E}_{(x,\varphi)}^\mathbb{Q} [ f (X_t , \Phi_t )] = \mathbb{E}^\mathbb{Q} [f(X_t^x, \Phi_t^\varphi)]$, where $(X_t^x, \Phi_t^\varphi )_{t \geq 0}$ is the unique strong solution to \eqref{Dynamics after Girsanov}. The continuation and stopping region associated to this problem are then given by 
\begin{align} \label{continuation after girsanov}
\mathcal{C}_2 &:= \{ (x, \varphi) \in \mathbb{R} \times (0,\infty) : \quad \overline{v}(x,\varphi) > (e^x- \kappa)(1+ \varphi) \}, \\ \label{stopping after girsanov} \mathcal{S}_2 &:= \{ (x, \varphi) \in \mathbb{R} \times (0,\infty) : \quad \overline{v}(x,\varphi) = (e^x- \kappa)(1+ \varphi) \}.
\end{align}
With regard to the lower-semicontinuity of $v$ and \eqref{v and vbar equivalence}, we find that $(x,\varphi) \mapsto \overline{v}(x, \varphi)$ is lower-semicontinuous as well. Hence, the stopping region $\mathcal{S}_2$ of \eqref{stopping after girsanov} is a closed set, while the continuation region $\mathcal{C}_2$ of \eqref{continuation after girsanov} is open. Also, $\tau^* := \tau^* (x , \varphi) := \inf \{ t \geq 0:~(X_t^x , \Phi_t^\varphi) \in \mathcal{S}_2 \}$ is optimal by Peskir and Shiryaev \cite{PeSh}, whenever $\mathbb{Q}$-a.s.~finite. Furthermore, we define
\begin{align}\label{boundary b}
b(\varphi) := \inf \{ x \in \mathbb{R}: \quad \overline{v}(x,\varphi) \leq (e^x-\kappa)(1+ \varphi) \},
\end{align}
with $\inf \emptyset = \infty$. In the following lemma, we derive some preliminary properties of the value function \eqref{value fct opt after girsanov}. In light of the relation \eqref{v and vbar equivalence} we notice that some of the following results are a direct consequence of Lemma \ref{Lemma: Properties of v}.
\begin{lemma}\label{Lemma: Properties of vbar}
The value function $\overline{v}$ of (\ref{value fct opt after girsanov}) is such that 
\begin{itemize}
\item[i)] $0 \leq \overline{v}(x,\varphi) \leq K_1 e^x (1 + \varphi)$ for all $(x, \varphi) \in \mathbb{R} \times (0, \infty)$ and some $K_1 > 0$;
\item[ii)] $x \mapsto \overline{v}(x, \varphi)$ is nondecreasing;
\item[iii)] $\varphi \mapsto \overline{v} (x, \varphi)$ is nondecreasing; 
\item[iv)] $(x, \varphi)\mapsto \overline{v}(x, \varphi)$ is locally Lipschitz over $\mathbb{R} \times (0,\infty)$; 
\item[v)] $\varphi \mapsto \overline{v}(x, \varphi)$ and $x \mapsto \overline{v} (x, \varphi)$ are convex. 
\end{itemize}
\end{lemma}
\begin{proof}
Property ii) follows from Lemma \ref{Lemma: Properties of v} i), upon using equality \eqref{v and vbar equivalence}. We prove the remaining claims separately. \\[0.1cm] 
i) For the lower bound, we notice that $\{ (x,\varphi) \in \mathbb{R} \times (0,\infty) : ~(e^x- \kappa) < 0 \} \subset \mathcal{C}_2$. Hence, since $\Phi^\varphi \geq 0$ a.s.,  we have $\overline{v}(x, \varphi) \geq 0$ for all $(x,\varphi) \in \mathbb{R} \times (0, \infty)$. For the upper bound, we observe that for any stopping time $\tau$ 
\begin{align*}
\mathbb{E}_{(x, \varphi)}^\mathbb{Q} \big[ e^{-r \tau} (e^{X_\tau} - \kappa) (1 + \Phi_\tau ) \big] 
&= (1 + \varphi)\mathbb{E}_{(x, \pi)} \big[ e^{-r \tau} ( e^{X_\tau} - \kappa) \big] \\
&\leq (1 + \varphi)\mathbb{E} \big[ e^{-r \tau} e^{x + \mu_1 \tau + \sigma W_\tau}  \big] \leq K_1 e^x (1 + \varphi),
\end{align*}
for $\pi = \varphi/(1 + \varphi)$ and the last inequality follows from standard estimates upon using Assumption \ref{Assumption: Well-posedness}. \\[0.1cm]
iii) Let $\varphi, \varphi' \in (0, \infty)$ with $\varphi' > \varphi$ and notice that $\Phi_t^\varphi = \varphi e^{- \frac12 \gamma^2 t + \gamma B_t}$. For $x\in \mathbb{R}$ and $\tau^* := \tau^* (x, \varphi)$ optimal for $\overline{v}(x,\varphi)$ we have
\begin{align*}
\overline{v}(x, \varphi') - \overline{v}(x, \varphi) &\geq \mathbb{E}_{(x, \varphi')}^\mathbb{Q} \big[ e^{-r \tau^*} \big( e^{X_{\tau^*}} - \kappa \big) \big( 1 + \Phi_{\tau^*} \big) \big] - \mathbb{E}_{(x, \varphi)}^\mathbb{Q} \big[  e^{-r \tau^*} \big( e^{X_{\tau^*}} - \kappa \big) \big( 1 + \Phi_{\tau^*} \big) \big] \\
&= \mathbb{E}^\mathbb{Q} \big[ e^{-r \tau^*} \big( e^{X_{\tau^*}^x} - \kappa\big) (\varphi' - \varphi)e^{- \frac12 \gamma^2 \tau^* + \gamma B_{\tau^*}} \big] \geq 0,
\end{align*}
where the last inequality exploits that $\{ (x,\varphi) \in \mathbb{R} \times (0,\infty): ~(e^x - \kappa)< 0\} \subset \mathcal{C}_2$, and the claim follows.  \\[0.1cm]
iv) Let $x,x' \in \mathbb{R}$, $\pi \in (0,1)$ and $\varphi,  \varphi' \in (0, \infty)$. Recall $v$ of \eqref{Value fct of opt stopp}. Again, standard estimates  yield 
\begin{align*}
\vert v(x,\pi) - v(x' , \pi) \vert \leq K_1 \vert e^x - e^{x'} \vert, \qquad \text{as well as} \qquad \vert \overline{v}(x, \varphi) - \overline{v}(x, \varphi') \vert \leq K_2 e^x \vert \varphi - \varphi' \vert ,
\end{align*}
for some $K_1, K_2 >0$. Hence, using \eqref{v and vbar equivalence}, we obtain
\begin{align}\label{vbar lipschitz}
\vert \overline{v} (x, \varphi) - \overline{v}(x', \varphi') \vert &\leq \vert \overline{v}(x, \varphi) - \overline{v}(x' ,\varphi) \vert + \vert \overline{v}(x', \varphi) - \overline{v}(x', \varphi') \vert \nonumber \\
&\leq  K_1 (1 + \varphi ) \vert e^x - e^{x'} \vert + K_2  e^{x'} \vert \varphi - \varphi' \vert,
\end{align}
and thus the locally-Lipschitz property follows. \\
v) We first prove convexity regarding $\varphi \in (0,\infty)$. For $\varphi_1, \varphi_2 \in (0,\infty)$, $x \in \mathbb{R}$ and $\lambda \in (0,1)$ we set $\overline{\varphi} := \lambda \varphi_1 + (1- \lambda) \varphi_2$ and obtain
\begin{align*}
\overline{v} (x, \overline{\varphi}) &= \sup_\tau \mathbb{E}^\mathbb{Q}_{(x, \overline{\varphi})} \Big[ e^{-r \tau } (e^{X_\tau} -\kappa) (1 + \overline{\varphi}
e^{ -\frac12 \gamma^2 \tau + \gamma B_\tau } ) \Big] \\
&\leq \sup_\tau \mathbb{E}^\mathbb{Q} \Big[ e^{-r \tau} (e^{X_\tau^x} -\kappa)  \lambda (1 + \varphi_1 e^{ -\frac12 \gamma^2 \tau + \gamma B_\tau } ) \Big]  \\
&\hspace*{4cm}+ \sup_\tau \mathbb{E}^\mathbb{Q} \Big[ e^{-r \tau} (e^{X_\tau^x} - \kappa ) (1- \lambda) (1 + \varphi_2 e^{ -\frac12 \gamma^2 \tau + \gamma B_\tau }) \Big] \\
& = \lambda \overline{v}(x, \varphi_1) + (1-\lambda)\overline{v} (x, \varphi_2) ,
\end{align*}
and the claim follows. Analogously, upon exploiting the convexity of $x \mapsto e^x$, one can prove the convexity of $x \mapsto \overline{v}(x, \varphi)$. \qed
\end{proof}
\begin{lemma} \label{Lemma C2,S2 in terms of b}
The continuation and stopping region regions as in \eqref{continuation after girsanov}-\eqref{stopping after girsanov} are such that
\begin{align*}
\mathcal{C}_2 = \{ (x,\varphi) \in \mathbb{R} \times (0,\infty) : ~ x < b(\varphi) \}, \qquad \mathcal{S}_2 = \{ (x,\varphi) \in \mathbb{R} \times (0,\infty): ~ x \geq  b(\varphi) \}.
\end{align*}
\end{lemma}
\begin{proof}
We proceed similarly to Lemma \ref{Lemma: C1,S1 in terms of a}. We first notice that the the second-order differential operator associated with the two-dimensional process $(X, \Phi)$ is such that
\begin{align}\label{Operator X Phi}
\mathcal{L}_{X,\Phi} f = \mu_0 \partial_x f + \frac12 \sigma^2 \partial_{xx}f  + \frac12 \gamma^2 \varphi^2 \partial_{\varphi \varphi}f  + \gamma \varphi \sigma \partial_{x \varphi}f, \qquad \forall f \in C^2 (\mathbb{R} \times (0,\infty) ),
\end{align} 
and apply Dynkin's formula to obtain
\begin{align}\label{Def ubar}
\overline{u} (x, \varphi) &:= \overline{v}(x,\varphi) - (e^x-\kappa)(1 + \varphi) \nonumber \\ 
&= \sup_\tau \mathbb{E}^\mathbb{Q}_{(x,\varphi)} \Big[ \int_0^\tau e^{-rt} \Big(  e^{X_t} (\mu_0 + \frac12 \sigma^2 -r) + r\kappa + \Phi_t \Big( e^{X_t} ( \mu_1 + \frac12 \sigma^2 -r) + r \kappa \Big) \Big) dt \Big].
\end{align}
For $x_2 > x_1$ and $\tau^* := \tau^* (x_2,\varphi)$ optimal for $ \overline{v} (x_2 , \varphi)$ we have
\begin{align*}
\overline{u}(x_1,\varphi) &- \overline{u}(x_2,\varphi) \\
&\geq
  \mathbb{E}^\mathbb{Q} \Big[ \int_0^{\tau^*} e^{-rt} \Big(  \big( e^{X_t^{x_2}} - e^{X_t^{x_1}} \big)  (r - \mu_0 - \frac12 \sigma^2 ) + \Phi_t \big(e^{X_t^{x_2}} - e^{X_t^{x_1}} \big)(r- \mu_1 + \frac12 \sigma^2 ) \Big) dt \Big]\\
& \geq 0,
\end{align*}
where the last inequality follows from $X^{x_2} \geq X^{x_1}$ $\mathbb{Q}$-a.s. and Assumption \ref{Assumption: Well-posedness}. 
Hence, for $(x_1, \varphi) \in \mathcal{S}_2$ and $x_2 > x_1$, we obtain $0 \leq \overline{u}(x_2 , \varphi) \leq \overline{u}(x_1, \varphi) = 0$ and the claim follows. \qed
\end{proof}
It is interesting to notice that there exists a one-to-one correspondence between the continuation regions $\mathcal{C}_1$ and $\mathcal{C}_2$ of \eqref{Continuation region C1} and \eqref{continuation after girsanov} as well as the stopping regions $\mathcal{S}_1$ and $\mathcal{S}_2$ of \eqref{Stopping region S1} and \eqref{stopping after girsanov}. Indeed, introducing the diffeomorphism  
\begin{align}\label{Transformation T}
T:= (T_1,T_2): \mathbb{R} \times (0,1) \to \mathbb{R} \times (0, \infty), \qquad (T_1(x , \pi) , T_2(x, \pi)) := \Big( x , \frac{\pi}{1-\pi} \Big),
\end{align}
with inverse 
\begin{align*}
T^{-1} (x, \varphi) := \Big( x , \frac{\varphi}{1 + \varphi} \Big), \qquad (x,\varphi) \in \mathbb{R} \times (0,\infty),
\end{align*}
one has
\begin{align*}
\mathcal{C}_2 = T ( \mathcal{C}_1 ) \qquad \text{as well as } \qquad \mathcal{S}_2 = T(\mathcal{S}_1). 
\end{align*}
Furthermore, upon using Lemma \ref{Lemma: C1,S1 in terms of a} and Lemma \ref{Lemma C2,S2 in terms of b}, we find that 
\begin{align}\label{Boundaries a b Transformation}
b(\varphi) = a \Big( \frac{\varphi}{1+\varphi} \Big) .
\end{align}
Due to this explicit relationship between the optimal stopping boundaries, we obtain some first results  on $b$ thanks to Lemma \ref{Lemma: Properties of boundary a}.
\begin{lemma}\label{Lemma: Properties boundary b}
The boundary $b(\varphi)$ of (\ref{boundary b}) is such that 
\begin{itemize}
\item[i)] $\varphi \mapsto b(\varphi)$ is nondecreasing on $(0,\infty)$;
\item[ii)] $\varphi \mapsto b(\varphi)$ is left-continuous;
\item[iii)] $b$ is bounded by  $x_0^* \leq b(\varphi)\leq x_1^*$ for all $\varphi \in (0,\infty)$, with $x_0^*$ and $x_1^*$ as in Lemma \ref{Lemma: Properties of boundary a}. 
\end{itemize}
\end{lemma}
 The relationship \eqref{Boundaries a b Transformation} and the transformation \eqref{Transformation T} allow us to trace back our results from this section - as well as from the following section - to the initial optimal stopping problem \eqref{Value fct of opt stopp}. Moreover, \eqref{Boundaries a b Transformation} turns out to be valuable in the proof of Lemma \ref{Lemma: Properties boundary b}, since proving the monotonicity result i) as well as the boundedness iii) is not straightforward without exploiting the relation between $b$ and $a$ and the results of Lemma \ref{Lemma: Properties of boundary a}.
\section{A Parabolic Formulation} \label{Section: Parabolic Formulation}
Observe that the dynamics of the processes $X$ and $\Phi$ in \eqref{Dynamics after Girsanov} are driven by the same Brownian motion. In order to account for this degeneracy, we pass yet to another formulation of the optimal stopping problem. To this end, we rely on a transformation that reveals the true parabolic nature of the generator $\mathcal{L}_{X,\Phi}$ as in \eqref{Operator X Phi}; i.e.~that poses it in its canonical form (cf.~Folland \cite{Fol}). Define 
\begin{align}\label{Transformation Tbar}
\overline{T} := (\overline{T}_1 , \overline{T}_2) : \mathbb{R} \times (0,\infty) \to \mathbb{R}^2, \quad (\overline{T}_1 (x, \varphi) , \overline{T}_2 (x, \varphi)) := \Big( x, \frac{\sigma}{\gamma} \ln (\varphi) -x \Big), 
\end{align}
for any $(x,\varphi) \in \mathbb{R} \times (0,\infty)$, which is a diffeomorphism with inverse given by 
\begin{align}\label{Transformation Tbar Inverse}
\overline{T}^{-1} (x,z) := \Big(x, e^{\frac{\gamma}{\sigma} (x+z) } \Big), \qquad (x,z) \in \mathbb{R}^2 .
\end{align}
With regard to the transformation \eqref{Transformation Tbar} we can introduce the process
\begin{align}\label{Process Z}
Z_t = \frac{\sigma}{\gamma} \ln (\Phi_t) - X_t, \qquad t \geq 0,
\end{align}
and an application of It\^{o}'s formula reveals that its dynamics are given by 
\begin{align}\label{Dynamics Z}
dZ_t &= -\frac12 (\mu_1 + \mu_0) dt, \quad \qquad Z_0 = z := \frac{\sigma}{\gamma} \ln (\varphi) - x.
\end{align}
Furthermore, we can define the transformed version of the value function $\overline{v}$ of \eqref{value fct opt after girsanov} via
\begin{align}\label{Value function parabolic formulation}
\hat{v}(x,z) := \overline{v} \big(x, e^{\frac{\gamma}{\sigma}(x+z)} \big) = \sup_\tau \mathbb{E}^\mathbb{Q}_{(x,z)} \big[ e^{-r \tau} (e^{X_\tau} - \kappa) (1 + e^{ \frac{\gamma}{\sigma} (X_\tau + Z_\tau)} ) \big],   
\end{align}
for $(x,z) \in \mathbb{R}^2$ and where now $\mathbb{E}^\mathbb{Q}_{(x,z)} [\cdot ] = \mathbb{E}^\mathbb{Q} [ \cdot \vert X_0 = x, Z_0 =z ]$. In light of this explicit relationship between the value functions $\overline{v}$ and $\hat{v}$, we can conclude the following result from Lemma \ref{Lemma: Properties of vbar}.
\begin{lemma}\label{Lemma: vhat continuous}
The value function $\hat{v}(x,z)$ of \eqref{Value function parabolic formulation} is locally Lipschitz continuous over $\mathbb{R}^2$.
\end{lemma}
The associated continuation and stopping region are given by 
\begin{align}\label{Continuation region C3}
\mathcal{C}_3 &:= \{ (x,z) \in \mathbb{R}^2: \quad \hat{v}(x,z) > (e^x-\kappa)(1 + e^{\frac{\gamma}{\sigma}(x+z)}) \}, \\
\label{Stopping region S3}
\mathcal{S}_3 &:= \{ (x,z) \in \mathbb{R}^2: \quad \hat{v}(x,z) = (e^x-\kappa)(1 + e^{\frac{\gamma}{\sigma}(x+z)}) \}, 
\end{align}
where $\mathcal{C}_3$ is open and $\mathcal{S}_3$ is closed. Furthermore, the global diffeomorphism  \eqref{Transformation Tbar} implies that $\mathcal{C}_3 = \overline{T}(\mathcal{C}_2)$ as well as $\mathcal{S}_3 = \overline{T}(\mathcal{S}_2)$, with $\mathcal{C}_2$ and $\mathcal{S}_2$ as in \eqref{continuation after girsanov}-\eqref{stopping after girsanov}. Notice that the second-order infinitesimal generator associated to the process $(X, Z)$ is now such that 
\begin{align}\label{Operator L X Z}
\mathcal{L}_{X,Z}f = \mu_0 \partial_xf + \frac12 \sigma^2 \partial_{xx}f - \frac12 (\mu_1 + \mu_0) \partial_z f, \qquad \forall f \in C^{2,1} (\mathbb{R}^2). 
\end{align}
We can rely on standard arguments from classical PDE theory as well as optimal stopping theory (see, e.g., Karatzas and Shreve \cite{KaSh2}, Section 2.7, Th. 7.7) and obtain the following lemma.
\begin{lemma}\label{Lemma: vhat solves boundary value problem}
The value function $\widehat{v}$ of \eqref{v and vbar equivalence} is the unique classical $C^{2,1}$-solution to the boundary value problem
\begin{align}\label{Boundary Value Problem vbar}
(\mathcal{L}_{X,Z} -r )w = 0 \quad \text{in }\mathcal{R} \qquad \text{and} \quad w \vert_{\partial \mathcal{R}} = \hat{v} \vert_{\partial \mathcal{R}},
\end{align}
for $\mathcal{L}_{X,Z}$ as in \eqref{Operator L X Z} and any open set $\mathcal{R}$ such that its closure is contained in the continuation region $\mathcal{C}_3$ of \eqref{Continuation region C3}. In particular, $\hat{v} \in C^{2,1} (\mathcal{C}_3)$.  
\end{lemma} 
In the following, we aim at investigating the geometry of the state space in the coordinates $(X,Z)$. To this end, we define the generalised inverse of the nondecreasing boundary $b$ by
\begin{align}\label{Inverse of b}
b^{-1} (x) := \inf \{ \varphi \in (0,\infty): ~ b(\varphi) > x \} ,
\end{align}
such that the continuation region $\mathcal{C}_2$ of \eqref{continuation after girsanov} rewrites as
\begin{align}\label{Continuation C2 in terms of b-1}
\mathcal{C}_2 = \{ (x,\varphi) \in \mathbb{R} \times (0,\infty):~ b^{-1} (x) < \varphi \}.
\end{align}
Since $\varphi \mapsto b(\varphi)$ is nondecreasing by Lemma \ref{Lemma: Properties boundary b}, we observe that 
\begin{align*}
(x , z) \in \mathcal{C}_3 ~ \Longleftrightarrow ~ (x, e^{\frac{\gamma}{\sigma}(x+z)}) \in \mathcal{C}_2 ~ \Longleftrightarrow ~ e^{\frac{\gamma}{\sigma}(x+z)} > b^{-1} (x) ~ \Longleftrightarrow ~ z > \frac{\sigma}{\gamma} \log (b^{-1} (x)) -x ,
\end{align*}
and by setting 
\begin{align}\label{Boundary c-1}
c^{-1} (x) := \frac{\sigma}{\gamma} \log (b^{-1} (x)) -x ,
\end{align}
we can rewrite \eqref{Continuation region C3} and \eqref{Stopping region S3} as
\begin{align}\label{C3 and S3 in terms of c-1}
\mathcal{C}_3 = \{ (x,z) \in \mathbb{R}^2:~z> c^{-1}(x) \} , \qquad \quad \mathcal{S}_3 = \{ (x,z) \in \mathbb{R}^2:~z \leq c^{-1}(x) \}.
\end{align}
In contrast to the optimal stopping problems in the formulations \eqref{Value fct of opt stopp} and \eqref{value fct opt after girsanov}, deriving the monotonicity of the boundary $x \mapsto c^{-1}(x)$ is not straightforward. Moreover - and  differently to related contributions such as Federico et al.\,\cite{FeFe} - we cannot trace it back to the monotonicity of the boundary $b$ of \eqref{boundary b}, since its generalised inverse $b^{-1}$ is nondecreasing as well, and this does not imply monotonicity of $x \mapsto c^{-1} (x)$. To this end, we follow and adapt arguments presented in Section 4.4 of De Angelis \cite{DeA}, which studies separately the two cases in which the deterministic process $Z$ as in \eqref{Dynamics Z} is either increasing ($\mu_0 + \mu_1 \geq 0$) or decreasing ($\mu_0 + \mu_1 < 0$).
\\
For the following analysis, it is useful to define
\begin{align}\label{u hat}
\hat{u}(x,z) := \hat{v}(x,z) - (e^x-\kappa)(1+e^{\frac{\gamma}{\sigma}(x+z)}) ,
\end{align} 
as well as 
\begin{align}\label{Definition g(x,z)}
g(x,z) &:= (\mathcal{L}_{X,Z} -r) \big(( e^x - \kappa)(1 + e^{\frac{\gamma}{\sigma}(x+z)} ) \big) \nonumber \\
&= e^x \big( \frac12 \sigma^2 + \mu_0 -r \big) + r \kappa + e^{\frac{\gamma}{\sigma}(x+z) } \big( e^x \big( \frac{1}{2}\sigma^2 + \mu_1 - r \big) + r \kappa \big) ,
\end{align}
and we observe that an application of Dynkin's formula implies 
\begin{align}\label{Dynkin for vhat}
\hat{u}(x,z) = \sup_\tau \mathbb{E}^\mathbb{Q}_{(x,z)} \Big[ \int_0^\tau  e^{-rt} g(X_t,Z_t) dt \Big], \qquad (x,z) \in \mathbb{R}^2 .
 \end{align}
\begin{proposition} \label{Proposition Mono 1}
Let $\mu_0 + \mu_1 \geq 0$. Then there exists a nondecreasing function $c: \mathbb{R} \to \mathbb{R}$ such that the continuation region $\mathcal{C}_3$ of \eqref{Continuation region C3} rewrites as 
\begin{align}\label{Continuation region C3 in terms of c}
\mathcal{C}_3 = \{ (x,z) \in \mathbb{R}^2:~ x < c(z) \} .
\end{align}
\end{proposition}
\begin{proof} Let $(x_0,z_0) \in \mathcal{S}_3$, $x_1 > x_0$ and notice that \eqref{C3 and S3 in terms of c-1} implies  $(-\infty , z_0] \times \{x_0 \} \in \mathcal{S}_3$. Furthermore, we have $x_0 > x_0^*$ and since the process $Z$ is decreasing, we observe that the process $(X^{x_1} , Z^{z_0})$ crosses the half-line $(- \infty , z_0] \times \{x_0\}$ before reaching the level $x_0^*$. Hence, we have $\mathbb{Q}_{x_1 , z_0 }[\tau^* < \tau_{x_0^*}] =1$, where $\tau_{x_0^*} :=\inf \{t \geq 0: \,\, X_t^{x_1} = x_0^* \}$ and $\mathbb{Q}_{x_1, z_0} [\cdot] = \mathbb{Q} [\cdot \vert X_0 = x_1 , Z_0 = z_0 ])$. Moreover, it can be verified that the second condition of Assumption \ref{Assumption: Well-posedness} implies $x_0^* > \tilde{x}$, with the latter given by 
\begin{align} \label{x tilde for monotonicity}
\tilde{x} := \log \Big( \frac{r \kappa}{r - \frac12 \sigma^2 - \mu_1} \Big).
\end{align}
Consequently, we have $\exp (X_s^{x_1}) (r - \frac12 \sigma^2 - \mu_1) > r \kappa$ for all $s \in [0,\tau^*)$ and \eqref{Definition g(x,z)}-\eqref{Dynkin for vhat} imply $\hat{u}(x_1 , z_0) \leq 0$ for all $x_1 > x_0$, and therefore $\{z_0\} \times [x_0, \infty) \in \mathcal{S}_3$. We can thus define
\begin{align}\label{Boundary c}
c(z):= \inf \{ x \in \mathbb{R}: ~ (x,z) \in \mathcal{S}_3 \}.
\end{align}
and observe that \eqref{C3 and S3 in terms of c-1} implies that $z \mapsto c(z)$ is nondecreasing. \qed
\end{proof}
In order to establish the same result in the case when $\mu_0 + \mu_1 < 0$, we first state the following lemma. 
\begin{lemma}\label{Lemma: hat v_z}
We have
\begin{align}\label{deriv wrt z}
\hat{v}_z (x,z) = \mathbb{E}^\mathbb{Q}_{(x,z)} \left[ \frac{\gamma}{\sigma} e^{-r \tau^*} (e^{X_{\tau^*}} - \kappa) e^{\frac{\gamma}{\sigma} (X_{\tau^*} + Z_{\tau^*} )} \one_{\{ \tau^* < \infty \} } \right],
\end{align}
for all $(x,z) \in \mathbb{R}^2 \setminus \partial \mathcal{C}_3$ and $\tau^* := \tau^* (x,z)$. 
\end{lemma}
\begin{proof}
For $(x,z) \in \mathcal{S}_3$ the claim follows immediately, since $\mathbb{Q}_{(x,z)} [\tau^* = 0] =1$. Hence, we let $(x,z) \in \mathcal{C}_3$ and for $\epsilon > 0$ we obtain
\begin{align}\label{Inequality for Der vhat z}
\hat{v}(x, z + \epsilon) - \hat{v}(x,z) 
&\geq \mathbb{E}^\mathbb{Q} \Big[ e^{-r (\tau^* \wedge t)} \big( \hat{v}(X_{\tau^* \wedge t}^x , Z_{\tau^* \wedge t}^{z + \epsilon} ) - \hat{v} ( X_{\tau^* \wedge t}^x , Z_{\tau^* \wedge t}^z ) \big) \Big] \nonumber \\
&\geq \mathbb{E}^\mathbb{Q} \Big[ e^{-r \tau^*} \big( e^{ X_{\tau^*}^x } - \kappa \big) e^{\frac{\gamma}{\sigma} X_{\tau^*}^x }  \big( e^{\frac{\gamma}{\sigma} Z_{\tau^* }^{z + \epsilon}) } - e^{\frac{\gamma}{\sigma} Z_{\tau^* }^{z}) } \big) \one_{ \{ \tau^* < t \} } \Big] \nonumber \\ 
&\hspace*{3cm}+ \mathbb{E}^\mathbb{Q} \Big[ e^{-r  t} \big( \hat{v}(X_{t}^x , Z_{t}^{z + \epsilon} ) - \hat{v} ( X_{t}^x , Z_{ t}^z ) \big) \one_{ \{ \tau^* > t \} } \Big],
\end{align}
where the first inequality follows from the supermartingale property of $\big( e^{-r (\tau \wedge t)} \hat{v}(X_{\tau \wedge t}^x , Z_{\tau \wedge t}^{z + \epsilon}) \big)_t$ and the martingale property of $\big( e^{-r (\tau^* \wedge t)} \hat{v}(X_{\tau^* \wedge t}^x , Z_{\tau^* \wedge t}^z )\big)_t$ for $\tau^* := \tau^* (x,z)$. Upon employing a change of measure as in Section \ref{Section: Decoupling Change of Measure}, we find 
\begin{align*}
\mathbb{E}^\mathbb{Q}_{(x,z)} \big[ e^{-rt} \vert \hat{v}(X_t , Z_t ) \vert \big] 
&\leq \mathbb{E}^\mathbb{Q}_{(x,z)} \big[ e^{-rt} \vert \overline{v}(x , e^{\frac{\gamma}{\sigma} (X_t + Z_t )} ) \vert \big] 
\leq K_1 \mathbb{E}^\mathbb{Q}_{(x, \exp (\frac{\gamma}{\sigma}(x+z))} \big[ e^{-rt} e^{X_t} (1 + \Phi_t) \big] \\
&= K_1 (1 + e^{\frac{\gamma}{\sigma}(x+z)} ) \mathbb{E}_{(x, \pi)} \big[ e^{-rt} e^{X_t} \big], 
\end{align*}
where $\pi = e^{\frac{\gamma}{\sigma}(x+z)} /(1 + e^{\frac{\gamma}{\sigma}(x+z)} )$. It is then easy to verify that Assumption \ref{Assumption: Well-posedness} implies 
\begin{align*}
\lim_{t \uparrow \infty} \mathbb{E}^\mathbb{Q}_{(x,z)} \big[  e^{-rt}  \hat{v}(X_t , Z_t )  \big] = 0,
\end{align*}
and hence, applying dominated convergence in \eqref{Inequality for Der vhat z} as $t \uparrow \infty$ yields
\begin{align}\label{Inequ vhat_z 1}
\hat{v}(x, z + \epsilon) - \hat{v}(x,z) 
&\geq \mathbb{E}^\mathbb{Q} \Big[ e^{-r \tau^*} \big( e^{ X_{\tau^*}^x } - \kappa \big) e^{\frac{\gamma}{\sigma} X_{\tau^*}^x }  \big( e^{\frac{\gamma}{\sigma} Z_{\tau^* }^{z + \epsilon} } - e^{\frac{\gamma}{\sigma} Z_{\tau^* }^{z} } \big) \one_{ \{ \tau^* < \infty \} } \Big]. 
\end{align}
Similar arguments show 
\begin{align} \label{Inequ vhat_z 2}
\hat{v}(x, z) - \hat{v}(x,z - \epsilon) 
&\leq \mathbb{E}^\mathbb{Q} \Big[ e^{-r \tau^*} \big( e^{ X_{\tau^*}^x } - \kappa \big) e^{\frac{\gamma}{\sigma} X_{\tau^*}^x }  \big( e^{\frac{\gamma}{\sigma} Z_{\tau^* }^{z + \epsilon} } - e^{\frac{\gamma}{\sigma} Z_{\tau^* }^{z} } \big) \one_{ \{ \tau^* < \infty \} } \Big],
\end{align}
and since $\hat{v} \in C^{2,1} (\mathcal{C}_3)$ (cf. Lemma \ref{Lemma: vhat solves boundary value problem}), dividing \eqref{Inequ vhat_z 1} and \eqref{Inequ vhat_z 2} by $\epsilon$ and letting $\epsilon \downarrow 0$, we obtain the desired result. \qed
\end{proof}
\begin{proposition}\label{Proposition Mono 2}
Let $\mu_0 + \mu_1 < 0$. There exists a nondecreasing function $c:\mathbb{R} \to \mathbb{R}$ such that the continuation region of \eqref{Continuation region C3} can be written as
\begin{align}\label{Continuation Region C3 in terms of c 2}
\mathcal{C}_3 = \{ (x,z) \in \mathbb{R}^2:~ x < c(z) \} .
\end{align}
\end{proposition}
\begin{proof}
Let $(x,z) \in \mathbb{R}^2$. Notice that $x < x_0^*$ implies $(x',z) \in \mathcal{C}_3 $ for all $x' < x$ and $z \in \mathbb{R}$, because of Lemma \ref{Lemma: Properties of boundary a} and since the transformations $T_1$ and $\overline{T}_1$ of  \eqref{Transformation T} and \eqref{Transformation Tbar}, respectively, are the identity; hence, $\{ (x,z): ~ x < x_0^* \} \subset \mathcal{C}_3$. We can thus focus on the case that $x \geq x_0^*$ and distinguish two possibilities: 
\begin{itemize}
\item[i)] $\hat{u}_x (x,z) \leq 0 \quad \forall x \in (x_0^*,\infty)$ such that $(x,z) \in \mathcal{C}_3$; 
\item[ii)] $\exists\, x_0 \in \mathbb{R}$, $x_0 > x_0^*$ such that $(x_0,z) \in \mathcal{C}_3$ and $\hat{u}_x (x_0 ,z) > 0$. 
\end{itemize}
In case i), the map $x \mapsto \hat{u}(x,z)$ is decreasing for $x \in (x_0^*, \infty)$ and $(x,z) \in \mathcal{C}_3$. Hence, for any $(x,z)$ in the latter region we obtain $(-\infty,x] \times \{z\} \in \mathcal{C}_3$ and the claim follows in the same spirit as in Proposition \ref{Proposition Mono 1}. In case ii), we establish a contradiction scheme. As a first step, we show that ii) implies $[x_0, \infty) \times \{ z \} \in \mathcal{C}_3$, which will then lead to a contradiction.  We start by noticing that Lemma \ref{Lemma: vhat solves boundary value problem} and \eqref{Dynkin for vhat} imply
\begin{align}\label{L -r u hat}
(\mathcal{L}_{X,Z} - r )\hat{u} (x_0,z) =  - g(x_0,z),
\end{align} 
for $(x_0,z)$ as given in ii) above. Since $\mu_0 < 0$ and $\hat{u}_x (x_0,z) > 0$ we have $\mu_0 \hat{u}(x_0,z) < 0$, and thus 
\begin{align}\label{second deriv}
\frac12 \sigma^2 \hat{u}_{xx} (x_0,z) &= r \hat{u} (x_0,z) - \mu_0 \hat{u}_x (x_0,z) + \frac12 (\mu_0 + \mu_1) \hat{u}_z (x_0,z) - g(x_0,z) \\
&> r \hat{u} (x_0,z)  + \frac12 (\mu_0 + \mu_1) \hat{u}_z (x_0,z) - g(x_0,z). \nonumber 
\end{align} 
Next, we notice that we can rewrite (\ref{deriv wrt z}) as 
\begin{align}\label{der of v wrt z}
\hat{v}_z(x,z) = \frac{\gamma}{\sigma} \Big( \hat{v}(x,z) - \mathbb{E}^\mathbb{Q} [e^{-r \tau^*} (e^{X_{\tau^*}^x} - \kappa)] \Big),
\end{align}
and since 
\begin{align*}
\hat{v}_z (x,z) = \hat{u}_z (x,z) + \frac{\gamma}{\sigma} (e^x-\kappa)e^{\frac{\gamma}{\sigma}(x+z)} \quad \text{and} \quad \hat{v}(x,z) = \hat{u}(x,z) + (e^x -\kappa)(1 + e^{\frac{\gamma}{\sigma}(x+z)}),
\end{align*}
(\ref{der of v wrt z}) gives
\begin{align*}
\hat{u}_z (x,z) +\frac{\gamma}{\sigma} (e^x-\kappa)e^{\frac{\gamma}{\sigma}(x+z)} = \frac{\gamma}{\sigma} \Big( \hat{u}(x,z) + (e^x- \kappa)(1 + e^{\frac{\gamma}{\sigma}(x+z)}) - \mathbb{E}^\mathbb{Q}[e^{-r \tau^*} (e^{X_{\tau^*}^x} - \kappa)] \Big), 
\end{align*}
which is equivalent to 
\begin{align*}
\hat{u}_z(x,z) = \frac{\gamma}{\sigma} \hat{u}(x,z) + \frac{\gamma}{\sigma} \Big( e^x- \kappa - \mathbb{E}^\mathbb{Q} [e^{-r \tau^*} (e^{X_{\tau^*}^x} - \kappa)] \Big).
\end{align*}
We can thus plug this last equality into (\ref{second deriv}) and obtain
\begin{align*}
\frac12 \sigma^2 &\hat{u}_{xx} (x_0,z) \\
 &> r \hat{u} (x_0,z)  + \frac12 (\mu_0 + \mu_1) \Big( \frac{\gamma}{\sigma} \hat{u}(x_0,z) + \frac{\gamma}{\sigma} \Big( e^{x_0} -\kappa - \mathbb{E}^\mathbb{Q} [e^{-r \tau^*} (e^{X_{\tau^*}^{x_0}} - \kappa)] \Big) \Big) - g(x_0,z) \\
&= \Big( r + \frac12 (\mu_0 + \mu_1) \frac{\gamma}{\sigma} \Big) ( \hat{u}(x_0 , z) + e^{x_0} - \kappa) - \frac{1}{2}(\mu_0 + \mu_1) \frac{\gamma}{\sigma}\mathbb{E}^\mathbb{Q} [e^{-r \tau^*} (e^{X_{\tau^*}^{x_0}} - \kappa)] - g(x_0,z) \\
&>0,
\end{align*}
where the last inequality follows precisely from $r > \frac{\gamma}{2 \sigma} \vert \mu_0 + \mu_1 \vert$ in Assumption \ref{Assumption: Well-posedness}, upon noticing that $x_0 > x_0^*$. We deduce that $\hat{u}_x (\cdot \, , z)$ increases in a right-neighbourhood of $x_0$ and repeating arguments for every $x > x_0$ yields $\hat{u}_x (\cdot \, , z) > 0$ on $[x_0 , \infty)$. It follows that $\hat{u}(\cdot \, , z)$ is increasing on $[x_0 , \infty)$ such that $[x_0 , \infty) \times \{z\} \in \mathcal{C}_3$ and (combining the latter with \eqref{C3 and S3 in terms of c-1}) we have $\mathcal{A} := [x_0 , \infty) \times [z_0 , \infty) \subset \mathcal{C}_3$. However, this leads to a contradiction. To see this, let $(x,z) \in \mathcal{A}$  and define $\tau_{x_0} := \inf \{ t>0: ~ X_t^x \leq x_0 \}$. Since $t \mapsto Z_t^z$ is increasing, the only possibility for the process $(X^x,Z^z)$ to exit $\mathcal{A}$ and thus eventually the continuation region, is by passing through the horizontal line $[x_0,\infty) \times \{z_0\}$. We thus have $\tau_{x_0} \leq \tau^*$ $\mathbb{Q}_{(x,z)}$-a.s.~and moreover, since $\mu_0< 0$, the stopping time $\tau_{x_0}$ is finite a.s. Upon using Lemma \ref{Lemma: Properties of vbar} i) and \eqref{Value function parabolic formulation}, it follows that 
\begin{align*}
(e^x - \kappa) (1 + e^{\frac{\gamma}{\sigma}(x+z)} ) 
< \hat{v}(x,z) 
&= \mathbb{E}^\mathbb{Q}_{(x,z)} \Big[ e^{-r \tau_{x_0}} \hat{v}(X_{\tau_{x_0}} ,Z_{\tau_{x_0}}) \Big] \\
&= \mathbb{E}^\mathbb{Q}_{(x,z)} \Big[ e^{-r \tau_{x_0}} \hat{v}(x_0 , z - \frac{1}{2} (\mu_0 + \mu_1) \tau_{x_0} ) \Big] \\
&\leq K_1 e^{x_0} \mathbb{E}^\mathbb{Q}_{(x,z)} \Big[ e^{-r \tau_{x_0}} \Big] + K_1 e^{ \frac{\gamma}{\sigma}( x_0 + z)} e^{x_0} \mathbb{E}^\mathbb{Q}_{(x,z)} \Big[  e^{- (r - \frac12 \frac{\gamma}{\sigma} \vert \mu_0 + \mu_1 \vert ) \tau_{x_0} } \Big] .
\end{align*}
Let now $\hat{r} := r- \frac{\gamma}{2 \sigma} \vert \mu_0 + \mu_1 \vert > 0$ and denote $\phi_r$ (resp.~$\phi_{\hat{r}}$) the strictly decreasing solution to $\frac12 \sigma^2 f_{xx} + \mu_0 f_x - q f = 0$, for $q \in \{r , \hat{r} \}$. Then, by results on hitting times for one-dimensional diffusions (see, e.g., Borodin and Salminen \cite{BoSa}, Ch.~\RomanNumeralCaps{2}), the above inequality is equivalent to 
\begin{align}\label{ineq monot 2}
(e^x - \kappa) (1 + e^{\frac{\gamma}{\sigma}(x+z) } ) \leq K_1 e^{x_0} \frac{\phi_r(x)}{\phi_r (x_0)} + K_1 e^{x_0} e^{\frac{\gamma}{\sigma}(x_0 +z) } \frac{\phi_{\hat{r}} (x)}{\phi_{\hat{r}} (x_0) } ,
\end{align}
which thus holds true for all $(x,z) \in \mathcal{A}$. Since $\mathcal{A}$ is right-connected, we can let $x \to \infty$ and notice that $(e^x -\kappa)(1 + e^{\frac{\gamma}{\sigma}(x+z)}) \to \infty$, while the right hand side of (\ref{ineq monot 2}) decreases to $0$ due to the decreasing property of $x \to \phi_q (x)$ for $q$ positive. We thus obtain a contradiction, which concludes our proof. \qed
\end{proof}
\begin{Remark} 
Notice that Propositions \ref{Proposition Mono 1} and \ref{Proposition Mono 2} imply that the function $x \mapsto c^{-1} (x)$ of \eqref{Boundary c-1} is nondecreasing as well. Moreover, we notice that 
\begin{align}
z > c^{-1}(x) ~ \Longleftrightarrow ~ c(z) > x,
\end{align}
and hence, the function $c^{-1}$ is the right-continuous inverse of $c$ and thus admits the representation
\begin{align}\label{Boundary c-1 inverse of c}
c^{-1} (x) = \inf \{ z \in \mathbb{R}:~ c(z) > x \}.
\end{align}
In light of the connection \eqref{Boundary c-1} between $c^{-1}$ and $b^{-1}$ (the generalised inverse of the boundary $b$), equation \eqref{Boundary c-1 inverse of c} allows us to trace back our results to the formulation of Section \ref{Section: Decoupling Change of Measure} and then - through the representation \eqref{Boundaries a b Transformation} - to the original setting of Section \ref{Section: First related OSP}. 
\end{Remark}
\hspace{-0.55cm}6.1 \textbf{Regularity of the value function and of the optimal stopping boundary.}
We established the existence of a nondecreasing boundary $z \mapsto c(z)$, such that $\mathbb{R}^2$ is split into the continuation region $\mathcal{C}_3$ of \eqref{Continuation region C3} and the stopping region $\mathcal{S}_3$ of \eqref{Stopping region S3}. In the following, we derive some further properties of the optimal stopping boundary and of the value function $\hat{v}$ of \eqref{Value function parabolic formulation}. We first state the following result, which will be helpful in the forthcoming analysis. 
\begin{lemma}\label{Lemma: uhat z geq 0 in C}
We have $\hat{u}_z (x,z) \geq 0$ for $(x,z) \in \mathcal{C}_3$.
\end{lemma} 
\begin{proof}
Because of \eqref{Value function parabolic formulation} and \eqref{Transformation Tbar}, we have that $\overline{v}(x, \varphi)$ as in \eqref{value fct opt after girsanov} is such that $\overline{v}(x, \varphi) = \hat{v}(x, \frac{\sigma}{\gamma} \ln (\varphi) - x)$ for all $(x,\varphi) \in \mathbb{R} \times (0, \infty)$. Since $\hat{v}_z \in C^0 (\mathcal{C}_3)$ by Lemma \ref{Lemma: vhat solves boundary value problem}, we then also have $\overline{v}_\varphi \in C^0 (\mathcal{C}_2)$. Furthermore, $\varphi \mapsto \overline{v} (x, \varphi)$ is convex on $(0,\infty)$ by Lemma \ref{Lemma: Properties of vbar} iv) and thus also $\varphi \mapsto \overline{u}(x, \varphi)$ of \eqref{Def ubar}. Then, for $(x,\varphi) \in \mathcal{C}_2$ and $\varphi' = b^{-1}(x)$ such that $(x, \varphi') \in \partial \mathcal{C}_2$, we obtain (as $\overline{u}_\varphi \in C^0 (\mathcal{C}_2)$ as well)
\begin{align*}
0 \leq \overline{u}(x, \varphi) = \overline{u} (x, \varphi)  -  \overline{u} (x, b^{-1} (x)) \leq \overline{u}_\varphi (x, \varphi) (\varphi - b^{-1} (x)),
\end{align*}
and $\varphi > b^{-1} (x)$ implies $\overline{u}_\varphi (x, \varphi) \geq 0$ for $(x,\varphi) \in \mathcal{C}_2$. In light of the relation \eqref{Value function parabolic formulation} we then obtain $\hat{u}_z (x,z) \geq 0$ on $\mathcal{C}_3$. \qed
\end{proof}

\begin{proposition}\label{Proposition boundary c continuous}
The optimal stopping boundary $c(z)$ is such that $x_0^* \leq c(z) \leq x_1^*$ for all $z \in \mathbb{R}$ and with $x_0^*$ and $x_1^*$ as in Lemma \ref{Lemma: Properties of boundary a}. Furthermore, we have $c \in C(\mathbb{R})$. 
\end{proposition}
\begin{proof}
The first part of the claim follows from Lemma \ref{Lemma: Properties boundary b} iii) and by noticing that the transformation $\overline{T}_1$ of \eqref{Transformation Tbar} is the identity. We derive the continuity of $z \mapsto c(z) $ in two steps. \\[0.1cm]
1) \emph{Left-Continuity:} Let $z_0 \in \mathbb{R}$ and $z_n \uparrow z_0$ as $n \to \infty$. Since $z \mapsto c(z)$ is nondecreasing and $\mathcal{S}_3$ is closed, we obtain $\lim_{n \to \infty} (c(z_n), z_n) = (c(z_0 -), z_0 ) \in \mathcal{S}_3$, where $c(z_0 -)$ denotes the left limit of $c$ at $z_0$.  The definition of $c$ in \eqref{Boundary c} implies $c(z_0 - ) \geq c(z_0)$, but since $c$ is nondecreasing, we must have $c(z_0 -) = c(z_0)$ and the claim follows. \\[0.1cm]
2) \emph{Right-Continuity:} We argue by contradiction  and assume there exists $z_0 \in \mathbb{R}$ s.t. $c(z_0) < c(z_0 +)$. Using techniques developed in De Angelis \cite{DeACont}, we take $c(z_0) < x_1 < x_2 < c(z_0 +)$ and a nonnegative function $\phi \in C^\infty_c (x_1,x_2)$ such that $\int_{x_1}^{x_2} \phi (x) dx =1$. Recalling \eqref{L -r u hat}, we have 
\begin{align}\label{L -r uhat right conti}
\mathcal{L}_{X,Z} \hat{u}(x,z) - r \hat{u} (x,z)
= - g(x,z),
\end{align}
for $(x,z) \in (x_1,x_2)\times (z_0, \infty)$. In the following, it is helpful to treat the cases i) $\mu_0 + \mu_1 \geq 0$ and ii) $\mu_0 + \mu_1 < 0$ separately.  Let us start with i) and recall that $\hat{u}_z (x,z) \geq 0$ for $x$ and $z$ as above, due to Lemma \ref{Lemma: uhat z geq 0 in C}. Integration by parts reveals 
\begin{align*}
0 & \geq - \frac12 (\mu_0 + \mu_1) \int_{x_1}^{x_2} \hat{u}_z (x,z) \phi (x) dx \\
&= \int_{x_1}^{x_2} \Big( r \hat{u}(x,z) - \mu_0 \hat{u}_x (x,z) - \frac12 \sigma^2 \hat{u}_{xx} (x,z) - g(x,z) \Big) \phi(x) dx  \\
&=\int_{x_1}^{x_2} \big( r \hat{u}(x,z) \phi(x) + \mu_0 \hat{u}(x,z) \phi' (x) - \frac12 \sigma^2 \hat{u}(x,z) \phi'' (x) - g(x,z) \phi(x) \big) dx.
\end{align*}
Hence, employing dominated convergence as $z \downarrow z_0$ and using $\hat{u}(x,z_0)=0$, yields
\begin{align}\label{contra}
0 \geq  - \int_{x_1}^{x_2} g(x,z) \phi(x) dx > 0,
\end{align}
where the latter inequality follows from $x_1,x_2 \geq x_0^*$ and Assumption \ref{Assumption: Well-posedness}, which implies $x > \tilde{x}$ for all $x \in [x_1, x_2]$ and $\tilde{x}$ as in \eqref{x tilde for monotonicity}. We thus obtain a contradiction and $c(z_0) = c(z_0 + )$.

In case ii), we rely on classical results of internal regularity of PDEs (cf.~Th.~10 in Chapter 3 of Friedman \cite{Fr}), which allow to take derivatives in (\ref{L -r uhat right conti}) with respect to $x$ and have $\hat{u}_x \in C^{2,1}(\mathcal{C}_3)$ solving
\begin{align*}
(\mathcal{L}_{X,Z} -r ) \hat{u}_x (x,z) = - g_x (x,z),  \qquad (x,z) \in (x_1,x_2) \times (z_0, \infty).  
\end{align*}
Then, for $z > z_0$ we obtain
\begin{align}\label{l-r ux - gx = 0 for right conti}
\int_{x_1}^{x_2} \big( (\mathcal{L}_{X,Z} -r )\hat{u}_x (x,z) + g_x (x,z) \big)  \phi(x) dx = 0.
\end{align}
Let $F_{\phi} (z) := \int_{x_1}^{x_2} \hat{u}_{xz} (x,z) \phi (x) dx$. Integration by parts allows to rewrite \eqref{l-r ux - gx = 0 for right conti} as
\begin{align*}
\frac12 \vert \mu_0 + \mu_1 \vert F_\phi (z) &= \int_{x_1}^{x_2} \Big( r \hat{u}_x (x,z) - \frac12 \sigma^2 \hat{u}_{xxx} (x,z)  - \mu_0 \hat{u}_{xx}(x,z) - g_x (x,z) \Big) \phi(x) dx \\
&= \int_{x_1}^{x_2} \big(-r \hat{u}(x,z) \phi ' (x) + \frac12 \sigma^2 \hat{u}(x,z)  \phi ''' (x) - \mu_0 \hat{u}(x,z)  \phi '' (x) - g_x (x,z) \phi(x) \big) dx, 
\end{align*}
and using dominated convergence as $z \downarrow z_0$ as well as $\hat{u}(x,z_0)=0$ results in 
\begin{align*}
F_\phi (z_0 +) &= \frac{2}{\vert \mu_0 + \mu_1 \vert} \int_{x_1}^{x_2}  - g_x (x,z_0) \phi(x) dx 
\geq p_0 > 0,
\end{align*}
for some $p_0$, where the second to last inequality again follows from Assumption \ref{Assumption: For monotonicity of c}. Thus, there exists $\epsilon > 0$ such that $F_\phi (z) \geq p_0/2$ for all $z\in (z_0,z_0 + \epsilon)$ and we finally obtain 
\begin{align*}
\frac{1}{2}p_0 \epsilon &\leq \int_{z_0}^{z_0 + \epsilon} F_\phi (z) dz = \int_{z_0}^{z_0 + \epsilon} \int_{x_1}^{x_2} \hat{u}_{xz}(x,z) \phi (x) dx dz 
= - \int_{x_1}^{x_2} \int_{z_0}^{z_0 + \epsilon} \hat{u}_z (x,z) \phi ' (x) dz dx \\
&= - \int_{x_1}^{x_2} (\hat{u} (x, z_0 + \epsilon) - \hat{u} (x,z_0) ) \phi' (x) dx =  \int_{x_1}^{x_2} \hat{u}_x (x, z_0 + \epsilon)  \phi (x) dx \leq 0,
\end{align*}
where we used $\hat{u}(x,z_0)=0$ as well as $\hat{u}_x (x,z) \leq 0$ for $x \in [x_1, x_2]$ and $z > z_0$ (cf. Proposition \ref{Proposition Mono 2}). Hence, $c(z) = c(z +)$ for all $z \in \mathbb{R}$ and together with 1) we conclude that $z \mapsto c(z)$ is continuous. \qed
\end{proof}
In the next step, we derive the regularity of the value function. Its proof can be found in Appendix \ref{Appendix: Proof of Smooth fit}. 
\begin{proposition}\label{Proposition: Smooth fit}
The value function $\hat{v}$ of \eqref{Value function parabolic formulation} satisfies $\hat{v} \in C^1 (\mathbb{R}^2)$ and $\hat{v}_{xx} \in L_{\text{loc}}^\infty (\mathbb{R}^2)$.
\end{proposition}
In light of Proposition \ref{Proposition: Smooth fit},  we are able to derive an integral equation for the free boundary $c$. Let us first recall that by standard arguments, based on the strong Markov property  and Proposition \ref{Proposition: Smooth fit}, the value function $\hat{v}$ and the free boundary $c$ solve the free-boundary problem
\begin{align}\label{Free Boundary Problem}
\begin{cases}
(\mathcal{L}_{X,Z} -r ) \hat{v}(x,z) \leq 0, &(x,z) \in \mathbb{R}^2 , \\
(\mathcal{L}_{X,Z} -r ) \hat{v}(x,z) = 0, & x < c(z), \, z \in \mathbb{R}, \\
\hat{v}(x,z) \geq (e^x - \kappa) (1 + e^{\frac{\gamma}{\sigma}(x+z) }), &(x,z) \in \mathbb{R}^2, \\
\hat{v}(x,z) = (e^x - \kappa) (1 + e^{\frac{\gamma}{\sigma}(x+z) }), & x \geq c(z),\, z \in \mathbb{R}, \\
\hat{v}_x (x,z) = e^x (1 + e^{\frac{\gamma}{\sigma}(x+z) } ) + \frac{\gamma}{\sigma} (e^x - \kappa) e^{\frac{\gamma}{\sigma}(x+z)}, & x =c(z),\, z\in \mathbb{R}, \\
\hat{v}_z (x,z) = \frac{\gamma}{\sigma} (e^x -\kappa) e^{\frac{\gamma}{\sigma} (x+z) }, & x=c(z) , \, z \in \mathbb{R}.
\end{cases}
\end{align}
In the next Proposition, upon using a suitable application of It\^{o}'s Lemma, we derive a probabilistic representation of the value function $\hat{v}$. Its proof is postponed to Appendix \ref{Appendix: Proof of Proposition Prob Repr of hatv}.
\begin{proposition}\label{Proposition: Prob Repr of hatv}
Recall the free boundary $c$ of \eqref{Boundary c} and the function $g$ of \eqref{Definition g(x,z)}. For any $(x,z) \in \mathbb{R}^2$, the value function $\hat{v}$ can be written as 
\begin{align}\label{Probabilistic Repr of vhat}
\hat{v}(x,z) & = \mathbb{E}^\mathbb{Q}_{(x,z)} \Big[ - \int_0^\infty e^{-rs}  g(X_s , Z_s) \one_{ \{ X_s \geq c(Z_s) \} } ds \Big].
\end{align}
\end{proposition}
Denote now by 
\begin{align}\label{Def Density Fct G}
G(w; m,v) := \frac{1}{\sqrt{2 \pi v^2}} e^{- \frac{(w-m)^2}{2v^2}}, \quad w \in \mathbb{R}, \, m \in \mathbb{R}, \, v >0,
\end{align}
the density function of a Gaussian random variable with mean $m$ and variance $v^2$. Then, from Proposition \ref{Proposition: Prob Repr of hatv} we obtain the following result.
\begin{proposition}\label{Proposition: Integral Equation for c}
Let 
\begin{align*}
\mathcal{M}:= \big\{ f: \mathbb{R} \mapsto \mathbb{R}: ~~ f ~\text{is nondecreasing, continuous and s.t.}~x_0^* \leq f(z) \leq x_1^* \big\}.
\end{align*}
Then, the free boundary $c$ of \eqref{Boundary c} is the unique solution in $\mathcal{M}$ to the integral equation
\begin{align}\label{Integral Equation for c}
(&e^{c(z)} -\kappa)(1 + e^{\frac{\gamma}{\sigma}(c(z) + z)} )  =  \int_0^\infty e^{-rs} \Big( \int_\mathbb{R} - g(w, Z_s) G(w; c(z) + \mu_0 s, \sigma \sqrt{s}) \one_{ \{ w \geq c(z) \}} dw  \Big) ds  ,
\end{align}
with $g$ as in \eqref{Definition g(x,z)} and $G$ as in \eqref{Def Density Fct G}. 
\end{proposition}
\begin{proof}
We take $x = c(z)$ in Proposition \ref{Proposition: Prob Repr of hatv}. Employing the continuity of the value function we find 
\begin{align}\label{Equal Prop Integral Equ}
(e^{c(z)} -\kappa)(1 + e^{\frac{\gamma}{\sigma}(c(z) + z)} ) 
= \mathbb{E}^\mathbb{Q} \Big[ - \int_0^\infty e^{-rs}   g(X_s^{c(z)} , Z_s^z) \one_{ \{ X_s^{c(z)} \geq c(Z_s^z) \} } ds \Big], \quad z \in \mathbb{R}.
\end{align}
By noticing that $Z^z$ is deterministic and $X_s^{c(z)}$ is Gaussian under $\mathbb{Q}$ with mean $c(z) + \mu_0 s$ and variance $\sigma^2 s$, we can reformulate \eqref{Equal Prop Integral Equ} as \eqref{Integral Equation for c}, upon using \eqref{Def Density Fct G}. To show uniqueness one can employ a four-step-approach exploiting the superharmonic characterization of $\hat{v}$, as originally developed in Th.~3.1 of Peskir \cite{Pe}. Since the present setting does not exhibit additional challenges, we omit details for the sake of brevity. \qed
\end{proof}
\begin{Remark}
As is turns out, the integral equation \eqref{Integral Equation for c} allows to derive an integral equation for the boundary $b^{-1}$ of \eqref{Inverse of b} as well. Indeed, taking $z = c^{-1}(x)$ in \eqref{Integral Equation for c} and using \eqref{Boundary c-1}  yields 
\begin{align*}
(e^x - \kappa) (1 + b^{-1} (x) ) = \mathbb{E}^\mathbb{Q} \Big[ -  \int_0^\infty e^{-rs} g\big( X_s^x , \frac{\sigma}{\gamma} \ln ( \Phi_s^{b^{-1} (x) } ) - X_s^x \big) \one_{ \{ \Phi_s^{b^{-1}(x)} \leq b^{-1} (X_s^x) \} } ds \Big], \qquad x \in \mathbb{R}.
\end{align*} 
In particular, it follows from the latter 
\begin{align}\label{IntegralEqu b reformulate}
b^{-1} (x) = \frac{1}{e^x - \kappa}\mathbb{E}^\mathbb{Q} \Big[ -  \int_0^\infty e^{-rs} g\big( X_s^x , \frac{\sigma}{\gamma} \ln ( \Phi_s^{b^{-1} (x) } ) - X_s^x \big) \one_{ \{ \Phi_s^{b^{-1}(x)} \leq b^{-1} (X_s^x) \} } ds \Big] -1 , \quad x \in \mathbb{R}.
\end{align}
Notice that the domain of $b^{-1}$ is given by the interval $[x_0^*, x_1^*]$ (cf. Lemma \ref{Lemma: Properties boundary b}) and hence, we do not encounter any problems when dividing by $e^x - \kappa$ since Assumption \ref{Assumption: Well-posedness} guarantees $e^x - \kappa > 0$ for $x \geq x_0^*$. 
\end{Remark}
\section{Solution of the Optimal Execution Problem}\label{Section: Solution to Execution Problem}
In this section, we finally return to the optimal execution problem of Section \ref{Section: First related OSP} and provide its solution. Before we do so, it is helpful to transform the singular stochastic control problem \eqref{Value function control problem } by arguing as for the optimal stopping problem in Sections \ref{Section: Decoupling Change of Measure} and \ref{Section: Parabolic Formulation}, respectively. Since the arguments are in the same spirit of those developed in Section \ref{Section: Decoupling Change of Measure}, details are omitted (see also Section 4 in Federico et al.~\cite{FeFe}). First, we make a change of measure as in Section \ref{Section: Decoupling Change of Measure}, and for $\mathbb{Q}$ as introduced therein, we let
\begin{align}\label{Dynamics of X Control Problem Vbar}
d X_t^\xi = \mu_0 dt + \sigma d B_t - \alpha d \xi_t , \qquad X_{0-}^\xi = x,
\end{align}
denote the dynamics of the controlled process $X^\xi$ under $\mathbb{Q}$. Hence, conditionally to $X_{0-}^\xi = x,~Y_{0-}^\xi = y$ and $\Phi_0 = \varphi$, we introduce the transformed optimal control problem
\begin{align}\label{Value Function Vbar}
\overline{V}(x,y,\varphi) := \sup_{\xi \in \mathcal{A}(y)} \mathbb{E}^\mathbb{Q}_{(x,y,\varphi)} \left[ \int_0^\infty e^{-rt} \big( e^{X_t^\xi } - \kappa \big) (1 + \Phi_t) \circ d \xi_t \right], \qquad (x,y,\varphi) \in \mathbb{R} \times (0, \infty) \times (0,\infty),
\end{align}
and observe that $\overline{V}(x,y,\varphi) = (1 + \varphi) V(x,y,\frac{\varphi}{1 + \varphi})$. 
Furthermore, we set
\begin{align}\label{Def Zt xi}
Z_t^\xi := \frac{\sigma}{\gamma} \log ( \Phi_t) - X_t^\xi, \qquad z:= \frac{\sigma}{\gamma} \log (\varphi) -x,
\end{align}
for any $(x, \varphi)\in \mathbb{R} \times(0,\infty)$, which, through an application of It\^{o}-Meyer's formula, is easily shown to have dynamics
\begin{align}\label{Dynamics Zt xi}
d Z_t^\xi = - \frac12 (\mu_0 + \mu_1) dt + \alpha d \xi_t, \qquad Z_{0-}^\xi = z.
\end{align}
Finally, analogously to \eqref{Value function parabolic formulation}, we define  
\begin{align}\label{Value function Vhat}
\hat{V} (x,y,z) := \overline{V}(x,y,e^{\frac{\gamma}{\sigma}(x+z)} ) = \sup_{\xi \in \mathcal{A} (y) } \mathbb{E}^\mathbb{Q}_{(x,y,z)} \left[ \int_0^\infty e^{-rt} \big(e^{X_t^\xi} - \kappa \big) \big( 1 + e^{\frac{\gamma}{\sigma} (X_t^\xi + Z_t^\xi)} \big) \circ d \xi_t \right],
\end{align}
for $(x,y,z) \in \mathcal{O} := \mathbb{R} \times (0,\infty) \times \mathbb{R}$, where 
$\mathbb{E}^\mathbb{Q}_{(x,y,z)}$ denotes the expectation conditional on $X_{0-}^\xi = x,~ Y_{0-}^\xi =y$ and $Z_{0-}^\xi = z$. \\
In the following, we introduce a candidate for the value function $V$ of \eqref{Value function control problem } and - through the explicit relationships between the value functions $v, \overline{v}$ and $\hat{v}$ - also for the value functions $\overline{V}$ and $\hat{V}$ of \eqref{Value Function Vbar} and \eqref{Value function Vhat}. To this end, we set 
\begin{align}\label{Candidate Value Fct U}
U (x,y,\pi) := \frac{1}{\alpha} \int_{x - \alpha y }^x v(x' , \pi ) d x' ,
\end{align}
where $v$ denotes the value function of \eqref{Value fct of opt stopp}. Upon using the explicit relationship \eqref{v and vbar equivalence} of $v$ and $\overline{v}$ it follows that 
\begin{align}\label{Candidate Ubar}
\overline{U}(x,y, \varphi) := (1 + \varphi) U \big( x,y, \frac{\varphi}{1 + \varphi} \big)  = (1 + \varphi) \frac{1}{\alpha} \int_{x - \alpha y}^x v \big( x' , \frac{\varphi}{1 + \varphi} \big) d x' = \frac{1}{\alpha} \int_{x - \alpha y}^x \overline{v} (x', \varphi) dx' ,
\end{align}
as the candidate for the value function $\overline{V}$ of \eqref{Value Function Vbar}. Furthermore, by defining $\hat{U}(x,y,z) := \overline{U}(x,y, e^{\frac{\gamma}{\sigma}(x+z)} )$ and exploiting the relationship \eqref{Value function parabolic formulation} we can derive 
\begin{align}\label{Candidate Uhat}
\hat{U}(x,y,z) =  \frac{1}{\alpha} \int_{x -\alpha y}^x \hat{v}(x', x + z - x') dx' = \frac{1}{\alpha} \int_z^{z + \alpha y} \hat{v} (x + z -q,q) dq ,
\end{align}
where the last equality above follows from a simple change of variables. With regard to Proposition \ref{Proposition: Smooth fit} we can state the following result, whose proof is based on direct computations. 
\begin{lemma}\label{Lemma: hatU is C1}
The function $\hat{U}$ of \eqref{Candidate Uhat} is such that $\hat{U} \in C^1 (\mathcal{O})$. Moreover, $\hat{U}_{xy},\,\hat{U}_{yz} \in C(\mathcal{O})$ and $\hat{U}_{xx} , \, \hat{U}_{xz} \in L_{\text{loc}}^\infty (\mathcal{O})$. 
\end{lemma}
\begin{proof}
Notice that \eqref{Candidate Uhat} gives 
\begin{align}
\hat{U}_x (x,y,z) &= \frac{1}{\alpha} \int_z^{z + \alpha y} \hat{v}_x (x + z - q,q)dq,  \qquad 
\hat{U}_y (x,y,z) = \hat{v}(x - \alpha y, z + \alpha y),  \label{Uhat_x and Uhat_y}\\
\hat{U}_z (x,y,z) &= \frac{1}{\alpha} \int_z^{z + \alpha y} \hat{v}_x (x + z -q,q) dq + \frac{1}{\alpha} \left[ \hat{v} (x- \alpha y, z + \alpha y) - \hat{v}(x,z) \right] \nonumber \\
&= \frac{1}{\alpha} \int_z^{z + \alpha y} \hat{v}_z (x+z -q ,q) dq, \label{Uhat_z}
\end{align}
so that $\hat{U}_x, \hat{U}_y$ and $\hat{U}_z$ are continuous due to Proposition \ref{Proposition: Smooth fit}. Moreover, Proposition \ref{Proposition: Smooth fit} also implies that
\begin{align}\label{Uhat_xx}
\hat{U}_{xx}(x,y,z) = \frac{1}{\alpha} \int_z^{z + \alpha y} \hat{v}_{xx} (x+z-q,q)dq ,
\end{align}
is locally bounded, and the mixed derivatives 
\begin{align*}
\hat{U}_{xy} (x,y,z) = \hat{v}_x(x- \alpha y ,z + \alpha y), \qquad
\hat{U}_{yz}(x,y,z) = \hat{v}_z (x - \alpha y, z + \alpha y),
\end{align*}
are continuous, while 
\begin{align*}
\hat{U}_{xz} (x,y,z) = \frac{1}{\alpha} \int_z^{z + \alpha y} \hat{v}_{xx} (x + z -q,q) dq + \frac{1}{\alpha} \left[ \hat{v}_x (x - \alpha y, z + \alpha y) - \hat{v}_x (x,z) \right],
\end{align*}
is locally bounded. Furthermore, it is easy to see that $\hat{U}_{yx} = \hat{U}_{xy}$, $\hat{U}_{xz} = \hat{U}_{zx}$ and $\hat{U}_{yz} = \hat{U}_{zy}$. \qed
\end{proof}
The proof of the next corollary follows from \eqref{Uhat_x and Uhat_y}-\eqref{Uhat_z} and direct computations.
\begin{corollary}\label{Corollary: Rewrite C2 and C3}
One has
\begin{align}\label{aUhatx - aUhatz + Uhaty = vhat}
\alpha \hat{U}_x (x,y,z) - \alpha \hat{U}_z (x,y,z) + \hat{U}_y (x,y,z) = \hat{v}(x,z) \geq \big( e^x - \kappa \big) \big( 1 + e^{\frac{\gamma}{\sigma}(x+ z) } \big),
\end{align}
so
\begin{align}
\label{eq:W3}
\mathbb{W}_3:=\big\{ (x,y,z) \in \mathcal{O}:\,\alpha \hat{U}_x (x,y,z) - \alpha \hat{U}_z (x,y,z) + \hat{U}_y (x,y,z) > \big( e^x - \kappa \big) \big( 1 + e^{\frac{\gamma}{\sigma}(x+ z) } \big) \big\} = \mathcal{C}_3,  
\end{align}
with $\mathcal{C}_3$ as in \eqref{Continuation region C3}. Furthermore, $\hat{U}_x - \hat{U}_z = \overline{U}_x$ as well as $\hat{U}_y = \overline{U}_y$, and we have 
\begin{align}
\label{eq:W2}
\mathbb{W}_2:=\big\{ (x,y,\varphi) \in \mathbb{R} \times (0,\infty) \times (0,\infty):\, \alpha \overline{U}_x (x,y,\varphi)  + \overline{U}_y (x,y,\varphi) > \big( e^x - \kappa \big) \big( 1 + \varphi \big) \big\} = \mathcal{C}_2,
\end{align}
with $\mathcal{C}_2$ of \eqref{continuation after girsanov}.
\end{corollary}

\hspace{-0.55cm}7.1 \textbf{Construction of the optimal control for the state space process $(X,Y,\Phi)$.} Recall $b$ as in \eqref{boundary b}, which is nondecreasing and left-continuous by Lemma \ref{Lemma: Properties boundary b}. Then, for any $(x,y,\varphi) \in \mathbb{R} \times (0,\infty)\times (0,\infty)$, define the admissible control strategy  
\begin{align}\label{Control xi hat explicit}
\hat{\xi}_t := y \wedge \sup_{0 \leq s \leq t} \frac{1}{\alpha} \Big[ x - b(\Phi_s^\varphi) + \mu_0 s + \sigma B_s \Big]^+, \qquad t \geq 0, \qquad \hat{\xi}_{0-} = 0,
\end{align}
according to which the investor should only execute a lump-sum amount of shares whenever the process $X_{t-}$ is strictly inside the selling region and hence strictly above the boundary $b(\Phi_t)$. More precisely, if $y \leq \frac{1}{\alpha} (x - b(\varphi))$ it is optimal to sell the complete amount of shares instantaneously, while for $y > \frac{1}{\alpha} (x - b(\varphi))$ the system is brought immediately to the level $(X_0, Y_0, \Phi_0)= (b(\varphi), y - \frac{1}{\alpha} (x - b(\varphi)), \varphi)$. Afterwards, the strategy \eqref{Control xi hat explicit} prescribes to take action whenever the process $X_t$ approaches the boundary $b(\Phi_t)$ from below and the process $(X_t,Y_t)$ is obliquely reflected at the belief-dependent boundary $b(\Phi_t)$ in the direction $(- \alpha, -1)$. Hence, the process $X_t$ is kept inside the interval $(-\infty , b(\Phi_t)]$ with ``minimal effort''. These actions are the so-called \textit{Skorokhod reflection-type policies} and caused by the continuous part $\hat{\xi}^c$ of the control $\hat{\xi}$. Notice that the nondecreasing process $\hat{\xi}$, and the induced random measure $d\hat{\xi}$ on $[0,\infty)$, are such that (recall \eqref{eq:W2})
\begin{align}\label{Control xi hat X Phi}
\begin{cases}
(X_{t}^{\hat{\xi}} , Y_{t}^{\hat{\xi}} , \Phi_t ) \in \overline{\mathbb{W}}_2,  \quad \mathbb{Q} \otimes dt \text{-a.s.};\\
d \hat{\xi}_t ~\text{has support on}~\{ t \geq 0:~ (X_{t-}^{\hat{\xi}} , Y_{t-}^{\hat{\xi}}, \Phi_t) \notin \mathbb{W}_2 \}; \\
\hat{\xi}_t \leq y, \quad t \geq 0.
\end{cases}
\end{align}
Furthermore, due to \eqref{Dynamics Zt xi}-\eqref{Value function Vhat} and Corollary \ref{Corollary: Rewrite C2 and C3}, we can express the control $\hat{\xi}$ equivalently in terms of the state-process $(X^{\hat{\xi}} , Y^{\hat{\xi}}, Z^{\hat{\xi}})$ by (cf.\ \eqref{eq:W3})
\begin{align}\label{Control xi hat X Z}
\begin{cases}
(X_t^{\hat{\xi}} , Y_t^{\hat{\xi}} , Z_t^{\hat{\xi}} ) \in \overline{\mathbb{W}}_3,  \quad \mathbb{Q} \otimes dt \text{-a.s.};\\
d \hat{\xi}_t ~\text{has support on}~\{ t \geq 0:~ (X_{t-}^{\hat{\xi}}, Y_{t-}^{\hat{\xi}}, Z_{t-}^{\hat{\xi}} ) \notin \mathbb{W}_3 \}; \\
\hat{\xi}_t \leq y, \quad t \geq 0.
\end{cases}
\end{align}
In the following, we prove that in fact $\hat{\xi}$ is an optimal control for problem \eqref{Value function Vhat} and $\hat{U} = \hat{V}$. As an immediate consequence we have that $\overline{U} = \overline{V}$ and $U=V$. 
\begin{theorem}[Verification Theorem]\label{Theorem: Verification Theorem}
Let $(x,y,z) \in \mathbb{R} \times [0,\infty) \times \mathbb{R}$ and $\hat{U}(x,y,z)$ as in \eqref{Candidate Uhat}. Then, one has $\hat{U}(x,y,z) = \hat{V} (x,y,z)$ and $\hat{\xi}$ as in \eqref{Control xi hat explicit} is optimal for the singular control problem \eqref{Value function Vhat}. 
\end{theorem}
\begin{proof}
First of all, for $y =0$ we have $\hat{U}(x,0,z) =0= \hat{V}(x,0,z)$. Hence, in the following we assume $(x,y,z)\in \mathcal{O}$. 

\textbf{1.~We prove $\hat{U} \geq \hat{V}$}. Take an arbitrary control $\xi \in \mathcal{A}(y)$ and for $R>0$ and $N \in \mathbb{N}$ we set $\tau_{R,N} := \inf \{ s \geq 0:~\vert (X_s^\xi , Z_s^\xi)\vert > R \} \wedge N$. Due to Lemma \ref{Lemma: hatU is C1} we can proceed as in Fleming and Soner \cite{FlSo}, Chapter 8, Th.~4.1 to obtain (after performing an approximation of $\hat{U}$ via mollifiers and taking limits)
\begin{align}\label{Ito on Uhat in VerThe}
\mathbb{E}^\mathbb{Q}_{(x,y,z)} \big[ &e^{-r \tau_{R,N}}  \hat{U}(X_{\tau_{R,N}}^\xi , Y_{\tau_{R,N}}^\xi , Z_{\tau_{R,N}}^\xi ) - \hat{U}(x,y,z) \big] \nonumber\\
&= \mathbb{E}_{(x,y,z)}^\mathbb{Q} \Big[ \int_0^{\tau_{R,N}} e^{-rs} (\mathcal{L}_{X,Z} -r) \hat{U}(X_s^\xi , Y_s^\xi , Z_s^\xi ) ds + \underbrace{\sigma \int_0^{\tau_{R,N}} e^{-rs} \hat{U}_x (X_{\tau_{R,N}}^\xi , Y_{\tau_{R,N}}^\xi , Z_{\tau_{R,N}}^\xi) d B_s}_{=: M_{R,N}} \nonumber \\
&\quad + \sum_{0 \leq s \leq \tau_{R,N}} e^{-rs} \left( \hat{U}(X_s^\xi , Y_s^\xi , Z_s^\xi) - \hat{U}(X_{s-}^\xi , Y_{s-}^\xi , Z_{s-}^\xi) \right) \nonumber \\
&\quad + \int_0^{\tau_{R,N}} e^{-rs} \left( - \alpha \hat{U}_x (X_s^\xi , Y_s^\xi , Z_s^\xi) - \hat{U}_y (X_s^\xi , Y_s^\xi , Z_s^\xi) + \alpha \hat{U}_z (X_s^\xi , Y_s^\xi , Z_s^\xi) \right) d \xi_s^c \Big].
\end{align}
Notice that 
\begin{align}\label{Uhat - Uhat in VerThe}
\hat{U}(X_s^\xi , Y_s^\xi , Z_s^\xi) - \hat{U}(X_{s-}^\xi , Y_{s-}^\xi , Z_{s-}^\xi) 
&= \hat{U}(X_{s-}^\xi - \alpha \Delta \xi_s , Y_{s-}^\xi - \Delta \xi_s , Z_{s-}^\xi + \alpha \Delta \xi_s ) - \hat{U}(X_{s-}^\xi , Y_{s-}^\xi ,Z_{s-}^\xi ) \nonumber \\
&= \int_0^{\Delta \xi_s} \frac{\partial \hat{U} (X_{s-}^\xi - \alpha u, Y_{s-}^\xi - u, Z_{s-}^\xi + \alpha u)}{\partial u} du \nonumber \\
&= \int_0^{\Delta \xi_s} (- \alpha \hat{U}_x - \hat{U}_y + \alpha \hat{U}_z) (X_{s-}^\xi - \alpha u , Y_{s-}^\xi - u, Z_{s-}^\xi + \alpha u)du.
\end{align}
Hence, combining \eqref{Ito on Uhat in VerThe} and \eqref{Uhat - Uhat in VerThe}, upon adding the term 
\begin{align*}
\mathbb{E}_{(x,y,z)}^\mathbb{Q} \big[ \int_0^{\tau_{R,N}} e^{-rs} (e^{X_s^\xi} - \kappa)(1 + e^{\frac{\gamma}{\sigma} (X_s^\xi + Z_s^\xi)} ) d \xi_s^c + \sum_{0 \leq s \leq \tau_{R,N}} e^{-rs} \int_0^{\Delta \xi_s} (e^{X_{s-}^\xi - \alpha u} - \kappa)(1 + e^{\frac{\gamma}{\sigma} (X_{s-}^\xi + Z_{s-}^\xi )} ) du \big],
\end{align*}
on both sides, yields
\begin{align}\label{Equality in VerThe}
&\mathbb{E}_{(x,y,z)}^\mathbb{Q} \Big[ \int_0^{\tau_{R,N}} e^{-rs} (e^{X_s^\xi} -\kappa) (1 + e^{\frac{\gamma}{\sigma} (X_s^\xi + Z_s^\xi)}) d \xi_s^c  \nonumber \\
&\hspace*{2cm}  + \sum_{0 \leq s \leq \tau_{R,N}} e^{-rs} \int_0^{\Delta \xi_s} ( e^{X_{s-}^\xi - \alpha u} -\kappa)(1 + e^{\frac{\gamma}{\sigma} (X_{s-}^\xi + Z_{s-}^\xi )}) du - \hat{U} (x,y,z) \Big] \nonumber \\
= ~ &\mathbb{E}_{(x,y,z)}^\mathbb{Q} \Big[ \int_0^{\tau_{R,N}} e^{-rs} (\mathcal{L}_{X,Z}- r) \hat{U}(X_s^\xi, Y_s^\xi , Z_s^\xi) ds + M_{R,N} - e^{-r \tau_{R,N}}  \hat{U}(X_{\tau_{R,N}}^\xi , Y_{\tau_{R,N}}^\xi , Z_{\tau_{R,N}}^\xi )\nonumber \\
&\hspace*{2cm} + \sum_{0 \leq s \leq \tau_{R,N}} e^{-rs} \int_0^{\Delta \xi_s} (- \alpha \hat{U}_x - \hat{U}_y + \alpha \hat{U}_z) (X_{s-}^\xi - \alpha u , Y_{s-}^\xi - u, Z_{s -}^\xi + \alpha u )  \nonumber \\
&\hspace*{8.3cm} + (e^{X_{s-}^\xi - \alpha u} - \kappa)(1 + e^{\frac{\gamma}{\sigma} (X_{s-}^\xi + Z_{s-}^\xi )} ) du  \nonumber \\
&\hspace*{2cm} + \int_0^{\tau_{R,N}} e^{-rs} \big( ( - \alpha \hat{U}_x - \hat{U}_y + \alpha \hat{U}_z) (X_s^\xi , Y_s^\xi , Z_s^\xi) + (e^{X_s^\xi} -\kappa) (1 + e^{\frac{\gamma}{\sigma} (X_s^\xi + Z_s^\xi)} ) \big) d \xi_s^c \Big].
\end{align}
We observe that \eqref{Uhat_x and Uhat_y}-\eqref{Uhat_xx} imply
\begin{align}\label{L-r Uhat leq 0}
 (\mathcal{L}_{X,Z} -r )\hat{U}(x,y,z) = \frac{1}{\alpha} \int_{x-\alpha y}^{x} (\mathcal{L}_{X,Z} - r) \hat{v}(x' ,x+z - x') dx' \leq 0,
\end{align} 
where the last inequality follows from the supermartingale property of $(e^{-rt} \hat{v}(X_t, Z_t))_t$ combined with the regularity obtained in Proposition \ref{Proposition: Smooth fit}. Hence, due to \eqref{aUhatx - aUhatz + Uhaty = vhat} and since $\hat{U} \geq 0$ and $\mathbb{E}^\mathbb{Q}_{(x,y,z)} [M_{R,N}] = 0$,  \eqref{Equality in VerThe} writes as 
\begin{align}\label{Inequality VerThe}
\hat{U} (x,y,z) 
&\geq \mathbb{E}^\mathbb{Q}_{(x,y,z)} \Big[ \int_0^{\tau_{R,N}} e^{-rs} (e^{X_s^\xi} - \kappa) (1 + e^{\frac{\gamma}{\sigma} (X_s^\xi + Z_s^\xi)}) d \xi_s^c  \nonumber \\
&\hspace*{5cm}+  \sum_{0 \leq s \leq \tau_{R,N}} e^{-rs} \int_0^{\Delta \xi_s} ( e^{X_{s-}^\xi - \alpha u} - \kappa)(1 + e^{\frac{\gamma}{\sigma} (X_{s-}^\xi + Z_{s-}^\xi )}) du \Big].
\end{align}
Taking limits as $R \uparrow \infty$ as well as $N \uparrow \infty$, invoking the dominated convergence theorem due to Assumption \ref{Assumption: Well-posedness}, we obtain
\begin{align}
\hat{U}(x,y,z) &\geq \mathbb{E}^\mathbb{Q}_{(x,y,z)} \Big[ \int_0^\infty e^{-rs} (e^{X_s^\xi} - \kappa) (1 + e^{\frac{\gamma}{\sigma}(X_s^\xi + Z_s^\xi)} ) d \xi_s^c \nonumber \\
& \hspace*{5cm}+  \sum_{s: \Delta \xi_s \neq 0} e^{-rs} \int_0^{\Delta \xi_s} ( e^{X_{s-}^\xi - \alpha u} - \kappa)(1 + e^{\frac{\gamma}{\sigma} (X_{s-}^\xi + Z_{s-}^\xi )}) du \Big] \nonumber \\
&= J(x,y,z,\xi).
\end{align}
Since $\xi$ was arbitrary, we have
\begin{align}\label{Uhat geq Vhat}
\hat{U}(x,y,z) \geq \hat{V}(x,y,z),
\end{align}
for all $(x,y,z) \in \mathcal{O}$. That is, $\hat{U} \geq \hat{V}$ on $\mathcal{O}$. 

\textbf{2.~We prove that $\hat{U} \leq \hat{V}$}. In order to accomplish that, let $\hat{\xi}$ satisfy the conditions in \eqref{Control xi hat X Z} and define $\hat{\tau}_{R,N} = \inf \{ t \geq 0: ~ \vert (X_t^{\hat{\xi}} , Z_t^{\hat{\xi}} ) \vert > R \} \wedge N$, again for $R>0$ and $N \in \mathbb{N}$. Notice that the properties of $\hat{\xi}$ imply equalities in \eqref{aUhatx - aUhatz + Uhaty = vhat} and \eqref{L-r Uhat leq 0}, where the equality in \eqref{L-r Uhat leq 0} follows from the monotonicity of $c$ and we can deduce that $(x' , x+z -x') \in \mathbb{W}_3$ for $(x,y,z) \in \mathbb{W}_3$ and $x' \leq x$. Employing the same arguments as in the first part of the proof yields
\begin{align}\label{Equality for VerThe}
\hat{U} (x,y,z) &= \mathbb{E}^\mathbb{Q} _{(x,y,z)} \Big[ e^{-r \hat{\tau}_{R,N}} \hat{U}(X_{\hat{\tau}_{R,N}}^{\hat{\xi}} , Y_{\hat{\tau}_{R,N}}^{\hat{\xi}} , Z_{\hat{\tau}_{R,N}}^{\hat{\xi}} ) \Big] \nonumber \\ 
&\quad + \mathbb{E}^\mathbb{Q}_{(x,y,z)} \Big[ \int_0^{\hat{\tau}_{R,N}} e^{-rs} (e^{X_s^{\hat{\xi}}} -\kappa) (1 + e^{\frac{\gamma}{\sigma} (X_s^{\hat{\xi}} + Z_s^{\hat{\xi}})}) d \hat{\xi}_s^{c} \nonumber \\
&\hspace*{5cm}+  \sum_{0 \leq s \leq \hat{\tau}_{R,N}} e^{-rs} \int_0^{\Delta \hat{\xi}_s} ( e^{X_{s-}^{\hat{\xi}} - \alpha u} - \kappa)(1 + e^{\frac{\gamma}{\sigma} (X_{s-}^{\hat{\xi}} + Z_{s-}^{\hat{\xi}} )}) du \Big].
\end{align}
It is thus left to prove that
\begin{align}\label{Limit for VerThe}
\lim_{N \uparrow \infty} \lim_{R \uparrow \infty}  \mathbb{E}^\mathbb{Q}_{(x,y,z)} \Big[e^{-r \hat{\tau}_{R,N}} \hat{U}(X_{\hat{\tau}_{R,N}}^{\hat{\xi}} , Y_{\hat{\tau}_{R,N}}^{\hat{\xi}} , Z_{\hat{\tau}_{R,N}}^{\hat{\xi}} ) \Big] = 0,
\end{align}
since taking limits as $R \uparrow \infty $ and $N \uparrow \infty$ together with \eqref{Inequality VerThe} implies $J(x,y,z,\hat{\xi}) = \hat{U} (x,y,z)$ and hence $\hat{V}(x,y,z) \geq \hat{U}(x,y,z)$ for all $(x,y,z) \in \mathcal{O}$. Combining the latter with \eqref{Uhat geq Vhat} yields $\hat{U} = \hat{V}$ on $\mathcal{O}$. 

In order to prove \eqref{Limit for VerThe} we notice that Lemma \ref{Lemma: Properties of vbar} i), \eqref{Value function parabolic formulation} and \eqref{Candidate Uhat} imply 
\begin{align}\label{Upper bound on Uhat}
\hat{U}(x,y,z) \leq \frac{1}{\alpha} \int_z^{z+\alpha y} K_1 e^{x+z-q} (1 + e^{\frac{\gamma}{\sigma} (x+z -q + q) }) dq = \frac{1}{\alpha} K_1 e^x (1 + e^{\frac{\gamma}{\sigma} (x+z) }) (1 - e^{- \alpha y}),
\end{align}
and, since $y \mapsto \hat{U}(x,y,z)$ is increasing, we obtain 
\begin{align*}
0 \leq e^{-r \hat{\tau}_{R,N}} \hat{U}(X_{\hat{\tau}_{R,N}}^{\hat{\xi}} , Y_{\hat{\tau}_{R,N}}^{\hat{\xi}} , Z_{\hat{\tau}_{R,N}}^{\hat{\xi}} ) &\leq e^{-r \hat{\tau}_{R,N}} \hat{U}(X_{\hat{\tau}_{R,N}}^{\hat{\xi}} , y , Z_{\hat{\tau}_{R,N}}^{\hat{\xi}} ) \\ 
&\leq \frac{K_1}{\alpha} (1 - e^{- \alpha y} ) e^{-r \hat{\tau}_{R,N}} e^{X_{\hat{\tau}_{R,N}}^0} (1 + e^{ \frac{\gamma}{\sigma} (X_{\hat{\tau}_{R,N}}^0 + Z_{\hat{\tau}_{R,N}}^0 ) } ),
\end{align*}
where we used that $X_t^\xi \leq X_t^0$ as well as $X_t^\xi + Z_t^\xi = X_t^0 + Z_t^0$ a.s. Hence, taking expectations  yields
\begin{align}
0 
\leq \mathbb{E}^\mathbb{Q}_{(x,y,z)} \Big[ e^{-r \hat{\tau}_{R,N}} &\hat{U}(X_{\hat{\tau}_{R,N}}^{\hat{\xi}} , Y_{\hat{\tau}_{R,N}}^{\hat{\xi}} , Z_{\hat{\tau}_{R,N}}^{\hat{\xi}} )  \Big] \nonumber \\
&\leq \frac{K_1}{\alpha} (1 - e^{- \alpha y} ) \mathbb{E}^\mathbb{Q}_{(x,y,z)} \Big[ e^{-r \hat{\tau}_{R,N}} e^{X_{\hat{\tau}_{R,N}}^0} (1 + e^{\frac{\gamma}{\sigma} (X_{\hat{\tau}_{R,N}}^0 + Z_{\hat{\tau}_{R,N}}^0 )} ) \Big] \nonumber \\ 
&= \frac{K_1}{\alpha} (1 - e^{-\alpha y})(1 + e^{\frac{\gamma}{\sigma}(x+z)}) \mathbb{E}_{(x,y,\pi)} \Big[ e^{-r \hat{\tau}_{R,N}} e^{X_{\hat{\tau}_{R,N}}^0} \Big], 
\end{align}
with $\pi := e^{\frac{\gamma}{\sigma}(x+z)}/(1 + e^{\frac{\gamma}{\sigma}(x+z)} )$ and the last equality follows from a change of measure as in Section \ref{Section: Decoupling Change of Measure}. Upon using Assumption \ref{Assumption: Well-posedness}, it is easy to check that \eqref{Limit for VerThe} holds true, thus completing the proof. \qed
\end{proof}
\begin{Remark}
We can use the transformation \eqref{Transformation Tbar} from $(x,z)$- to $(x, \varphi)$-coordinates in order to show that $\hat{\xi}$ is an optimal control for problem \eqref{Value Function Vbar} as well.  
Indeed, recall \eqref{Process Z} and $\hat{V}(x,y,z) = \hat{V}(x,y,\frac{\sigma}{\gamma} \ln (\varphi) -x) = \overline{V}(x,y,\varphi)$ to conclude 
\begin{align*}
\overline{V}(x,y,\varphi) &= \mathbb{E}^\mathbb{Q}_{(x,y,\varphi)} \Big[ \int_0^\infty e^{-rs} (e^{X_s^{\hat{\xi}}} - \kappa ) (1 + \Phi_s) d\hat{\xi}_s^{c}  + \sum_{s:~\Delta \hat{\xi}_s \neq 0} e^{-rs} \int_0^{\Delta \hat{\xi}_s} (e^{X_{s-}^{\hat{\xi}} - \alpha u} - \kappa)(1 + \Phi_{s} ) du \Big].
\end{align*}
Furthermore, the latter equation and \eqref{Candidate Ubar} imply $U(x,y,\pi)=V(x,y,\pi)$ for all $(x,y,\pi) \in \mathbb{R}\times (0,\infty) \times (0,1)$. 
\end{Remark}
\begin{Remark}
Letting $\tilde{\tau}(x,y,\varphi) := \inf \{ t \geq 0:~x + \mu_0 t + \sigma B_t \geq b(\Phi_s^\varphi) \}$, the optimal execution strategy $\hat{\xi}$ as in \eqref{Control xi hat explicit} converges as $\alpha \downarrow 0$ to the execution strategy 
\begin{align*}
\tilde{\xi}_t = 
\begin{cases}
0 & t < \tilde{\tau} (x,y,\varphi), \\
y & t \geq \tilde{\tau} (x,y,\varphi),
\end{cases}
\end{align*}
which prescribes to sell the total amount of shares instantaneously when the process $X$ reaches the optimal execution boundary $b(\Phi)$. It is interesting to notice that the optimal solution and the value function are robust w.r.t.~the parameter $\alpha$. Indeed, by L'H\^{o}pital's rule, we see from \eqref{Candidate Ubar} that $\lim_{\alpha \downarrow 0} = y \overline{v}(x, \varphi), ~(x,y,\varphi) \in \mathbb{R}\times (0,\infty)\times (0,\infty)$. It is in fact easy to show via a verification theorem that $y \overline{v}(x,\varphi)$ and $\ttilde{\xi}$ are the value function and the optimal execution rule in the problem with no market impact. 
\end{Remark}
\begin{Remark}\label{Remark: limY}
Let $\hat{\sigma} := \inf \{ t \geq 0: ~ Y_t^\xi = 0 \}$ denote the time at which the portfolio is fully depleted. Imposing the constraint that the investor has to sell all assets until terminal time (cf.~Guo and Zervos \cite{GuZe}), we notice that for $y \leq \frac{1}{\alpha} (x - b(\varphi))$ the control strategy $\hat{\xi}$ of \eqref{Control xi hat explicit} still defines an optimal control, as the complete amount of shares is sold immediately at time $t=0$. However, for $y > \frac{1}{\alpha} (x - b(\varphi))$, simple calculations yield
\begin{align*}
\lim_{T \uparrow \infty} \mathbb{Q}( \hat{\sigma} > T) \geq 1 - \exp \Big( \frac{2 \mu_0}{\sigma^2} (\alpha y + x_0^* - x) \Big), 
\end{align*}
and we notice that for increasing $y$ and decreasing $x$, the probability increases that the investor does not sell the  entire amount of shares until terminal time. Hence, if we restrict the admissible strategies to all $\xi \in \mathcal{A}(y)$ such that $\lim_{T \to \infty} Y_T^\xi = 0$, the control strategy $\hat{\xi}$ of \eqref{Control xi hat explicit} does not provide an admissible execution strategy. In this case, arguing as in Guo and Zervos \cite{GuZe}, Proposition 5.1, we can use $\hat{\xi}$ to construct a sequence of $\epsilon$-optimal strategies. 
\end{Remark}

\section{Numerical Study}\label{Section: Comparative Statics Analysis}
In this section, we (i) perform a comparative statics analysis on the optimal execution boundaries $a$ and $b$ of \eqref{boundary a} and \eqref{boundary b}, respectively, as well as (ii) investigate the \textit{value of information} in our model, by comparing the value function $V$ of \eqref{Value function control problem } to the value of an \textit{average drift problem}. \\[0.1cm]
\subsection{Comparative Statics Analysis}\label{Sec: CompStat} Based on the integral equation \eqref{Integral Equation for c} we implement a recursive numerical scheme, which relies on an application of the Monte-Carlo method. To this end, we let $\zeta$ denote an auxiliary exponentially distributed random variable with parameter $r$, that is independent of the Brownian motion $B$. Recalling that \eqref{Integral Equation for c} can be reformulated as \eqref{IntegralEqu b reformulate}, we  notice that the latter takes the shape of a fixed point problem
\begin{align}\label{Fixed point problem}
b^{-1} (x) = \Gamma (b^{-1}(x),x;b^{-1}), 
\end{align}
for $x \in \mathbb{R}$ and $b^{-1}$ being the generalized inverse of $b$ as in \eqref{Inverse of b}. Here, the operator $\Gamma$ is defined via 
\begin{align}\label{Operator}
\Gamma (\varphi, x ;f) := \frac{1}{e^x - \kappa} \frac{1}{r} \mathbb{E}^\mathbb{Q} \Big[ -   g\big( X_\zeta^x , \frac{\sigma}{\gamma} \ln ( \Phi_\zeta^\varphi  ) - X_\zeta^x \big) \one_{ \{ \Phi_\zeta^\varphi \leq f (X_\zeta^x) \} }  \Big] -1 ,
\end{align}
for $(x,\varphi) \in \mathbb{R} \times (0,\infty)$ and a function $f: \mathbb{R} \to (0,\infty) $. By employing techniques seen in Christensen and Salminen \cite{ChSa}, Dammann and Ferrari \cite{DaFe} and Detemple and Kitapbayev \cite{DeKi}, we aim to solve \eqref{Fixed point problem} via an iterative scheme. To this end, we let
\begin{align}\label{Sequence of boundaries}
(b^{-1})^{[n]} (x) = \Gamma ((b^{-1})^{[n-1]}(x) , x; (b^{-1})^{[n-1]} ), \qquad x\in \mathbb{R}, n \geq 1,
\end{align}
define a sequence of boundaries and - for a given boundary $(b^{-1})^{[k]}$ - we estimate the expectation in \eqref{Operator} by 
\begin{align*}
 - \frac{1}{N} \sum_{i=1}^N g\Big( X_{\zeta_i}^{i,x} ,\frac{\sigma}{\gamma} \ln \big( \Phi_{\zeta_i}^{i, (b^{-1})^{[k]}(x)} \big) - X_{\zeta_i}^{i, x} \big) \one_{ \big\{  \Phi_{\zeta_i}^{i,(b^{-1})^{[k]} (x)}     \leq (b^{-1})^{[k]} ( X_{\zeta_i}^{i,x} )  \big\} },
\end{align*}
where $N$ denotes the total amount of realizations of the exponential random variable. We can choose the initial boundary $(b^{-1})^{[0]}$ as a simple exponential function with $(b^{-1})^{[0]} (x_0^*) = 0$ and $(b^{-1})^{[0]} (x) \to \infty$ for $x \uparrow x_1^*$ with $x_0^*$ and $x_1^*$ as in \eqref{x0*} and Remark \ref{Remark: Full Info mu1}, respectively. The numerical scheme \eqref{Sequence of boundaries} is then iterated until the variation between steps drops below a predetermined level. Finally, we calculate $b$ from its generalized inverse $b^{-1}$ and can transform the resulting boundary according to the explicit relationship \eqref{Boundaries a b Transformation}. We can thus study the sensitivity of $b(\varphi)$ as well as $a(\pi)$ with respect to some of the model's parameters. 
\begin{figure}[b] 
\begin{center}
  \includegraphics[width=8cm]{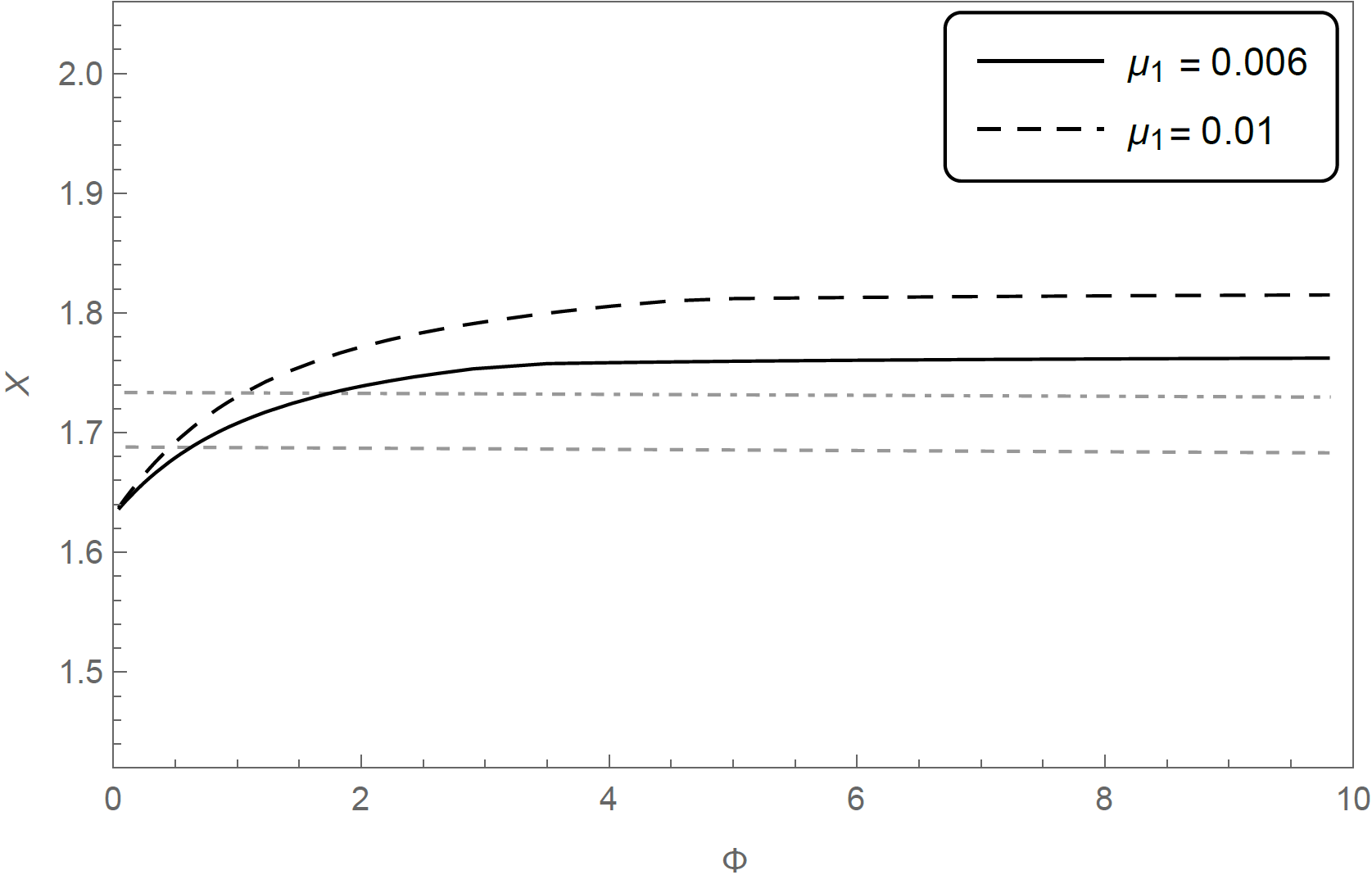}\includegraphics[width=8cm]{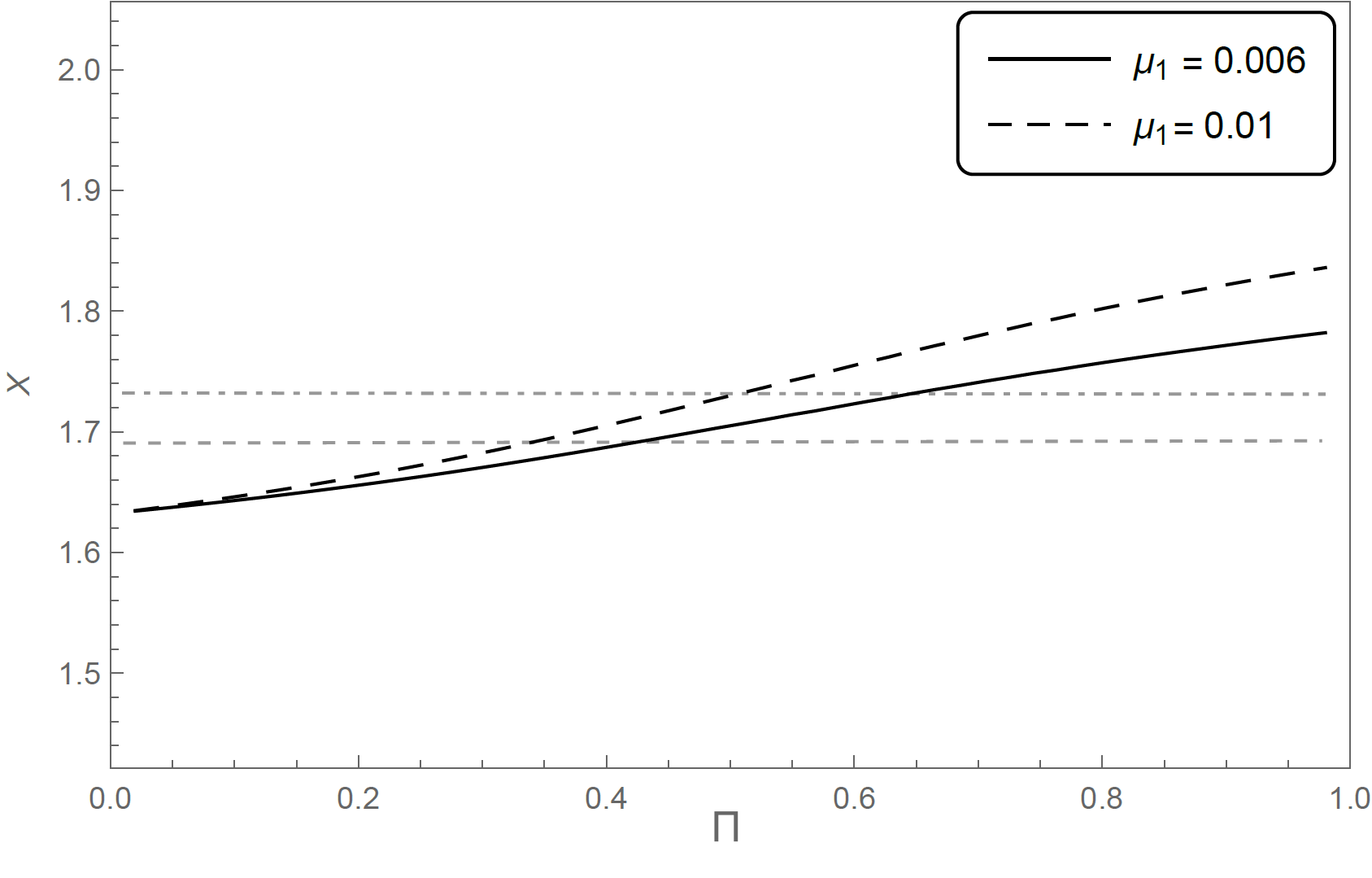} \caption{The optimal execution boundaries $b(\varphi)$ and $a(\pi)$ as well as the pre-committed strategies for different values of $\mu_1$ and following parameters: $r =0.07 ,\, \mu_0 = - 0.01 , \,\sigma = 0.17 ,\, \kappa = 3,\,\pi = 0.6.$} \label{Figure: Comp Stat Mu}
\end{center} 
\end{figure}
Furthermore, we can compare the belief-dependent boundaries to the strategy of a pre-committed agent, who - after forming an initial belief $\pi = \mathbb{P}[\mu = \mu_1 \vert \mathcal{F}_t^X]$ - refrains from updating her belief and thus acts as if the drift value was constant and equal to $\mu_1 \pi + \mu_0 (1- \pi)$.    The resulting strategy is then triggered by a constant execution threshold, which is of similar structure as the one derived in Section \ref{Section: Benchmark Problem}. Consequently, we observe that such an agent  cannot react to any price movements on the market and is thus not able to decrease or increase the target price at which she would like to sell the asset. \\[0.15cm] 
8.1.1 \textit{Sensitivity with respect to the drift.}
In Figure \ref{Figure: Comp Stat Mu} we can observe the sensitivity of the optimal execution boundaries with respect to one of the possible drift values. Since an increase in $\mu_1$ implies higher expected prices on the market, the investor delays her decision to sell a fraction of her shares and waits for larger prices to evolve. This effect is strongest for higher values of $\pi$, which reflect a stronger belief in the drift $\mu_1$. On the other hand, we notice that the lower bound $x_0^*$ remains untouched by a change in $\mu_1$, since it results from the case of full information when $\mu = \mu_0$. Consequently, for a strong belief towards the drift value $\mu_0$, the investor does not significantly change her execution strategy. \\[0.15cm] 
8.1.2 \textit{Sensitivity with respect to the discount rate.} 
Figure \ref{Figure: Comp Stat R} shows the effect on the boundaries $a$ and $b$ for a change in $r$, the latter can be interpreted as the subjective impatience of the investor. For an increasing value of $r$ the investor gets more impatient and discounts future revenues more heavily. Consequently, the investor is willing to liquidate her assets earlier, which is realized by decreasing the target price she aims at achieving on the market. This clear effect can be observed for every value of belief $\pi \in [0,1]$. 
\begin{figure}[h] 
\begin{center}
 \includegraphics[width=8cm]{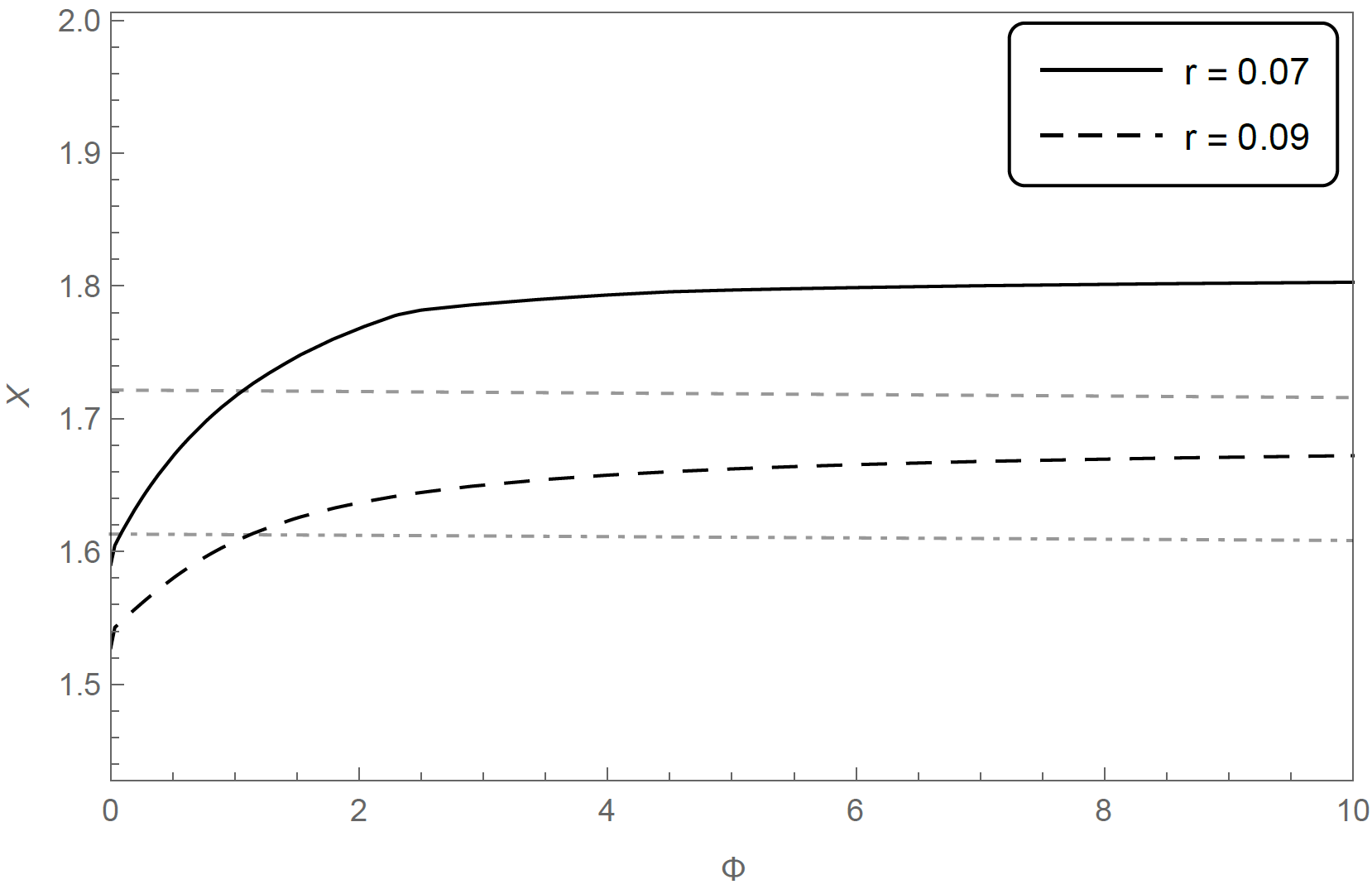}  
  \includegraphics[width=8cm]{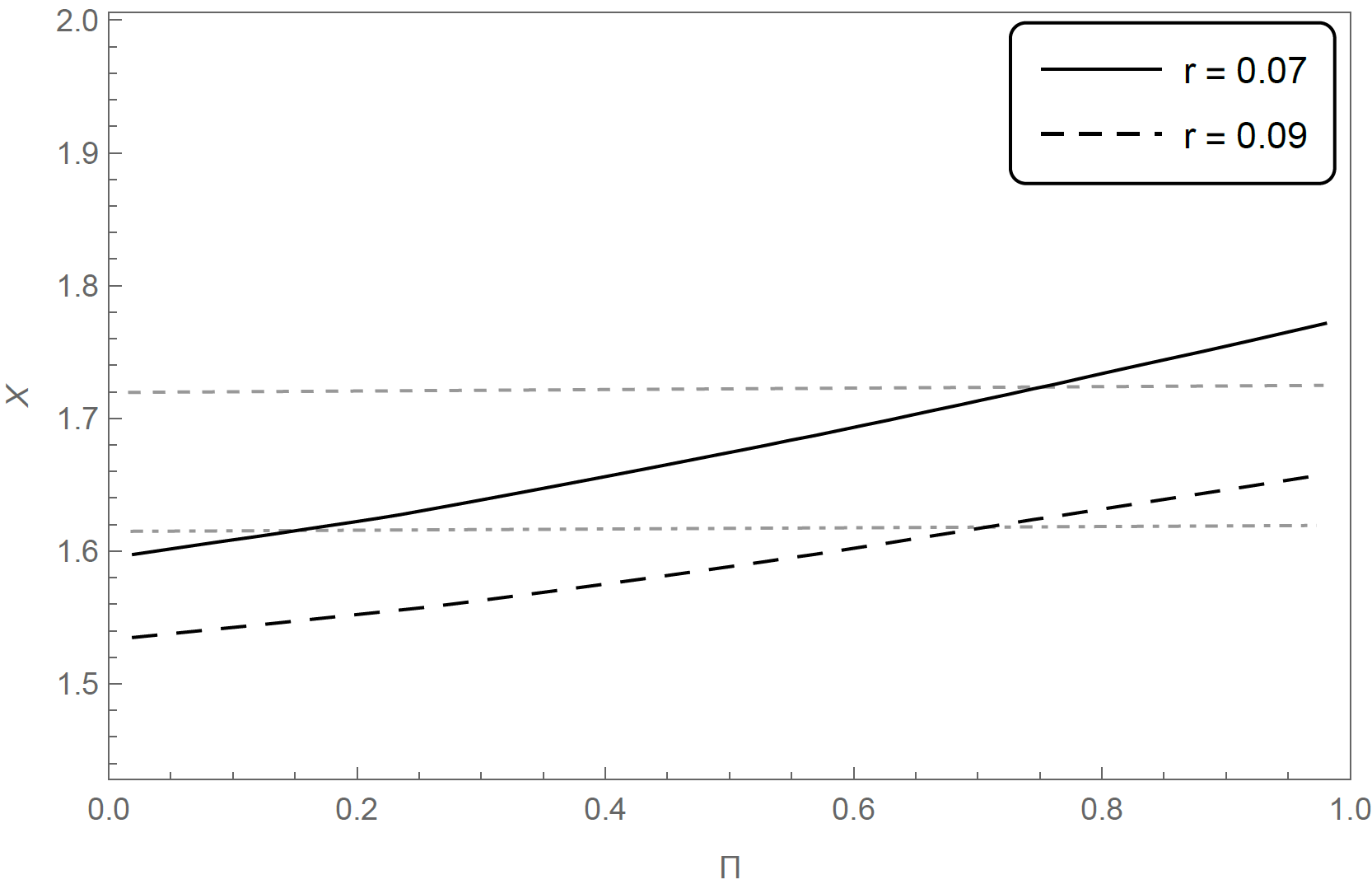} \caption{The optimal execution boundaries $b(\varphi)$ and $a(\pi)$ as well as the pre-committed strategies for different values of $r$ and following parameters: $\mu_0 = - 0.01,\, \mu_1 = 0.007, \,\sigma = 0.17 ,\, \kappa = 3.,\, \pi = 0.6.$} \label{Figure: Comp Stat R}
\end{center} 
\end{figure}\text{}\\[0.15cm]
8.1.3 \textit{Sensitivity with respect to the volatility.}
The sensitivity of the optimal execution boundaries $a$ and $b$ on the volatility of the underlying asset is more delicate. As pointed out by D\'{e}camps et al.~\cite{Deca}, who consider an optimal stopping problem of a structure similar to the one in \eqref{Value fct of opt stopp}, the effect of an increase in volatility is ambiguous and cannot always be predicted with the help of standard real option models (see for example Dixit and Pindyck \cite{Dixit}, McDonald and Siegel \cite{McDonald}). In general, one expects an increasing value function with rising volatility, as this increases the spread of possible future values of the asset and thus the maximal possible profit, while the maximal possible loss remains unchanged. The investor exploits this upside potential by delaying her liquidation decision and increasing the target price she aims at realizing on the market. This effect, widely known and referred to as the ``real option effect" in D\'{e}camps et al.~\cite{Deca}, can be observed in the benchmark case of \eqref{Value function V mu0} as well as in the problem \eqref{Value function control problem } under partial information,  as Figure \ref{Figure: Comp Stat Sigma} reveals. 
\begin{figure}[h] 
\begin{center}
  \includegraphics[width=7.95cm]{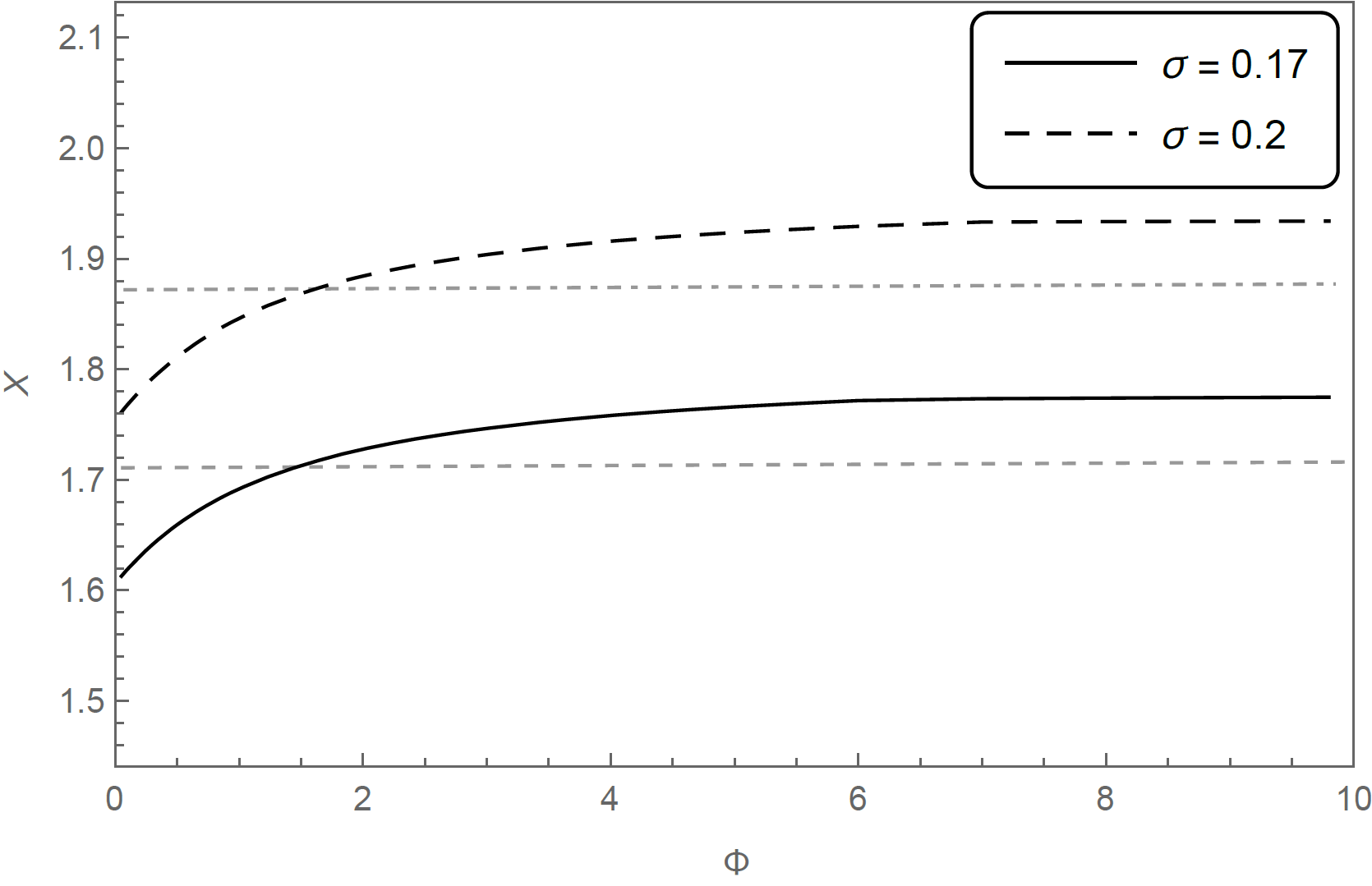}\includegraphics[width=8cm]{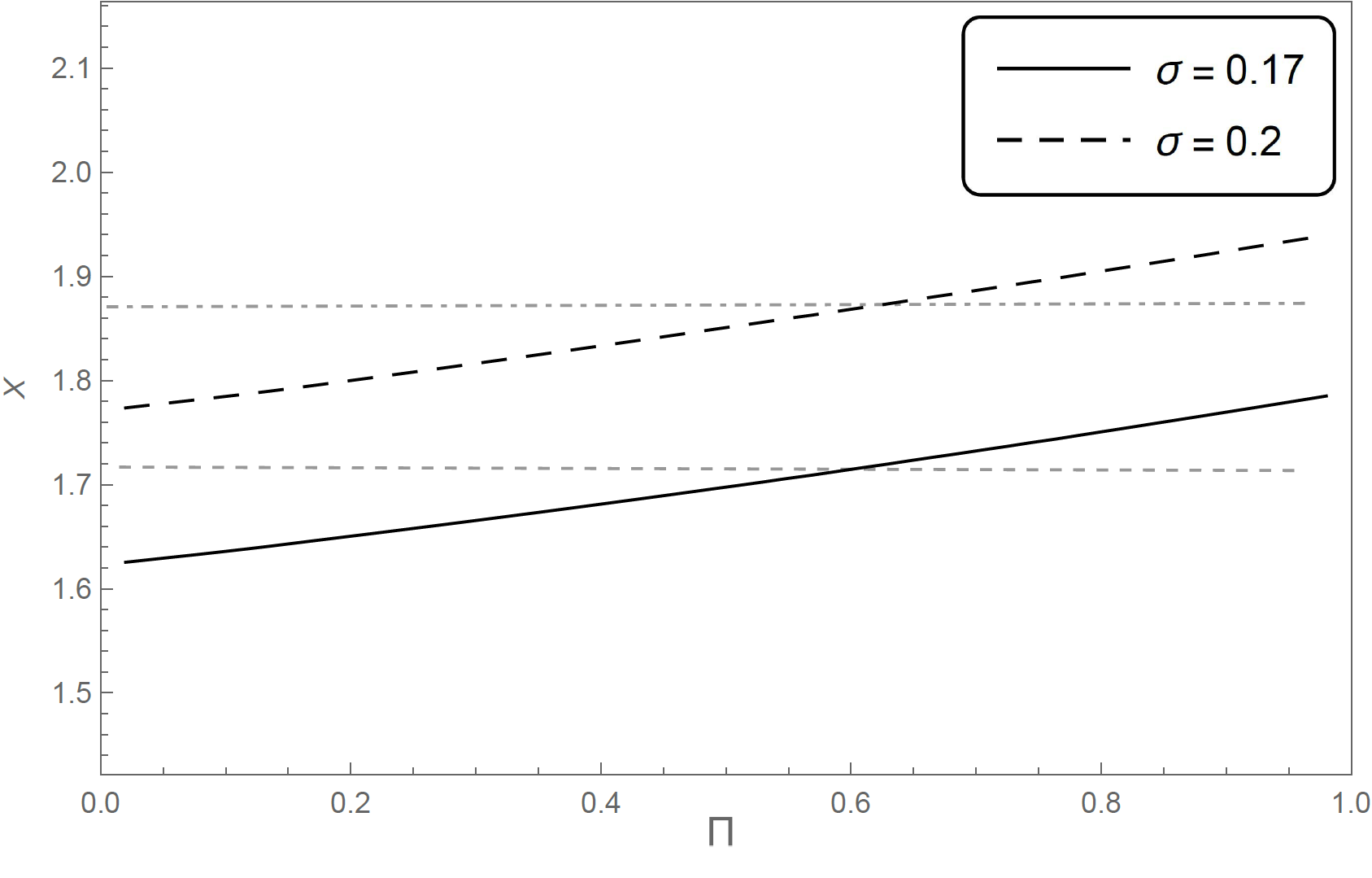} \caption{The optimal execution boundaries $b(\varphi)$ and $a(\pi)$ as well as the pre-committed strategies for different values of $\sigma$ and following parameters: $r =0.07 ,\, \mu_0 = - 0.01 , \, \mu_1 = 0.007 ,\, \kappa = 3,\, \pi = 0.6.$} \label{Figure: Comp Stat Sigma}
\end{center} 
\end{figure} \text{}\\
However, this effect does not need to be robust. To understand how an increase in volatility might indeed harm the investor, we recall the dynamics of the belief process $\Pi$, given by \eqref{dynamics control problem}. In particular, we observe that increasing volatility lowers the signal-to-noise ratio $\gamma = (\mu_1 - \mu_0)/\sigma$ (determining  the variance of the process $\Pi$) and thus the \textit{efficiency of learning}. The latter effect is in contrast to the mentioned real option effect, and the sensitivity of the value function with respect to an increase in volatility ``depends on which of the real option and the inefficient learning effect dominates" (D\'{e}camps et al.~\cite{Deca}, p. 487). The overall impact of a change in volatility thus clearly depends on the parameters' constellation of the model, a division of the parameters' space is however not straightforward. For a broader discussion on this subject we refer to D\'{e}camps et al.~\cite{Deca}, Section 6.2.  \\[0.15cm]
8.2 \textbf{The Value of Information.}
Here, we want to address the question on whether incomplete information about the drift actually harms or benefits the investor. To this end, we introduce the ``average drift problem", whose value is denoted by $V^A (x,y)$ and modelled as in \eqref{Value function V mu0}, but with constant and known drift $\pi \mu_0  + (1 - \pi)\mu_1$; i.e.\ the average of $\mu$ with respect to the prior Bernoulli distribution. We then investigate the preference of an investor faced with the decision of choosing between two portfolios containing assets with either an unknown drift coefficient, or with a constant and known \textit{average} drift. An analytical attempt to answer this question is presented in D\'{e}camps et al.~\cite{Deca}, although the derived result does not hold true in general, as pointed out by Klein \cite{Klein}. 
\begin{figure}[h] 
\begin{center}
  \includegraphics[width=7.95cm]{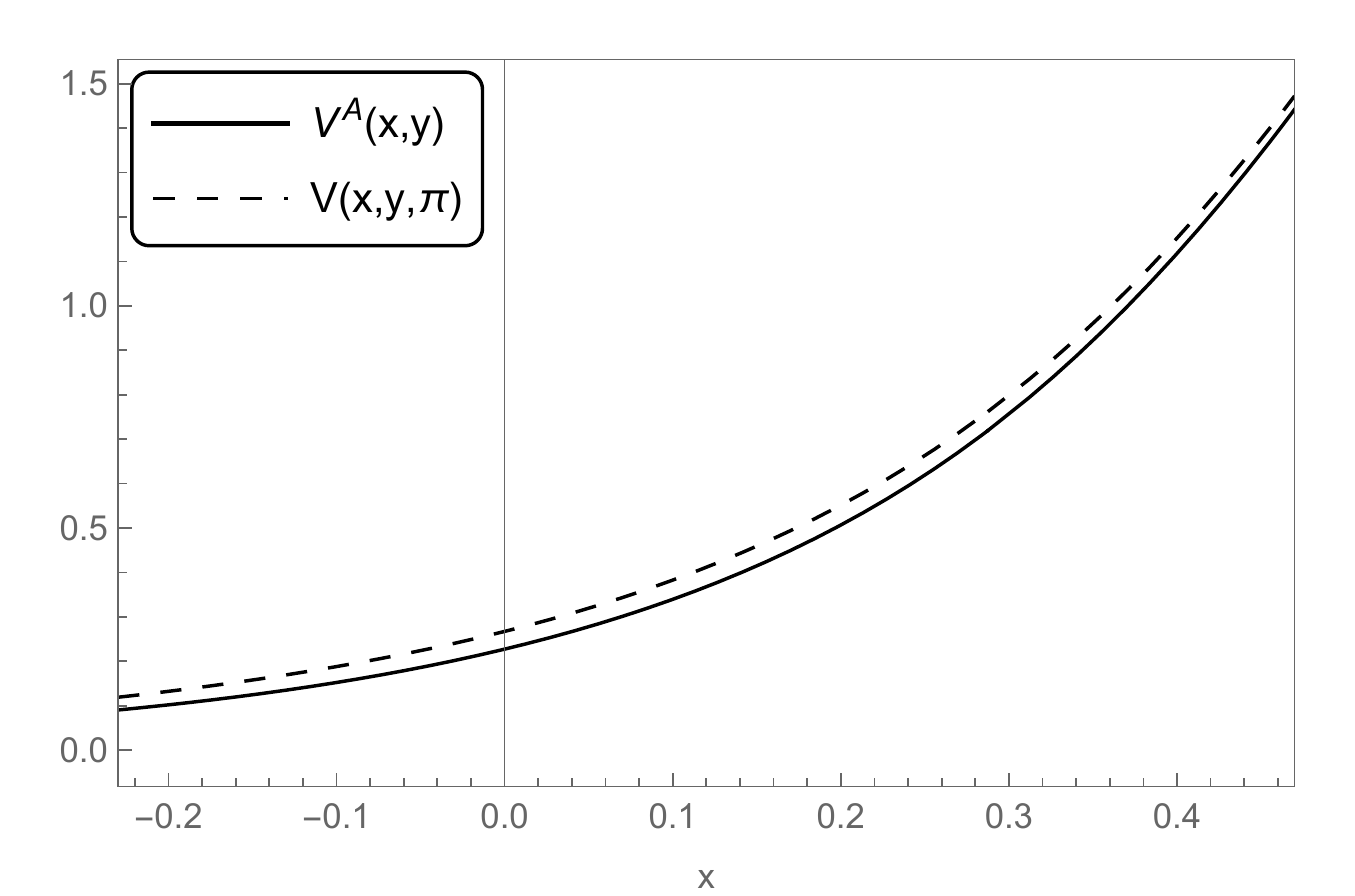}\includegraphics[width=8cm]{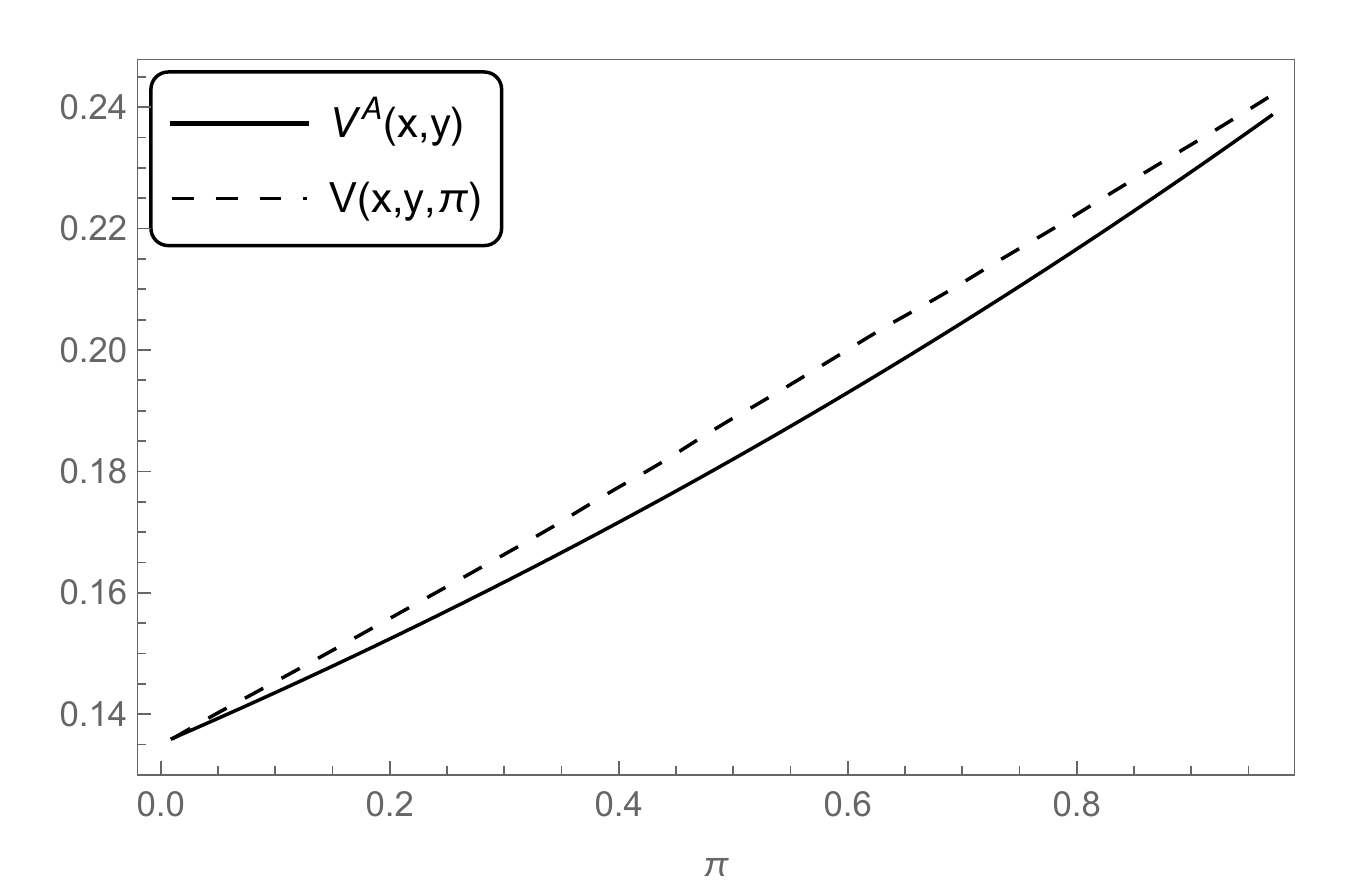} \caption{The value function $V$ of \eqref{Value function control problem } and the average drift value function $V^A$ as functions of $x$ and $\pi$, respectively. The parameters of the model have been specified as $r =0.15 ,\,\sigma = 0.15,\, \mu_0 = - 0.012 , \, \mu_1 = 0.01 ,\, \kappa = 1,\, \pi = 0.3, x = -0.1$} \label{Figure: Comp Stat Value1}
\end{center} 
\end{figure} 
Here, we are able to analyse this question with numerical methods based on the numerical evaluation of the optimal execution boundary (cf.\ Section 8.1) and the representation \eqref{Probabilistic Repr of vhat} of the optimal stopping value function $\hat{v}$. In order to accomplish that, we plug in the numerical evaluation of $b^{-1}$ into \eqref{Probabilistic Repr of vhat} and we transform the result according to \eqref{Transformation T}. This yields the value function $v$ of \eqref{Value fct of opt stopp}, which can be finally integrated via \eqref{Candidate Value Fct U} to obtain a numerical approximation of the control problem's value function $V$. \\
\begin{figure}[h] 
\begin{center}
  \includegraphics[width=7.95cm]{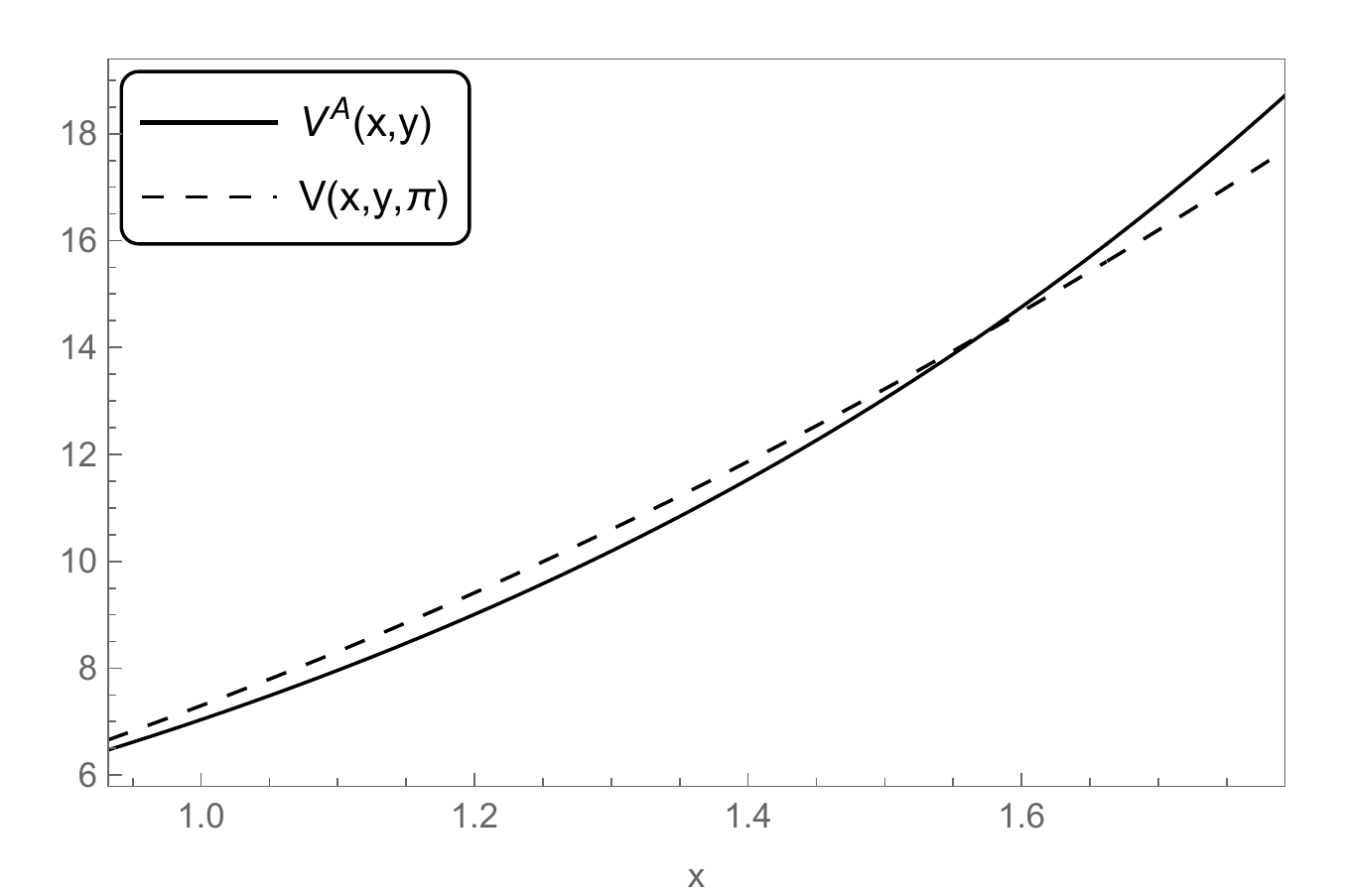} \caption{The value function $V$ of \eqref{Value function control problem } and the average drift value function $V^A$ as functions of $x$. The parameters of the model have been specified as $r =0.2 ,\,\sigma = 0.5,\, \mu_0 = - 0.012 , \, \mu_1 = 0.01 ,\, \kappa = 1,\, \pi = 0.5$} \label{Figure: Comp Stat Uncert}
\end{center} 
\end{figure} 
In general, the results derived in D\'{e}camps et al.~\cite{Deca} and Klein \cite{Klein} suggest that the overall impact of introducing uncertainty over the drift is governed by two separate effects: The introduction of uncertainty in general and the impact of learning. If learning is efficient, which is achieved by -- for example -- specifying a small volatility coefficient $\sigma$, the latter effect seems to outweighs the former and the investor indeed prefers the problem with only incomplete information on the return. We observe this overall effect in Figure \ref{Figure: Comp Stat Value1}.\\
In their model, D\'{e}camps et al.~\cite{Deca} give an analytical proof to this observation in an optimal stopping environment, although restricting the possible drift values to $0$ and $1$. For small values of $\sigma$, depending on the other parameters in the model, this result seems to hold true in our more generalized framework. \\
Nevertheless, this effect cannot be expected to be robust over the whole parameter space. In an example, where the parameter values are aligned such that $\beta_0 + \beta_1 = \sigma^2$ (and thus $\mu_0 = - \mu_1$ in our model), Klein \cite{Klein} obtains an explicit solution to the optimal stopping problem and shows how the introduction of uncertainty might harm the decision maker. This effect appears to have the peculiarity of being, at least in some cases, dependent on the initial value of the price process, as it determines the distance to the target price at which the investor is willing to execute. 
We can observe an example of this in Figure \ref{Figure: Comp Stat Uncert}. In particular, if the asset's price is close to the target value under the current belief and learning is inefficient, the investor will not choose a portfolio with drift uncertainty. This is due to the fact that the downside risk outweighs the upside potential, which could only be achieved if learning is efficient. On the other hand, we observe that for low prices the upside potential might still dominate and the investor is willing to choose the uncertain environment, even if learning is inefficient.

\appendix\normalsize 
\section{Proof of Proposition \ref{Proposition: Smooth fit}}\label{Appendix: Proof of Smooth fit}
The proof follows the lines of Section 4 in \cite{DeA}, suitably adapted to the present setting, and it is obtained through a series of intermediate results. Let $(x,z)\in \mathbb{R}^2$ be given and fixed and set
\begin{align}\label{Stopping times for regularity}
\sigma_* := \sigma_* (x,z) := \inf \{ t \geq 0: ~(X_t^x, Z_t^z) \in \mathcal{S}_3 \}, \quad 
\hat{\sigma}_* := \hat{\sigma}_* (x,z) := \inf \{ t \geq 0: ~(X_t^x, Z_t^z) \in \text{int}(\mathcal{S}_3) \},
\end{align} 
and observe that $\sigma_* = \tau_*$ $\mathbb{Q}$-a.s.\,on $\mathbb{R}^2 \setminus \partial \mathcal{C}_3$ due to the continuity of paths. It is crucial to show that this equality also holds for the boundary points $(x_0, z_0) \in \partial \mathcal{C}_3$. As it turns out, the cases i) $\mu_0 + \mu_1 \geq 0$ and ii) $\mu_0 + \mu_1 < 0$ should be treated in different fashions and the latter case exhibits some more technical difficulties than the first case. Let us start with case i), in which the needed result  follows upon using the law of iterated logarithm. 
\begin{proposition} \label{Proposition: Cont of stopping times in case i)} Assume that $\mu_0 + \mu_1 \geq 0$. Let $(x_n,z_n) \in \mathcal{C}_3$ be a sequence with $(x_n, z_n) \to (x_0,z_0) \in \partial \mathcal{C}_3 $, such that $x_0 = c(z_0)$. We then have $\tau^* (x_n, z_n) \downarrow 0$ as well as $\hat{\sigma}_* (x_n, z_n) \downarrow 0$ $\mathbb{Q}$-a.s.
\end{proposition}

\begin{proof}
Fix $\omega \in \Omega$ and assume that $\limsup_{n \to \infty} \tau^* (x_n, z_n) (\omega) =: \delta > 0$. Hence, there exists a subsequence (still labelled by $(x_n,z_n)$) such that 
\begin{align}\label{X lower c(Z)}
X_t^{x_n} (\omega) < c(Z_t^{z_n}) \quad \forall n \in \mathbb{N},~ \forall t \in [0, \delta/2],
\end{align}
which is equivalent to 
\begin{align*}
x_n + \mu_0 t + \sigma B_t (\omega) < c(z_n - \frac{1}{2}(\mu_0 + \mu_1) t) \quad \forall n \in \mathbb{N}, \forall t \in [0, \delta/2].
\end{align*}
Upon using that $z \mapsto c (z)$ is continuous, we let $n \to \infty$ and obtain 
\begin{align}\label{Inequ cont of stop times i)}
\sigma B_t (\omega) \leq  c(z_0 - \frac{1}{2}(\mu_0+ \mu_1)t) - x_0 - \mu_0 t \leq c(z_0) - x_0 - \mu_0 t = - \mu_0 t \quad \forall t \in [0, \delta/2],
\end{align}
where the last inequality follows from $\mu_0 + \mu_1 \geq 0$ and Proposition \ref{Proposition Mono 1}. 
On the other hand, by the law of iterated logarithm, there exists a sequence $(t_n) \downarrow 0 $ for all $ \epsilon > 0$ such that 
\begin{align}\label{law iter log}
B_{t_n} \geq (1- \epsilon) \sqrt{2 t_n \log \big( \log \big(\frac{1}{t_n}\big)\big)} \qquad  \forall n \in \mathbb{N}.
\end{align}
Combining \eqref{Inequ cont of stop times i)} and \eqref{law iter log} implies
\begin{align*}
\frac{1}{t} \sigma (1- \epsilon) \sqrt{2 t \log \big( \log \big( \frac{1}{t} \big) \big) }  \leq - \mu_0,
\end{align*}
but since $\sqrt{2 t \log ( \log (1/ t))}/t \to \infty$ for $t \downarrow 0$, (\ref{X lower c(Z)}) can only happen on a $\mathbb{Q}$-null set. Thus $\tau^*(x_n,z_n) \downarrow 0$ and by replacing the strict inequality in (\ref{X lower c(Z)}) by "$\leq$", we obtain that $\hat{\sigma}_* (x_n ,z_n) \downarrow 0$ as well. \qed 
\end{proof}
Notice that the proof of Proposition \ref{Proposition: Cont of stopping times in case i)} cannot be replicated for the case ii), in which $\mu_0 + \mu_1 < 0$, since the last inequality in \eqref{Inequ cont of stop times i)} does not longer apply. As is turns out, in order to prove the same result for case ii), we have to take a longer route. The reason for this lies in the fact that the process  $(X,Z)$ is moving towards the right in the state space and hence - keeping in mind that the continuation region $\mathcal{C}_3$ of \eqref{Continuation region C3 in terms of c} lies below the increasing boundary $c$ - could possibly evade from the stopping set. In the following, we  show that this is not the case by adapting the procedure in of Section 4 in De Angelis \cite{DeA}. As a first step, we state the following Lemma, whose proof follows the lines of Cox and Peskir \cite{CoPe}, Corollary 8, and is thus omitted for the sake of brevity.
\begin{lemma}\label{Lemma: Equality of Stopping times}
Assume that $\mu_0 + \mu_1 < 0$ and $r > \frac{\gamma}{2 \sigma} \vert \mu_0 + \mu_1 \vert$. We have $\mathbb{Q}( \sigma_* = \hat{\sigma}_*)=1$.
\end{lemma}
In the next step, we aim at proving \textit{regularity} of the boundary points for the stopping set $\mathcal{S}_3$ in the sense of diffusions, that is, for $(x,z) \in \partial \mathcal{C}_3$ we have 
\begin{align}\label{Q = 0 reg of bound points}
\mathbb{Q}_{x,z}(\sigma_* > 0) =0.
\end{align}
It is clear from Blumenthal's 0-1 law that if \eqref{Q = 0 reg of bound points} does not hold, we have $\mathbb{Q}_{x,z} (\sigma_* > 0 ) =1$. Due to the mentioned structure of the problem this is not a straightforward task, since we cannot apply an argument similar to the one on Proposition \ref{Proposition: Cont of stopping times in case i)}. Instead, we establish the result in two steps and begin by showing that the classical smooth-fit property holds at the free-boundary, i.e.\,continuity of $\hat{v}_x ( \cdot ,z)$.
\begin{lemma} \label{Lemma: Classical smooth fit}
Assume that $\mu_0 + \mu_1 < 0$ and $r > \frac{\gamma}{2 \sigma} \vert \mu_0 + \mu_1 \vert$. For $\hat{v}$ of \eqref{Value function parabolic formulation} we have $\hat{v}_x (\cdot\, , z) \in C(\mathbb{R})$, or, equivalently, $\hat{u}_x (\cdot \, , z) \in C(\mathbb{R})$ for $\hat{u}$ of \eqref{u hat}.
\end{lemma}
\begin{proof}
From \eqref{L -r u hat} we obtain 
\begin{align*}
\frac12 \sigma^2 \hat{u}_{xx} (x,z) &= r \hat{u} (x,z) - \mu_0 \hat{u}_x (x,z) + \frac12 (\mu_0 + \mu_1 )\hat{u}_z (x,z)  - g (x,z),
\end{align*}
for $(x,z) \in \mathcal{C}_3$, and due to \eqref{vbar lipschitz} (which implies an analogous result for $\hat{v}$) we deduce that for a bounded set $B$, we must have that $\hat{u}_{xx}$ is bounded on the closure of $B \cap \mathcal{C}_3$. Moreover, we recall that $\hat{u}_x \leq 0$ in $\mathcal{C}_3$, as verified in the proof of Proposition \ref{Proposition Mono 2}. Aiming for a contradiction we now assume that for $(x_0, z_0) \in \partial \mathcal{C}_3$, such that $x_0 = c(z_0)$, we have
\begin{align}\label{hat ux < delta for smooth fit}
\hat{u}_x (x_0 - ,z_0) < - \delta_0,
\end{align}
for some $\delta_0 > 0$. We now take a bounded rectangular neighbourhood of $(x_0,z_0)$ and let $\tau_B := \inf \{ t > 0: ~ (X_t,Z_t) \notin B \}$. Notice that 
\begin{align}\label{Inequality for SmFi}
\hat{u}(x_0 ,z_0)  \geq \mathbb{E}^\mathbb{Q}_{(x_0,z_0)} \Big[ e^{-r (\tau_B \wedge t)} \hat{u}(X_{\tau_B \wedge t}, Z_{\tau_B \wedge t} ) + \int_0^{\tau_B \wedge t} e^{-rs} g(X_s, Z_s)  ds \Big],
\end{align}
from the supermartingale property of $(e^{-rt}\hat{v} (X_t, Z_t))_t$. Recall Lemma \ref{Lemma: uhat z geq 0 in C} and since $t \mapsto Z_{\tau_B \wedge t}$ is increasing, we have $\hat{u}(X_{\tau_B \wedge t}^{x_0} , Z_{\tau_B \wedge t}^{z_0} ) \geq \hat{u}(X_{\tau_B \wedge t}^{x_0} , z_0 )$ $\mathbb{Q}$-a.s. Moreover, since the integrand on the right-hand side of \eqref{Inequality for SmFi} is bounded on $B$, we obtain
\begin{align}\label{equa in smfit}
\hat{u}(x_0,z_0) \geq \mathbb{E}^\mathbb{Q}_{(x_0,z_0)} \Big[ e^{-r (\tau_B \wedge t)} \hat{u}(X_{\tau_B \wedge t} , z_0) - c_B (\tau_B \wedge t) \Big],
\end{align}
where $c_B$ is a constant depending on $B$. Due to the previously discussed local boundedness of $\hat{u}_{xx}$, we can apply It\^{o}-Tanaka's formula to the first term in the expectation of \eqref{equa in smfit}. Let $\mathcal{L}_X := \frac12 \sigma^2 \partial_{xx} + \mu_0 \partial_x$ and denote the local time of $X$ at $x_0$ by $L^{x_0}$. Moreover, noticing that $\hat{u}_{xx} (\cdot \, ,z_0) = 0$ for $x > x_0$, we obtain 
\begin{align*}
\mathbb{E}^\mathbb{Q}_{(x_0,z_0)} \Big[ e^{-r (\tau_B \wedge t)} \hat{u}(X_{\tau_B \wedge t} ,z_0) \Big] 
&= \hat{u}(x_0,z_0) + \mathbb{E}^\mathbb{Q}_{(x_0,z_0)} \Big[ \int_0^{\tau_B \wedge t} e^{-rs} (\mathcal{L}_{X} -r ) \hat{u}(X_s ,z_0) \one_{ \{X_s \neq x_0 \} } ds \Big] \\
&\hspace*{2cm} - \mathbb{E}^\mathbb{Q}_{(x_0,z_0)} \Big[ \int_0^{\tau_B \wedge t} e^{-rs} \hat{u}_x (x_0 - ,z_0) d L_s^{x_0} \Big],
\end{align*}
and, combining this with (\ref{equa in smfit}), as well as noticing that $(\mathcal{L}_X -r) \hat{u}(X_s ,Z_s)$ is bounded on $B$, we find
\begin{align*}
0 &\geq \mathbb{E}^\mathbb{Q}_{(x_0,z_0)} \Big[ \int_0^{\tau_B \wedge t} e^{-rs} (\mathcal{L}_X -r) \hat{u}(X_s , z_0) \one_{ \{ X_s \neq x_0 \} } ds - c_B (\tau_B \wedge t) \Big] \\
&\hspace*{2cm} - \mathbb{E}^\mathbb{Q}_{(x_0,z_0)} \Big[ \int_0^{\tau_B \wedge t} e^{-rs} \hat{u}_x (x_0 - ,z_0 ) d L_s^{x_0} \Big] \\
&\geq \delta_0 e^{-rt} \mathbb{E}^\mathbb{Q}_{(x_0,z_0)} [ L_{\tau_B \wedge t}^{x_0} ] - c_B \mathbb{E}^\mathbb{Q}_{(x_0,z_0)} [ \tau_B \wedge t ],
\end{align*}
where we used our assumption \eqref{hat ux < delta for smooth fit} in the last inequality. Since this is equivalent to $c_B \mathbb{E}^\mathbb{Q}_{(x_0,z_0)} [\tau_B \wedge t] \geq \delta_0 e^{-rt} \mathbb{E}^\mathbb{Q}_{(x_0,z_0)} [L_{\tau_B \wedge t}^{x_0}]$, and $\mathbb{E}^\mathbb{Q}_{(x_0,z_0)} [\tau_B \wedge t] \approx t$ while $\mathbb{E}^\mathbb{Q}_{(x_0,z_0)} [L_{\tau_B \wedge t}^{x_0} ] \approx \sqrt{t}$ (see, e.g., Peskir \cite{PeCont}, Lemma 15), we obtain the desired contradiction. Hence, $\hat{u}_x (\cdot\, ,z) \in C(\mathbb{R})$. \qed
\end{proof}
We can now state the regularity of the boundary points.
\begin{proposition}\label{Proposition: Boundary points regular}
Assume that $\mu_0 + \mu_1 < 0$ and $r > \frac{\gamma}{2 \sigma} \vert \mu_0 + \mu_1 \vert$. All points $(x,z) \in \partial \mathcal{C}_3$ are regular, i.e.\,we have $\mathbb{Q}_{x,z}(\sigma_* > 0)=0$.
\end{proposition}
\begin{proof}
We argue by contradiction and show that if $\mathbb{Q}_{(x_0,z_0)}(\sigma_* > 0)=1$ for some boundary point $(x_0,z_0) \in \partial \mathcal{C}_3$ it follows that $\hat{u}_x (x_0 - ,z_0) < 0$, which contradicts Lemma \ref{Lemma: Classical smooth fit}. As a first step, we establish an upper bound for $\hat{u}_x$. Fix $(x,z) \in \mathcal{C}_3$ such that $x > \tilde{x}$, with the latter given by \eqref{x tilde for monotonicity}. Define $\tau_\epsilon := \tau_\epsilon (x) := \inf \{ t\geq 0:~ X_t^x = \tilde{x} + \epsilon \}$
and observe that - by strong Markov property - we have
\begin{align}\label{eq in upperb}
\hat{u}(x,z) 
&= \sup_\tau \mathbb{E}^\mathbb{Q}_{(x,z)} \Big[  e^{-r \tau_\epsilon} \hat{u}(\tilde{x} + \epsilon , Z_{\tau_\epsilon} ) \one_{ \{ \tau > \tau_\epsilon \} } + \int_0^{\tau_\epsilon \wedge \tau } e^{-rt} g(X_t, Z_t) dt \Big].
\end{align}
Moreover, we let $\tilde{\tau} := \tilde{\tau}(x) := \inf \{ t > 0:~X_t^x = \tilde{x} \}$, and for $\tau' := \tau^* (x,z)$ we obtain
 \begin{align}\label{Inequ regular boundary points}
\hat{u}(x- \epsilon,z) 
&= \mathbb{E}^\mathbb{Q}_{(x - \epsilon,z)} \Big[ e^{-r \tilde{\tau} (x-\epsilon) } \hat{u} (\tilde{x} , Z_{\tilde{\tau} (x-\epsilon)} ) \one_{ \{ \tau' > \tilde{\tau} (x-\epsilon) \} } + \int_0^{\tau' \wedge \tilde{\tau} (x-\epsilon)} e^{-rt} g(X_t, Z_t) dt \Big].
\end{align}
Notice that $\tau_\epsilon (x) = \tilde{\tau} (x-\epsilon)$. Hence, subtracting \eqref{Inequ regular boundary points} from (\ref{eq in upperb}) yields
\begin{align*}
\hat{u}(x,z) - \hat{u}(x - \epsilon,z) 
& = \mathbb{E}^\mathbb{Q} \Big[  e^{-r \tau_\epsilon} \big( \hat{u}(\tilde{x} + \epsilon , Z_{\tau_\epsilon}^z) - \hat{u}(\tilde{x}, Z_{\tau_\epsilon}^z ) \big) \one_{ \{ \tau' > \tau_\epsilon \} } \Big]  \\
&\hspace*{1cm} + \mathbb{E}^\mathbb{Q} \Big[ \int_0^{\tau_\epsilon \wedge \tau'} e^{-rt}  \big( g(X_t^x , Z_t^z) - g(X_t^{x- \epsilon}, Z_t^z)  \big) dt \Big].
\end{align*}  
Since $(\tilde{x} + \epsilon , Z_{\tau_\epsilon}^z ) \in \mathcal{C}_3 $ on $\{\tau' > \tau_\epsilon \}$  and $\hat{u}_x \leq 0$ in $\mathcal{C}_3$ (see Proposition \ref{Proposition Mono 2}), we must have 
\begin{align*}
\hat{u}(\tilde{x}, Z_{\tau_\epsilon}^z) \geq \hat{u} (\tilde{x} + \epsilon, Z_{\tau_\epsilon}^z ) ,
\end{align*}
and we obtain 
\begin{align*}
\hat{u}(x,z) - \hat{u}(x - \epsilon,z) 
&\leq 
 \mathbb{E}^\mathbb{Q} \Big[ \int_0^{\tau_\epsilon \wedge \tau'} e^{-rt}  \big( g(X_t^x , Z_t^z) - g(X_t^{x- \epsilon}, Z_t^z) \big) dt \Big]. 
\end{align*}
If we now divide by $\epsilon > 0$ and let $\epsilon \downarrow 0$, we obtain (since $\tau_\epsilon \downarrow \tilde{\tau}$ and $\tau' = \tau^* (x,z)$) 
\begin{align*}
\hat{u}_x (x,z) &\leq 
 \mathbb{E}^\mathbb{Q} \Big[ \int_0^{\tilde{\tau} \wedge \tau'} e^{-rt} g_x (X_t^x , Z_t) dt \Big].
\end{align*}
In the next step, we assume by contradiction that there exists $(x_0,z_0) \in \partial \mathcal{C}_3$ with $\mathbb{Q}_{x_0,z_0}(\sigma_* > 0)=1$ and take an increasing sequence $x_n \uparrow x_0$ such that $x_n > \tilde{x}$ for all $n \in \mathbb{N}$, which is possible due to Assumption \ref{Assumption: For monotonicity of c}. Let $\tau_n := \tau^* (x_n, z_n)$ and notice that $\tau_n = \sigma_n := \sigma_* (x_n, z_0)$ for all $n \in \mathbb{N}$ due to continuity of paths. Furthermore, $\sigma_n$ decreases in $n$ and $\sigma_n \geq \sigma_* := \sigma_* (x_0, z_0)$, since $x \mapsto X_t^x$ is increasing. Set $\tilde{\tau}^n := \tilde{\tau}(x_n)$ and notice that $\tilde{\tau}^n \uparrow \tilde{\tau}$. Moreover, we let $\sigma^\infty := \lim_{n \to \infty} \sigma_n$ and have
\begin{align*}
\sigma^\infty \wedge \tilde{\tau} = \lim_{n \to \infty} (\sigma_n \wedge \tilde{\tau}^n ) \geq \sigma_* \wedge \tilde{\tau} \qquad \mathbb{Q}\text{-a.s.}
\end{align*}
We then obtain 
\begin{align*}
\hat{u}_x (x_0- , z_0) 
= \lim_{n \to \infty} \hat{u}_x (x_n ,z_0) 
&\leq  \lim_{n \to \infty}  \mathbb{E}^\mathbb{Q} \Big[ \int_0^{\tilde{\tau} \wedge \sigma_n } e^{-rt} g_x (X_t^{x_n} , Z_t^{z_0} ) dt \Big]  \\
&=  \mathbb{E}^\mathbb{Q} \Big[ \int_0^{\tilde{\tau} \wedge \sigma^\infty } e^{-rt} g_x (X_t^{x_0} , Z_t^{z_0} ) dt \Big]
< 0,
\end{align*}
where we used $x_0 > \tilde{x}$ as well as $\tilde{\tau} \wedge \sigma^\infty > 0$ due to our assumption $\mathbb{Q}_{x_0 , z_0}(\sigma^\infty \geq \sigma^* > 0) =1$. But this contradicts Lemma \ref{Lemma: Classical smooth fit} and the claim follows. \qed 
\end{proof}
As a corollary of Lemma \ref{Lemma: Equality of Stopping times} and Proposition \ref{Proposition: Boundary points regular} we obtain
\begin{corollary}\label{Corollary Equality of stop times}
Assume that $\mu_0 + \mu_1 < 0$ and $r > \frac{\gamma}{2 \sigma} \vert \mu_0 + \mu_1 \vert$. Then, for all $(x,z) \in \mathbb{R}^2$ we have 
\begin{align*}
\mathbb{Q}_{x,z} ( \tau^* = \sigma_* = \hat{\sigma}_* ) = 1.
\end{align*}
\end{corollary}
This result allows us to state the continuity result of the optimal stopping time with respect to the initial data.
\begin{lemma} \label{Lemma: Continuity of stop times}
Assume that $\mu_0 + \mu_1 < 0$ and $r > \frac{\gamma}{2 \sigma} \vert \mu_0 + \mu_1 \vert$. We have $\lim_{n \to \infty} \tau^* (x_n,z_n) = \tau^* (x,z)$ for any $(x,z) \in \mathbb{R}^2$ and any sequence $(x_n,z_n) \to (x,z)$. In particular, if $(x,z)\in \partial\mathcal{C}_3$, the limit is zero.
\end{lemma}
\begin{proof}
Let $(x,z) \in \mathbb{R}^2$ and denote $\tau_n := \tau^* (x_n , z_n)$ as well as $\tau := \tau^* (x,z)$ for simplicity. In order to show lower-semicontinuity, we fix $\omega \in \Omega$ ouside of a null-set. For $\tau(\omega) = 0$ we are finished and thus assume $\tau (\omega) > \delta > 0$. Due to Proposition \ref{Proposition boundary c continuous} there exists $k_{\delta, \omega} > 0$ such that 
\begin{align*}
c(Z_t (\omega) ) - X_t (\omega) > k_{\delta, \omega} ,
\end{align*}
for all $t \in [0, \delta]$.
The map $(t,x,z) \mapsto c(Z_t^{z} (\omega)) - X_t^x (\omega)$ is uniformly continuous on any compact $[0, \delta ] \times K$, hence we can find $N_\omega \geq 1$ such that for all $n \geq N_\omega$ and $t \in [ 0 ,\delta ]$
\begin{align*}
c(Z_t^{z_n} (\omega)) - X_t^{x_n} ( \omega) > k_{\delta, \omega},
\end{align*}
and therefore $\liminf_n \tau_n (\omega) \geq \delta$. Since $\omega$ and $\delta$ were arbitrary, we obtain $\liminf_n \tau_n \geq \tau$ $\mathbb{Q}$-a.s. and thus lower-semicontinuity. By employing similar arguments we can show $\limsup_n \hat{\sigma}_n \leq \hat{\sigma}$ $\mathbb{Q}$-a.s.\,and the claim thus follows together with Corollary \ref{Corollary Equality of stop times}. \qed
\end{proof}
Before we finally state the proof of Proposition \ref{Proposition: Smooth fit}, we can derive a probabilistic representation of $v_x$ by employing  arguments similar to those employed in the proof of Lemma \ref{Lemma: hat v_z}.
\begin{lemma}
For all $(x,z) \in \mathbb{R}^2 \setminus \partial \mathcal{C}_3$, we have 
\begin{align*}
\hat{v}_x (x,z) = \mathbb{E}^\mathbb{Q}_{(x,z)} \Big[ e^{-r \tau^*} \Big( e^{X_{\tau^*}} (1 + e^{\frac{\gamma}{\sigma} (X_{\tau^*} + Z_{\tau^*} )} ) + \frac{\gamma}{\sigma} (e^{X_{\tau^*} } - \kappa) e^{\frac{\gamma}{\sigma} (X_{\tau^*} + Z_{\tau^*} )} ) \Big)\one_{ \{ \tau^* < \infty \} } \Big].
\end{align*}
\end{lemma}
We are therefore ready to prove Proposition \ref{Proposition: Smooth fit}.
\\[0.2cm]
\textbf{Proof of Proposition \ref{Proposition: Smooth fit}.} The first statement trivially holds true for $(x,z) \in$ int($\mathcal{S}_3$) and $(x,z) \in \mathcal{C}_3$, due to the result in Lemma \ref{Lemma: vhat solves boundary value problem}. It thus remains to prove that $\triangledown_{x,z} \hat{v}$ is continuous across the boundary $\partial \mathcal{C}_3$. Let $(x_0,z_0) \in \partial \mathcal{C}_3$ and take a sequence $(x_n, z_n) \to (x_0 ,z_0)$ with $\tau_n := \tau^* (x_n,z_n)$. For a fixed $t > 0$, we notice $(X_t, Z_t )\in \mathcal{C}_3$ on $\{ \tau_n > t \}$ and thus, 
upon using tower and Markov property, we obtain 
\begin{align*}
\hat{v}_x (x_n, z_n) 
&=\mathbb{E}^\mathbb{Q}_{(x_n,z_n)} \Big[ e^{-r \tau_n} \Big( e^{X_{\tau_n}} (1 + e^{\frac{\gamma}{\sigma} (X_{\tau_n} + Z_{\tau_n} )} ) + \frac{\gamma}{\sigma} (e^{X_{\tau_n}} - \kappa) e^{\frac{\gamma}{\sigma} (X_{\tau_n} +Z_{\tau_n} )} \Big) \one_{ \{\tau_n \leq t \}} \Big] \\ 
& \hspace*{2cm} + \mathbb{E}^\mathbb{Q}_{(x_n,z_n)} \Big[ e^{-rt} \hat{v}_x (X_t , Z_t ) \one_{ \{ \tau_n > t \} } \Big].
\end{align*}
Due to Assumption \ref{Assumption: Well-posedness} we can invoke dominated convergence as well as Lemma \ref{Lemma: Continuity of stop times} to obtain
\begin{align*}
\lim_{n \to \infty} \hat{v}_x (x_n, z_n) 
&= e^{x_0} (1+ e^{\frac{\gamma}{\sigma}(x_0 + z_0)}) + \frac{\gamma}{\sigma} (e^{x_0} - \kappa) e^{\frac{\gamma}{\sigma}(x_0 + z_0 )} 
= \frac{\partial}{\partial x} \Big( (e^x - \kappa) (1 + e^{\frac{\gamma}{\sigma}(x+z)} ) \Big) \Big\vert_{(x_0,z_0)}, 
\end{align*}
and hence, the continuity of $\hat{v}_x$ across the optimal boundary. The continuity of $\hat{v}_z$ across the free boundary follows similarly. For the last claim we observe that Lemma \ref{Lemma: vhat solves boundary value problem} implies 
\begin{align}\label{Equality for smooth fit}
\frac12 \sigma^2 \hat{v}_{xx}(x,z) = r \hat{v}(x,z) - \mu_0 \hat{v}_x (x,z) + \frac12 (\mu_0 + \mu_1) \hat{v}_z (x,z),
\end{align}
for all $(x,z) \in \mathcal{C}_3$. But the right-hand side of \eqref{Equality for smooth fit} only involves functions which are continuous on $\mathbb{R}^2$, hence we deduce that $\hat{v}_{xx}$ admits a continuous extension on $\overline{\mathcal{C}}_3$ and is therefore bounded therein. It follows that $\hat{v}_x ( \cdot \, , z)$ is locally Lipschitz continuous on $\overline{\mathcal{C}}_3$, with a Lipschitz constant $K(z)$ that is locally bounded on $\mathbb{R}$. Now, because $\hat{v}_x( \cdot \, , z)$ is infinitely many times continuously differentiable in the stopping region $\mathcal{S}_3$ (and hence locally bounded therein as well), we conclude that $\hat{v}_{xx} \in L_{\text{loc}}^\infty (\mathbb{R}^2) $. \qed

\section{Proof of Proposition \ref{Proposition: Prob Repr of hatv} }\label{Appendix: Proof of Proposition Prob Repr of hatv}
\begin{proof}
Let $R >0$ and define $\tau_R := \inf \{ t \geq 0: ~ \vert X_t \vert \geq R ~ \text{or} ~ \vert Z_t \vert \geq R \}$. Since $\hat{v} \in C^1 (\mathbb{R}^2 )$ and $\hat{v}_{xx} \in L_{\text{loc}}^\infty (\mathbb{R}^2)$, we can apply a weak version of Ito's Lemma (see, e.g., Bensoussan and Lions \cite{BeLi}, Lemma 8.1 and Th. 8.5, pp. 183-186) up to the stopping time $\tau_R \wedge T$ for some $T > 0$, which results in 
\begin{align}\label{integr equ ito 1 }
\hat{v}(x,z) = \mathbb{E}^\mathbb{Q}_{(x,z)} \Big[ e^{-r (\tau_R \wedge T)} \hat{v} (X_{\tau_R \wedge T} , Z_{\tau_R \wedge T} ) - \int_0^{\tau_R \wedge T} e^{-rs} (\mathcal{L}_{X,Z} -r) \hat{v} (X_s , Z_s) ds \Big].
\end{align}
The right-hand-side of \eqref{integr equ ito 1 } is well-defined, because $Z$ is deterministic, $X$ has an absolutely continuous transition density and $\mathcal{L}_{X,Z} \hat{v}$ is defined up to a set of zero Lebesgue measure. Since $\hat{v}$ solves the free-boundary problem \eqref{Free Boundary Problem}, we have
\begin{align*}
(\mathcal{L}_{X,Z} -r) \hat{v} (x,z) &=(\mathcal{L}_{X,Z} -r) \hat{v}(x,z) \one_{ \{x < c(z) \} }  + (\mathcal{L}_{X,Z} -r) \hat{v}(x,z) \one_{ \{ x \geq c(z) \} }  
=  g(x,z) \one_{ \{ x \geq c(z) \} } ,
\end{align*}
for almost all $(x,z) \in \mathbb{R}^2$. Using again that the transition density of $X$ is absolutely continuous with respect to the Lebesgue measure, equation (\ref{integr equ ito 1 }) becomes 
\begin{align*}
\hat{v}(x,z) 
= \mathbb{E}^\mathbb{Q}_{(x,z)} \Big[ e^{-r (\tau_R \wedge T)} \hat{v} (X_{\tau_R \wedge T}, Z_{\tau_R \wedge T} ) - \int_0^{\tau_R \wedge T} e^{-rs} g(X_s , Z_s) \one_{ \{ x \geq c(z) \} } ds \Big].
\end{align*}
Now, upon employing a change of measure as in Section \ref{Section: Decoupling Change of Measure}, we obtain
\begin{align}\label{Expectation for Proof}
\mathbb{E}^\mathbb{Q}_{(x,z)} \big[ e^{-r (\tau_R \wedge T)} \vert \hat{v} (X_{\tau_R \wedge T}, Z_{\tau_R \wedge T} ) \vert \big] 
&= \mathbb{E}^\mathbb{Q}_{(x,z)} \big[  e^{-r (\tau_R \wedge T)} \vert \overline{v} (X_{\tau_R \wedge T}, e^{\frac{\gamma}{\sigma}(X_{\tau_R \wedge T} + Z_{\tau_R \wedge T} )} ) \vert \big] \nonumber \\
&\leq K_1 \mathbb{E}^\mathbb{Q}_{(x, \exp(\frac{\gamma}{\sigma} (x+z))} \big[  e^{-r (\tau_R \wedge T)} e^{X_{\tau_R \wedge T} }(1 + \Phi_{\tau_R \wedge T} ) \big] \nonumber \\
&= K_1 (1 + e^{\frac{\gamma}{\sigma}(x+ z)} ) \mathbb{E}_{(x,\pi)} \big[ e^{-r ( \tau_R \wedge T)} e^{X_{\tau_R \wedge T}} \big],
\end{align}
where $\pi = e^{\frac{\gamma}{\sigma}(x+ z)}/(1 + e^{\frac{\gamma}{\sigma}(x+ z)})$. Due to Assumption \ref{Assumption: Well-posedness}, it is easy to verify that taking limits in \eqref{Expectation for Proof} yields
\begin{align}\label{lim lim Exp}
\lim_{T \uparrow \infty} \lim_{R \uparrow \infty} \mathbb{E}^\mathbb{Q}_{(x,z)} \Big[ e^{-r (\tau_R \wedge T)}  \hat{v}(X_{\tau_R \wedge T }, Z_{\tau_R \wedge T } ) \Big] = 0.
\end{align}
Furthermore, 
\begin{align}\label{Expectation 2 for Proof}
\mathbb{E}^\mathbb{Q}_{(x,z)} \Big[  &\int_0^{\tau_R \wedge T} e^{-rs} g(X_s , Z_s) \one_{ \{ x \geq c(z) \} } ds \Big] 
\leq \mathbb{E}^\mathbb{Q}_{(x,z)} \Big[  \int_0^\infty e^{-rs} \vert g(X_s , Z_s) \vert  ds \Big] \nonumber \\
&\leq \mathbb{E}^\mathbb{Q}_{(x, \exp( \frac{\gamma}{\sigma}(x+z))} \Big[  \int_0^\infty e^{-rs} \Big( e^{X_s} ( r - \frac12 \sigma^2 - \mu_0 ) + rk + \Phi_s \big( e^{X_s} ( r - \frac12 \sigma^2 - \mu_1) + rk\big) \Big) ds \Big] \nonumber \\
&\leq \mathbb{E}^\mathbb{Q}_{(x, \exp( \frac{\gamma}{\sigma}(x+z)} \Big[  \int_0^\infty e^{-rs} \Big( e^{X_s} ( r - \frac12 \sigma^2 - \mu_0 ) + rk  \Big) ds \Big] \nonumber \\
&\hspace*{4cm} + (1 + e^{\frac{\gamma}{\sigma}(x+z)}) \mathbb{E}_{(x, \pi)} \Big[  \int_0^\infty e^{-rs} \Big(  e^{X_s} ( r - \frac12 \sigma^2 - \mu_1) + rk \Big) ds \Big] < \infty ,
\end{align}
where $\pi = e^{\frac{\gamma}{\sigma}(x+ z)}/(1 + e^{\frac{\gamma}{\sigma}(x+ z)})$ and the last inequality follows again from Assumption \ref{Assumption: Well-posedness}. Hence, given the finiteness of the expectation in \eqref{Expectation 2 for Proof}, we can apply dominated convergence theorem in order to interchange expectation and limits as $R \uparrow \infty$ and $T \uparrow \infty$. Combining this result with \eqref{lim lim Exp} gives \eqref{Probabilistic Repr of vhat}, which completes our proof. \qed
\end{proof}

\section*{Acknowledgements} The authors would like to two anonymous referees for their pertinent comments on an earlier version of this work. Moreover, the authors gratefully acknowledge financial support by the \textit{Deutsche Forschungsgemeinschaft}
(DFG, German Research Foundation) - SFB 1283/2 2021 - 317210226.

\end{document}